%% file: main.tex
\documentclass{lmcs}

\input{style}

\begin{document}

\title{Parameterized Quantum Circuit Semantics Through Enriched Categories}

\author{Neil J. Ross}[a]
\address{Dalhousie University, Halifax, NS, Canada}
\email{neil.jr.ross@dal.ca}

\author{Scott Wesley}[b]
\address{Dalhousie University, Halifax, NS, Canada}
\email{scott.wesley@dal.ca}

\thanks{Scott Wesley is supported by a Killam Predoctoral Scholarship Level 2.}

\begin{abstract}
    It is well-known that combinatorial circuits are modeled mathematically by string diagrams in monoidal categories.
    Given a gate set $\Sigma$, the circuits over $\Sigma$ can be thought of as string diagrams in the free monoidal category generated by $\Sigma$.
    In this model, circuit semantics are then given by monoidal functors out of this free category.
    For quantum circuits, this functor is often valued in the category of unitary matrices.
    This model suffices for concrete quantum circuits, but fails to describe parameterized families of quantum circuits, such as those which arise in the analysis of ansatz circuits.
    In this paper, we introduce an approach to parameterized circuit semantics, which is based on enriched category theory.
    We first introduce an abstract categorical construction, and use this to gain new insights on controlled operations and quantum communication.
    We then study the special cases of Cartesian monoidal parameters and monoidal closed parameters, both endowing the parameterized semantics with useful constructions.
    We conclude by showing that the monoidal closed case can be used to unify two perspectives on quantum control.
\end{abstract}

\maketitle

\input{sections/introduction}
\input{sections/background}
\input{sections/motivation}
\input{sections/construction}
\input{sections/cart}
\input{sections/closed}
\input{sections/conclusion}

\section*{Acknowledgments}

We would like to thank Aaron Fairbanks for his feedback on the technical contributions of this paper, and Peter Selinger for his feedback on the presentation of this work.

\bibliographystyle{alphaurl}
\bibliography{references}

\end{document}

%% file: style.tex
\usepackage{xcolor}
\usepackage{braket}
\usepackage{amssymb}
\usepackage{graphicx}
\usepackage[export]{adjustbox}
\usepackage{xspace}
\usepackage{mathabx}
\usepackage{subcaption}
\usepackage{bbm}

\usepackage{hyperref}
\usepackage[capitalise]{cleveref}
\crefname{section}{Sec.}{Secs.}
\crefname{defi}{Def.}{Defs.}
\crefname{exa}{Ex.}{Exs.}
\crefname{prop}{Prop.}{Props.}
\crefname{thm}{Thm.}{Thms.}
\crefname{cor}{Cor.}{Cors.}
\crefname{lem}{Lem.}{Lems.}

\usepackage{tikz}
\usetikzlibrary{quantikz2}

\usepackage{stmaryrd}
\DeclarePairedDelimiter\sem{\llbracket}{\rrbracket}

\newcommand{\ev}{\mathsf{ev}}
\newcommand{\app}{\mathsf{app}}

%% file: sections/introduction.tex
\section{Introduction}

A combinatorial circuit is any circuit that maps inputs to outputs without retaining internal state.
These circuits have many important application, such as in complexity theory~\cite{Vollmer1982}, hardware design~\cite{HarrisWeste2010}, and quantum program compilation~\cite{CrossJavadiAbhari2022}.
This wide range of applications has motivated the formal mathematical study of combinatorial circuits.
One early approach to this problem was to view combinatorial circuits as string diagrams in a monoidal category~\cite{Hotz1965}.
This approach has seen wide acceptance in the field of quantum program design and analysis~\cite{HeunenVicary2019}.
In this model, each gate is a morphism $G: n \to m$ with $n$ input wires and $m$ output wires.
If the number of outputs from a gate $G$ matches the number of inputs to a gate $H$, then $H \circ G$ denotes the circuit obtained by connecting the outputs of $G$ to the inputs of $H$.
Given any two gates $G$ and $H$, $G \otimes H$ denotes the circuit obtained by placing $G$ and $H$ side-by-side without any of their inputs connecting.
In other words, $( \circ )$ denotes sequential composition and $( \otimes )$ denotes parallel composition.
Given a set $\Sigma$ of gates, $F( \Sigma )$ denotes the family of circuits obtained from the gates in $\Sigma$.
These circuits form a \emph{free monoidal category}.
Their semantic interpretation is then given by a \emph{monoidal functor} $\sem{-}: F( \Sigma ) \to \mathcal{D}$ where $\mathcal{D}$ denotes the semantic domain.
It is easy to show that $\sem{-}$ is fully determined by where it sends the gates in $\Sigma$, as is expected of combinatorial circuits.
This perspective on combinatorial circuits has had many successes, such as providing decision procedures for the equality of certain classes of classical circuits~\cite{Lafont2003}, and optimization procedures for certain classes of quantum circuits~(see~\cite{Backens2021} for an introduction).

In most quantum circuit semantics, $\mathcal{D}$ is taken to be the category of linear transformations between some class of spaces.
However, with the advent of quantum machine learning, this approach is no longer sufficient.
As in classical machine learning, a quantum machine learning model is a tensor network (i.e., a circuit) with one or more free parameters from some continuous space~(see~\cite{HugginsPatil2019}).
To see why this is challenging, consider a simple quantum machine learning model with two parameters $\{ p_1, p_2 \}$ which can be used to instantiate Pauli-$Z$ rotations denoted by the gate $R_Z( - ): 1 \to 1$.
For example, $R_Z( p_1 )$ and $R_Z( p_1 + p_2 )$ are both instantiations of the gate $R_Z( - )$.
Following the standard semantics, $\sem{R_Z( p_1 )}$ and $\sem{R_Z( p_1 + p_2)}$ should correspond to the following continuous maps with domain $\mathbb{R}^2$.
\begin{align*}
    \sem{R_Z( p_1 )}( x, y ) &=
    \begin{bmatrix}
        e^{-ix/2} & 0 \\
        0 & e^{ix/2}
    \end{bmatrix}
    &
    \sem{R_Z( p_1 + p_2 )}( x, y ) &=
    \begin{bmatrix}
        e^{-i(x+y)/2} & 0 \\
        0 & e^{i(x+y)/2}
    \end{bmatrix}
\end{align*}
One would then expect that $R_Z( p_1 )$ and $R_Z( p_1 + p_2 )$ compose as follows.
\begin{equation*}
    \sem{R_Z( p_1 ) \star R_Z( p_1 + p_2 )}( x, y )
    =
    \begin{bmatrix}
        e^{-ix/2} & 0 \\
        0 & e^{ix/2}
    \end{bmatrix}
    \begin{bmatrix}
        e^{-i(x+y)/2} & 0 \\
        0 & e^{i(x+y)/2}
    \end{bmatrix}
\end{equation*}
It should be noted that this new composition is neither the composition of continuous functions, nor the composition of linear functions.
Instead, this composition is best understood as: (1) copying the parameters; (2) providing a copy of the parameters to each function; (3) composing the resulting linear maps.
This composition has appeared ad-hoc in many prior works~\cite{JeandelPerdrix2018,MillerBakewell2020,XuLi2022,PehamBurgholzer2023,HongHuang2024}.
In an extended abstract~\cite{Wesley2025}, we noted that vector spaces are enriched over topological spaces, and that this enrichment explains why the construction is well-defined.
However, the underlying properties of enriched categories which make this construction possible have yet to be studied.

One alternative approach to parameterized quantum circuits is that of linear dependent type theory.
In~\cite{FuKishida2022}, it was shown that set-theoretic parameters could be introduced to a quantum programming language without violating the linearity of the language.
Moreover, the parameters could be used to parameterize not only the family of morphisms, but also their domains and codomains.
However, these parameters were introduced freely, meaning that the set of parameters would have no internal structure.
This means that the semantic model could not enforce that every quantum circuit is a continuous function, as required by quantum machine learning.
Moreover, when the parameter spaces are taken to be topological spaces, then the topology of the spaces can be used to encode additional functional requirements through the type system.
For example, if a gate $G( p )$ is parameterized by the space $\mathbb{R} / (2\pi \mathbb{Z} )$, then $\sem{G( p )}( x )$ must have period of at most $2\pi$, which cannot be enforced through the types in a purely set-theoretic model of parameters.

More fundamentally, the model of linear dependent type theory given in~\cite{FuKishida2022} does not account for the new sequential composition found in the motivating example.
In principle, the parameters could be copied and passed along each time two parameterized gates are composed.
However, given a fixed set of named parameters, there should never exist a case in which the composition is performed without first copying the parameters.
For this reason, the present paper looks for a new notion of parameterization which encodes the natural choices for sequential and parallel composition.
Since the parameterized families of maps go from parameter objects in some category $\mathcal{V}$ to hom-sets in some category $\mathcal{C}$, then it is clear that $\mathcal{C}$ must be enriched over $\mathcal{V}$.
Given this enrichment, it can then be asked what is necessary for the new sequential and parallel composition to be well-defined.
These questions turn out to be interesting not only from the perspective of parameterized quantum circuits, but also from the perspective of enriched category theory.

This new notion of parameterization also provides new insights into traditionally well-understood constructions in quantum computing.
For example, controlled operations can be understood as an instance of this new notion of parameterization.
Intuitively, the control qubits can be thought of as linear parameters, which are used to multiplex between the operations performed on the target qubits.
Since the control qubits are not modified during this procedure, then in some sense they are never consumed, and will persist after the controlled operation is finished.
Alternatively, we can think of the control qubit being copied in some linear way, with only the copy being consumed during the operation.
This new point-of-view is not only insightful, but allows for many well-known circuit relations to be re-obtained in a more structured manner.
This construction is in line with the comonoidal view of quantum control found in~\cite{HeunenVicary2019}.

The rest of this paper is structured as follows:
\cref{Sect:Background} briefly reviews the background required to read this paper;
\cref{Sect:Motivation} provides motivation for the paper through the study of quantum circuits with parameterized rotations (e.g., ansatz circuits);
\cref{Sect:Construction} outlines a generalized construction for parameterized categories, which is shown to have nice categorical properties and immediate applications to quantum information;
\cref{Sect:CartMon} studies the special case of parameterization over Cartesian monoidal categories, in which all objects in $\mathcal{V}$ can be seen as classical parameters (e.g., they can be copied and discarded);
\cref{Sect:Closed} studies the special case of self-enrichment as a generalization of control functors.

%% file: sections/background.tex
\section{Background}
\label{Sect:Background}

We assume familiarity with ordinary category theory~\cite{MacLane2010}, enriched category theory~\cite{Kelly1982}, and string diagrams~\cite{Selinger_2010}.
In this paper, we adopt the following notation.
\begin{itemize}
\item \textbf{MonCat} is the $2$-category of monoidal categories, lax monoidal functors, and monoidal natural transformations.
\item \textbf{BrMonCat} is the $2$-category of braided monoidal categories, lax monoidal functors, and monoidal natural transformations.
\item For a monoidal category $\mathcal{V}$, $\mathbf{Comon}( \mathcal{V} )$ is the category of comonoids in $\mathcal{V}$, and $\mathbf{CComon}( \mathcal{V} )$ is the sub-category of cocommutative comonoids in $\mathcal{V}$.
\item If $\mathcal{V}$ is a monoidal closed category with $f \in \mathcal{V}( X, Y )$, then we write $[f]$ for the image of $f$ under the isomorphism $\mathcal{V}( X, Y ) \cong \mathcal{V}( \mathbb{I}, [ X, Y ] )$ and $\app_{X,Y}: [ X, Y ] \otimes X \to Y$ for the corresponding evaluation map.
\item If $\mathcal{V}$ is a monoidal category, then $\mathcal{V}\mathbf{Cat}$ is the $2$-category of $\mathcal{V}$-categories, $\mathcal{V}$-natural transformations, and $\mathcal{V}$-functors.
\end{itemize}
We briefly review the key concepts used in this paper.

\vspace{1em}
\noindent
\textbf{Categories of Comonoids}.
Let $( \mathcal{V}, \otimes, \mathbb{I}, \alpha, \lambda, \rho )$ be a monoidal category.
We recall the following results from~\cite{Porst01062008}.
If $\mathcal{V}$ is braided monoidal with braiding $\beta$, then $\mathbf{Comon}( \mathcal{V} )$ is a monoidal category with $( P, d, e ) \otimes ( Q, \delta, \epsilon ) := ( P \otimes Q, ( 1 \otimes \beta_{P,Q} \otimes 1 ) \circ d \otimes \delta, e \otimes \epsilon )$ and monoidal unit $( \mathbb{I}, \lambda_{\mathbb{I}}^{-1}, 1_{\mathbb{I}} )$.
The unitors and associator are inherited from $\mathcal{V}$.
Moreover, if $\mathcal{V}$ is symmetric monoidal, then $\mathbf{Comon}( \mathcal{V} )$ is also a symmetric monoidal category.

\vspace{1em}
\noindent
\textbf{Cartesian Monoidal Categories}.
A Cartesian monoidal category is a monoidal category $\mathcal{M} = ( \mathcal{V}, \otimes, \mathbb{I}, \alpha, \lambda, \rho )$ in which $\otimes$ is the categorical product and $\mathbb{I}$ is the terminal object.
It follows from~\cite{Fox1976} that if $\mathcal{M}$ is symmetric monoidal, then $\mathcal{M}$ is Cartesian monoidal if and only if $\mathcal{M}$ is isomorphic to $\mathbf{Comon}( \mathcal{V} )$ as a symmetric monoidal category~\cite{Arkor2025}.
Consequently, each object in $\mathcal{M}$ admits a unique comonoid structure.
Theorems of this form are referred to as \emph{Fox's Theorems}.
A similar characterization is given in~\cite{HeunenVicary2019}.
The category $\mathcal{M}$ is said to have \emph{uniform deleting} if there exists a natural transformation $e: \mathcal{V} \Rightarrow \mathbb{I}$ satisfying $e_{\mathbb{I}} = \lambda_{\mathbb{I}}$ and $e_{A \otimes B} = \lambda_{\mathbb{I}} \circ ( e_A \otimes e_B )$ for each $A \in \mathcal{V}_0$ and $B \in \mathcal{V}_0$.
The category $\mathcal{M}$ is said to have \emph{uniform copying} if $\mathcal{M}$ is symmetric monoidal and there exists a natural transformation $\Delta: \mathcal{V} \Rightarrow \mathcal{V} \otimes \mathcal{V}$ such that $\Delta_{\mathbb{I}} = \rho_{\mathbb{I}}$ and the following equation of diagrams holds for each $A \in \mathcal{V}_0$ and $B \in \mathcal{V}_0$ (of course, $\lambda_{\mathbb{I}} = \rho_{\mathbb{I}}$~\cite{Kelly1964}).
\begin{equation}
    \label{Eqn:MonoidProd}
    \Delta_{P \otimes Q}
    =
    \includegraphics[valign=c,scale=0.85]{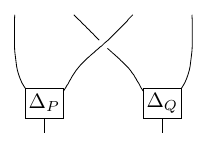}
\end{equation}

\begin{prop}[\cite{HeunenVicary2019}]
    \label{Prop:Cartesian}
    A monoidal category $\mathcal{M}$ is Cartesian monoidal if and only if $\mathcal{M}$ has uniform deleting $e: \mathcal{V} \to \mathbb{I}$ and uniform copying $\Delta: \mathcal{V} \to \mathcal{V} \otimes \mathcal{V}$ such that $( A, \Delta_A, e_A )$ is a comonoid object in $\mathcal{M}$ for each $A \in \mathcal{V}_0$.
\end{prop}

\vspace{1em}
\noindent
\textbf{Dagger Categories}.
A \emph{dagger structure}~\cite{Selinger2007} on a category $\mathcal{C}$ is a functor $\dagger: \mathcal{C} \to \mathcal{C}^{\text{op}}$ such that $X^\dagger = X$ for each $X \in \mathcal{C}_0$ and $( f^\dagger )^\dagger = f$ for each $f \in \mathcal{C}_1$.
If $( \dagger )$ is dagger structure on $\mathcal{C}$ and $( \ddagger )$ is a dagger structure on $\mathcal{D}$, then $F: \mathcal{C} \to \mathcal{D}$ is a dagger-functor when $F( f^\dagger ) = F( f )^\ddagger$ for all $f \in \mathcal{C}( X, Y )$.
A \emph{monoidal dagger structure}~\cite{Selinger2007,Selinger_2010} on a monoidal category $( \mathcal{V}, \otimes, \mathbb{I}, \alpha, \lambda, \rho )$ is a dagger structure $( \dagger )$ on $\mathcal{V}$ such that the following properties hold.
\begin{itemize}
\item If $f \in \mathcal{V}_1$ and $g \in \mathcal{V}_1$, then $( f \otimes g )^\dagger = f^\dagger \otimes g^\dagger$.
\item If $X \in \mathcal{C}_0$, $Y \in \mathcal{C}_0$, and $Z \in \mathcal{C}_0$, then $\alpha_{X,Y,Z}^\dagger = \alpha_{X,Y,Z}^{-1}$.
\item If $X \in \mathcal{C}_0$, then $\lambda_X^\dagger = \lambda_X^{-1}$ and $\rho_X^\dagger = \rho_X^{-1}$.
\end{itemize}
A \emph{vertical reflection} of a string diagram reflects the diagram vertically while replacing each morphism $f$ with its mirror image $f^\dagger$.
If $\dagger$ is a monoidal dagger structure, then vertical reflections respect the equality of string diagrams~\cite{Selinger_2010}.
In this case $( P, m ,e )$ is a monoid in $\mathcal{V}$ if and only if $( P, m^\dagger, e^\dagger )$ is a comonoid in $\mathcal{V}$.
In a monoidal closed category with a monoidal dagger structure, $\app_{Y,X} \circ ( [f^\dagger] \otimes 1_Y ) = ( [f]^\dagger \otimes 1_X ) \circ \app_{X,Y}^\dagger$.

\vspace{1em}
\noindent
\textbf{Enriched Categories}.
We write $\mathcal{C}^{\mathcal{V}}$ for a category enriched over $\mathcal{V}$ and $\mathcal{C}$ for the underlying category.
If $( F, \eta, \epsilon ): \mathcal{V} \to \mathcal{W}$ is a lax monoidal functor, then there exists a \emph{base change} functor $F_{*}: \mathcal{V}\mathbf{Cat} \to \mathcal{W}\mathbf{Cat}$~\cite{Kelly1982}.
If $\mathcal{C}^{\mathcal{V}}$ has enriched composition $M$ and enriched identity $1_X$, then $F_*( \mathcal{C}^{\mathcal{V}} )$ will have enriched composition $F( M ) \circ \eta$ and enriched identity $1_X \circ \epsilon$.
We will write $( \mathcal{C}^{\mathcal{V}}, \mathbb{I}, \otimes^{\mathcal{V}}, \alpha^{\mathcal{V}}, \lambda^{\mathcal{V}}, \rho^{\mathcal{V}} )$ for a $\mathcal{V}$-enriched monoidal category and $( \mathcal{C}, \mathbb{I}, \otimes, \alpha, \lambda, \rho )$ for the underlying monoidal category.
Since $\mathcal{V}$-enriched monoidal categories are pseudomonoids in $\mathcal{V}\mathbf{Cat}$, $F_*$ also sends $\mathcal{V}$-enriched monoidal categories to $\mathcal{W}$-enriched monoidal categories, provided $F$ is braided~\cite{Cruttwell2008}.
When $\mathcal{V}$ is symmetric monoidal, it is straight-forward to show that $F_*$ also preserves braided (resp. symmetric) pseudomonoids in $\mathcal{V}\mathbf{Cat}$, corresponding to $\mathcal{V}$-enriched braided (resp. symmetric) monoidal categories (see the supplementary material for more details).
When $\mathcal{V}$ is monoidal closed, it can be shown that $\mathcal{V}$ is the underlying category of some $\mathcal{V}$-enriched category whose hom-objects are the internal homs and whose composition is defined using the hom-tensor adjunction~\cite{Kelly1982}.
Moreover, if $\mathcal{V}$ is braided monoidal closed, then the monoidal product on $\mathcal{V}$ also translates to the $\mathcal{V}$-enriched setting (see~\cite{10.1093/imrn/rnx217} for a special case).
Throughout this paper, we will reason about enriched categories via string diagrams, as in~\cite{Wesley2024} and~\cite{10.1093/imrn/rnx217}.

\begin{figure}[t]
    \begin{subfigure}[b]{0.5\textwidth}
        \centering
        \begin{equation*}
            \includegraphics[valign=c,scale=0.85]{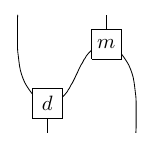}
            =
            \includegraphics[valign=c,scale=0.85]{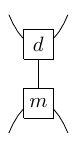}
            =
            \includegraphics[valign=c,scale=0.85]{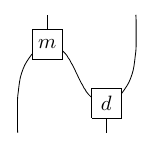}
        \end{equation*}
        \caption{The extended Frobenius law.}
        \label{Fig:Frob:Law}
    \end{subfigure}
    \begin{subfigure}[b]{0.22\textwidth}
        \centering
        \includegraphics[valign=c,scale=0.85]{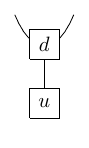}
        \caption{The unit $\eta$.}
        \label{Fig:Frob:Unit}
    \end{subfigure}
    \begin{subfigure}[b]{0.22\textwidth}
        \centering
        \includegraphics[valign=c,scale=0.85]{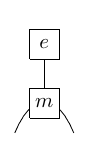}
        \caption{The counit $\epsilon$.}
        \label{Fig:Frob:Counit}
    \end{subfigure}
    \caption{Structure of a Frobenius algebra $A = ( P, m, d, u, e )$ with $P \vdash P$.}
        \label{Fig:Frob}
\end{figure}

\vspace{1em}
\noindent
\textbf{Controlled Operations}.
We write $\{ \ket{0}, \ket{1}, \ldots, \ket{n-1} \}$ for the standard basis in $\mathbb{C}^n$, $I_n$ for the $n \times n$ identity matrix, and $\bra{j}$ for the dual basis vector to $\ket{j}$.
Given a unitary transformation $U: \mathbb{C}^n \to \mathbb{C}^n$, we write $C( U ) = \ket{0} \bra{0} \otimes I_n + \ket{1} \bra{1} \otimes U$ for the \emph{standard controlled version of $U$}~\cite{NielsenChuang2010}.
This definition can be stated more generally in terms of special commutative $\dagger$-Frobenius structures, as outlined in~\cite{HeunenVicary2019}.
Recall that a \emph{Frobenius structure} in a braided monoidal category $\mathcal{V}$ is an object $P \in \mathcal{V}_0$ equipped with a monoid structure $( P, m, u )$ and a comonoid structure $( P, d, e )$ satisfying the equation in~\cref{Fig:Frob:Law}.
It follows that if $P$ admits a Frobenius structure, then $P$ is self-adjoint as in~\cref{Fig:Frob}.
Moreover, if $P$ is self-adjoint, then $( P, 1_P \otimes \eta \otimes 1_P, \epsilon )$ is a comonoid, referred to as the \emph{pair-of-pants} comonoid in~\cite{HeunenVicary2019}.
It follows from this adjunction that the comonoid $( P, d, e )$ is uniquely determined by the choice of monoid $( P, m, u )$ and the choice of counit $e$.

\begin{prop}[\cite{HeunenVicary2019}]
    \label{Prop:NonDegenForm}
    If $( P, m, u )$ is a monoid in $\mathcal{V}$, then there exists a bijective correspondence between the following structures.
    \begin{enumerate}
    \item Comonoids $( P, d, e )$ making $( P, m, d, u, e )$ a Frobenius algebra.
    \item Morphisms $e: P \to \mathbb{I}$ making $e \circ m$ the counit of an adjunction $P \vdash P$.
    \end{enumerate}
    Given such an $e: P \to \mathbb{I}$ and unit $\eta$, the comultiplication is defined as follows.
    \begin{equation*}
        \includegraphics[valign=c,scale=0.85]{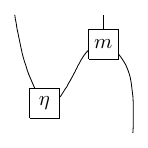}
        =
        d
        =
        \includegraphics[valign=c,scale=0.85]{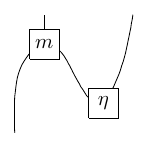}
    \end{equation*}
\end{prop}

If $( P, m, u )$ is commutative and $( P, d, e )$ is cocommutative, then $( P, m, d, u, e )$ is a commutative Frobenius algebra.
If $( \dagger )$ is a dagger-structure for $\mathcal{V}$ and $( P, d, e ) = ( P, m^\dagger, u^\dagger )$, then $A$ is $\dagger$-Frobenius algebra.
If $d \circ m = 1_P$, then $A$ is a special Frobenius algebra.
This yields a categorical characterization of orthogonal bases.

\begin{defi}[Copyable State~\cite{Coecke2012}]
    The \emph{copyable states} of a Frobenius algebra $( P, m, d, u, e )$ in $\mathcal{V}$ are the generalized elements $x \in \mathcal{V}( \mathbb{I}, P )$ such that $d \circ x = x \otimes x \circ \lambda_X^{-1}$.
\end{defi}

\begin{prop}[\cite{Coecke2012}]
    \label{Prop:BasisComonoid}
    Let $\mathcal{V} = \mathbb{C}\mathbf{FVect}$ be the category of finite-dimensional complex vector spaces, $\dagger: \mathcal{V} \to \mathcal{V}^\text{op}$ be the conjugate transpose functor, and $A = ( P, m, d, u, e )$ be a Frobenius algebra in $\mathcal{V}$.
    Then the following statements are equivalent.
    \begin{enumerate}
    \item $A$ is a special commutative $\dagger$-Frobenius algebra.
    \item The copyable states of $A$ form an orthonormal basis for $P$.
    \end{enumerate}
    Moreover, if $\mathcal{B} = \{ \ket{b_1}, \ket{b_2}, \ldots, \ket{b_n} \}$ is an orthonormal basis for $P$, then $\mathcal{B}$ consists of the copyable states for a special commutative $\dagger$-Frobenius algebra with comultiplication defined linearly by $d: \ket{b_j} \mapsto \ket{b_j} \otimes \ket{b_j}$ with $e = \sum_{j=1}^n \bra{b_j}$.
\end{prop}

As in linear algebra, it is possible to consider modules of monoid objects.
The free modules of a special commutative $\dagger$-Frobenius algebra $A$ correspond to measuring the monoid object with respect to the defining basis of $A$.
The module homomorphisms of these free modules then correspond to generalized controlled operations with respect to the defining basis.
In~\cite{HeunenVicary2019}, controlled operations are generalized further, to include module homomorphisms of all $\dagger$-modules over special commutative $\dagger$-Frobenius algebras.

\begin{defi}[Module]
    Let $M = ( P, m, u )$ be a monoid in $\mathcal{V}$.
    An $M$-module is an object $X \in \mathcal{V}_0$ with a choice of morphism $\mu \in \mathcal{V}( P \otimes X, X )$ satisfying the following equations.
    \begin{align*}
        \includegraphics[valign=c,scale=0.85]{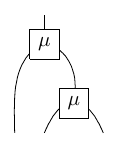}
        &=
        \includegraphics[valign=c,scale=0.85]{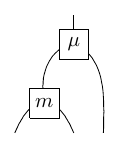}
        &
        \includegraphics[valign=c,scale=0.85]{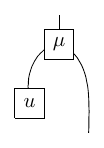}
        &=
        1_X
    \end{align*}
    If there exists some object $Y \in \mathcal{V}_0$ such that $X = P \otimes Y$ and $\mu = m \otimes 1_Y$, then $( X, \mu )$ is said to be a free $M$-module.
    The dual notion to an $M$-module is an $M$-comodule.
\end{defi}

\begin{defi}[Module Homomorphism]
    Let $M = ( P, m, u )$ be a monoid in $\mathcal{V}$ with $M$-modules $( X, \mu )$ and $( Y, \nu )$.
    An $M$-module homomorphism from $( X, \mu )$ to $( Y, \nu )$ is a morphism $f \in \mathcal{V}( X, Y )$ such that $f \circ \nu = \mu \circ ( 1_P \otimes f )$.
\end{defi}

More recently, control functors have been introduced, which generalize the endofunctor that maps each unitary $U$ to its controlled version $C( U )$~\cite{Delorme2026}.
In this paper we consider a generalization of control functors to braided monoidal categories.
Let $( \mathcal{C}, \otimes, \mathbb{I}, \beta, \alpha, \lambda, \rho )$ be a braided monoidal category with sub-category $\mathcal{C}_{\mathbf{Endo}}$ of endomorphisms.
A \emph{control functor} is an ordinary functor $C: \mathcal{C}_{\mathbf{Endo}} \to \mathcal{C}_{\mathbf{Endo}}$ equipped with an object $P \in \mathcal{C}_0$ such that $F_0( X ) = P \otimes X$ subject to the following axioms.
\begin{itemize}
\item \textbf{Tensor Axiom}.
      If $f \in \mathcal{C}( X, X )$ and $Z \in \mathcal{C}_0$, then $C( f \otimes 1_Z ) = C( f ) \otimes 1_Z$.
\item \textbf{Target Symmetry Axiom}.
      If $X, Y, Z, W \in \mathcal{V}_0$, $f \in \mathcal{C}( X \otimes Y \otimes Z \otimes W, X \otimes Y \otimes Z \otimes W )$ and $\sigma = 1_X \otimes \beta_{Y,Z} \otimes 1_W$, then $C( \sigma^{-1} \circ f \circ \sigma ) = ( 1_P \otimes \sigma^{-1} ) \circ C( f ) \circ ( 1_P \otimes \sigma )$.
\item \textbf{Control Symmetry Axiom}.
      If $f \in \mathcal{C}( X, X )$, then $( \beta_{P,P} \otimes 1_X ) \circ C( f ) = ( \beta_{P,P} \otimes 1_X ) \circ C( f )$.
\end{itemize}
A control functor is \emph{$( u, v )$-pointed} if there exists $u \in \mathcal{C}( \mathbb{I}, P )$ and $v \in \mathcal{C}( \mathbb{I}, P )$ which satisfy the \textbf{Activation Axiom}: $C( f ) \circ ( u \otimes 1_X ) = u \otimes 1_X$ and $C( f ) \circ ( v \otimes 1_X ) = v \otimes f$ for each $f \in \mathcal{C}( X, X )$.
A control functor is \emph{conguated} if it satisfies the \textbf{Conjugation Axiom}: $C( h \circ g \circ f ) = ( 1_P \otimes h ) \circ C( g ) \circ ( 1_P \otimes f )$ for all $f \in \mathcal{C}( X, Y )$,  $g \in \mathcal{C}( Y, Y )$, and $h \in \mathcal{C}( Y, X )$ satisfying $h \circ f = 1_X$.
A $\dagger$-control functor is a control functor that is also a dagger-endofunctor.

%% file: sections/motivation.tex
\section{Motivating Example: Quantum Machine Learning}
\label{Sect:Motivation}

To better illustrate our enriched approach to constructing parameterized quantum circuits, we will consider the problem of quantum machine learning.
We note that this discussion is not limited to quantum machine learning, and in fact can be generalized to a wider class of problems known as variational quantum algorithms~\cite{CerezoArrasmith2021}.
In quantum machine learning, neural networks are specified by tensor diagrams, which can then be converted into quantum circuits~\cite{HugginsPatil2019}.
Certain nodes in the neural network are labeled by weights (i.e., real numbers), which must be learned by a training algorithm.
When converting a neural network to a quantum circuit, each weighted node is replaced by a sub-circuit in which the weight appears as the angle provided to a rotation gate.
Then the corresponding circuit can be thought of as a family of ordinary quantum circuits, parameterized by a choice of weights for the neural network.

In prior work~\cite{JeandelPerdrix2018,XuLi2022,PehamBurgholzer2023,HongHuang2024}, these circuits have been modeled using square matrices whose entries are continuous functions.
Then sequential and parallel composition are interpreted as matrix multiplication and the Kronecker tensor product, respectively.
It is not hard to prove that this ad-hoc construction yields a symmetric monoidal category.
With even more care, it can be shown that this new category is enriched in \emph{nice} topological spaces.
However, these specialized proofs fail to identify the structures underlying this construction, and how they can be generalized to other semantic models.
For example, Theorem~4.2 in~\cite{PehamBurgholzer2023} considers the case where the matrix entries are analytic functions, which does not follow immediately from the case of continuous functions.

The goal of this section is to identify the underlying structures which make this example work, so that they can be generalized in~\cref{Sect:Construction}.
This section begins by identifying the properties that all such parameterized families of morphisms ought to satisfy.
It is then shown how to construct such parameterized families using the uniform copying and deletion morphisms in the category of topological spaces (as in~\cite{Wesley2025}).
This is used to motivate a more general construction which works for parameter objects from any monoidal category, using properties of base change functors.

\newcommand{\RealAmp}{\texttt{RA}\xspace}
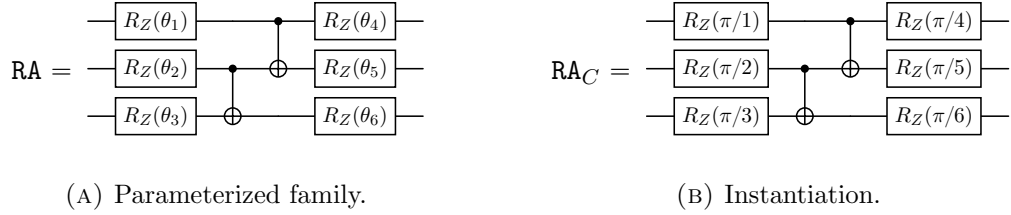
\begin{figure}[t]
    \begin{subfigure}[b]{0.48\textwidth}
        \centering
        \begin{equation*}
            \RealAmp = \input{circs/ra}
        \end{equation*}
        \caption{Parameterized family.}
        \label{Fig:Circuits:Family}
    \end{subfigure}
    \begin{subfigure}[b]{0.48\textwidth}
        \centering
        \begin{equation*}
            \RealAmp_C = \input{circs/ra_c}
        \end{equation*}
        \caption{Instantiation.}
        \label{Fig:Circuits:Instance}
    \end{subfigure}
    \caption{A two-layer real amplitude ansatz circuit acting on three qubits.}
    \label{Fig:Circuits}
\end{figure}

To understand the properties that parameterized families of circuits ought to satisfy, it will help to consider a concrete example.
In quantum machine learning, weighted convolution layers are often replaced by instantiations of the \emph{real amplitude ansatz circuit}~\cite{AbbasSutter2021,ArthurDate2022,DekelFrankel2023,KannoNakamura2024,YoffeEntin2024}.
A two-layer instantiation acting on three qubits is illustrated by \RealAmp in~\cref{Fig:Circuits:Family}.
Let $\Sigma_1$ denotes the PROP signature for Clifford circuits, $\Sigma_2$ denote the PROP signature for Clifford+$R_Z$ circuits, $\mathcal{C}$ denote the PROP category of unitary maps, $\mathcal{D}$ denote the new category whose morphisms are parameterized families of unitary maps, and the functor $\sem{-}_1: F( \Sigma_1 ) \to \mathcal{C}$ denote the standard semantic interpretation for $F( \Sigma_1 )$.
To interpret \RealAmp, the functor $\sem{-}_1: F( \Sigma_1 ) \to \mathcal{C}$ must be extended to a new functor $\sem{-}_2: F( \Sigma_2 ) \to \mathcal{D}$.
Since $F( \Sigma_1 )$ is a free PROP category, then $\mathcal{D}$ ought to be a symmetric monoidal category and $\sem{-}_2$ ought to be a braided monoidal functor, so that no semantic information is lost when interpreting \RealAmp.
However, simply requiring that $\sem{-}_2$ extend $\sem{-}_1$ as a braided monoidal functor fails to capture many important properties of parameterized quantum circuits.
For example, assume that the values $( \pi, \pi/2, \pi/3, \pi/4, \pi/5, \pi/6 )$ have been learned for the parameters $( \theta_1, \theta_2, \theta_3, \theta_4, \theta_5, \theta_6 )$.
Then the concrete convolution layer described by $\RealAmp$ would be the unitary matrix $\sem{\RealAmp}_2( \pi, \pi/2, \pi/3, \pi/4, \pi/5, \pi/6 )$.
This should correspond to the parameter-free circuit $\RealAmp_C$ found in~\cref{Fig:Circuits:Instance} obtained by substituting $( \pi, \pi/2, \pi/3, \pi/4, \pi/5, \pi/6 )$ for $( \theta_1, \theta_2, \theta_3, \theta_4, \theta_5, \theta_6 )$ in the circuit $\RealAmp$.

Moreover, the circuit $\RealAmp_C$ should be well-defined, since only this concrete circuit can be compiled to low-lever hardware instructions.
The obvious interpretation for each $R_Z( \pi/k )$ in $\RealAmp_C$ is the unitary matrix $\sem{R_Z( \theta_k )}_2( \pi, \pi/2, \pi/3, \pi/4, \pi/5, \pi/6 )$.
One would hope that $\sem{\RealAmp_C}_1 = \sem{\RealAmp}_2( \pi, \pi/2, \pi/3, \pi/4, \pi/5, \pi/6 )$.
This is only true if $\sem{G}_1 = \sem{G}_2( \theta )$ for all $G \in \Sigma_1$ and $\theta \in \mathbb{R}^6$.
Moreover, this substitution should respect sequential and parallel composition.
In other words, evaluation at $\theta \in \mathbb{R}^6$ should correspond to some braided monoidal functor $\ev_\theta: \mathcal{D} \to \mathcal{C}$.
Putting these requirements together, $\mathcal{D}$ should be a symmetric monoidal category which contains $\mathcal{C}$ as a subcategory, such that evaluating the parameterized morphisms in $\mathcal{D}$ at any parameter in $\mathbb{R}^6$ is a retraction onto $\mathcal{C}$.
More abstractly, the following diagram should commute, where the first square is obviously true, the second square captures the inclusion of $\mathcal{C}$ into $\mathcal{D}$, and the final triangle captures the retraction of $\mathcal{C}$ back onto itself through evaluation at $\theta \in \mathbb{R}^6$.
\begin{center}
    \includegraphics[scale=0.96]{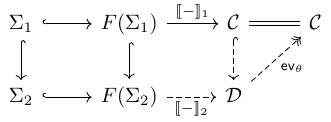}
\end{center}

First, the category $\mathcal{D}$ must be defined.
Since $\mathcal{D}$ is simply the \emph{parameterized} version of $\mathcal{C}$, then it makes sense that they have the same objects.
This means that $\mathcal{D}_0 = \mathcal{C}_0$.
Now, consider an arbitrary morphism $f \in \mathcal{D}( n, m )$ in the parameterized category.
As described above, $f$ should be a continuous map of the form $f: \mathbb{R}^6 \to \mathcal{C}( n, m )$ where $\mathbb{R}^6$ is viewed as a topological space under the product topology.
For this to make sense, $\mathcal{C}( n, m )$ must also be a topological space.
This means that $\mathcal{C}$ must be enriched in topological spaces, and consequently $f \in \mathbf{Top}( P, \mathcal{C}^{\mathbf{Top}}( n, m ) )$ where $P = \mathbb{R}^6$.
Since finite-dimensional vector spaces are enriched over topological spaces, then this is well-defined.

For $\mathcal{D}$ to be a category, the hom-sets must admit an associative unital composition operation $( - ) \star ( - )$.
Given two maps $f \in \mathcal{D}( n, m )$ and $g \in \mathcal{D}( m, t )$, the Cartesian product on \textbf{Top} can be used to obtain a new morphism $g \times f \in \mathbf{Top}( P \times P, \mathcal{C}^{\mathbf{Top}}( m, t ) \times \mathcal{C}^{\mathbf{Top}}( n, m ) )$.
Then using the enriched composition map, $M_{n,m,t} \in \mathbf{Top}( \mathcal{C}^{\mathbf{Top}}( m, t ) \times \mathcal{C}^{\mathbf{Top}}( n, m ), \mathcal{C}^{\mathbf{Top}}( n, t ) )$ it is possible to obtain a morphism in $\mathbf{Top}( P \times P, \mathcal{C}^{\mathbf{Top}}( n, t ) )$.
This morphism corresponds to the following string diagram in $\mathbf{Top}$.
\begin{equation*}
    M \circ ( g \times f ) \circ ( - )
    =
    \includegraphics[valign=c,scale=0.84]{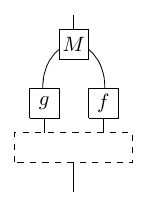}
\end{equation*}
Then to define $( - ) \star ( - )$, it suffices to find a morphism $d \in \mathbf{Top}( P, P \times P )$ which fills in the hole and makes the corresponding operation both associative and unital.
For the operation to be associative, the following equation must hold.
\begin{equation*}
    \includegraphics[valign=c,scale=0.84]{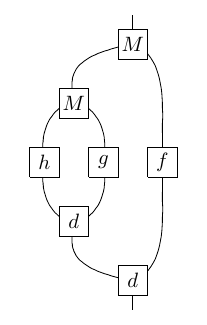}
    =
    ( h \star g ) \star f
    =
    h \star ( g \star f )
    =
    \includegraphics[valign=c,scale=0.84]{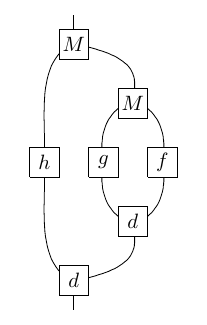}
\end{equation*}
It is suggested by these diagrams that $M$ should be associative and $d$ should be coassociative.
Since $M$ is already associative by definition, then it suffices to assume that $d$ is coassociative.
In this case, the following derivation holds.
\begin{center}
    \includegraphics[valign=c,scale=0.84]{figs/motivation/comp_assoc_lhs.pdf}
    =
    \includegraphics[valign=c,scale=0.84]{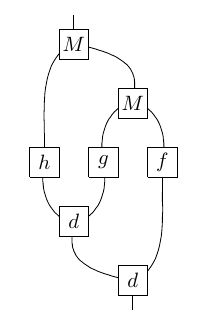}
    =
    \includegraphics[valign=c,scale=0.84]{figs/motivation/comp_assoc_rhs.pdf}
\end{center}
It remains to be shown that $( - ) \star ( - )$ is unital.
The first step is to find an identity morphism for each object $X \in \mathcal{C}_0$.
Since $\mathcal{C}$ is \textbf{Top}-enriched, then there exists a morphism $1_X \in \mathbf{Top}( \mathbb{I}, \mathcal{C}^{\mathbf{Top}}( X, X ) )$ such that $M \circ ( 1_X \times ( - ) )$ and $M \circ ( ( - ) \times 1_X )$ are the identity.
Ideally, this identity would lift to a morphism in $\mathcal{D}$ by composition with some morphism $e \in \mathbf{Top}( P, \mathbb{I} )$.
If $i_X := 1_X \circ e$ was the identity morphism for $X \in \mathcal{D}_0$, then $i_Y \star f = f$ and $f \star i_X$ for each $f \in \mathcal{D}( X, Y )$.
Similar to the case of associativity, it can be shown that if $e$ is the counit to $( P, d )$, then both equations hold.
In summary, if $\mathcal{C}$ is $\mathcal{D}$-enriched and $( P, d, e )$ is a comonoid in $\mathcal{C}$, then $\mathcal{D}$ is a category with composition given by $( - ) \star ( - )$.

It is not hard to check that $( - ) \star ( - )$ agrees with the parameterized matrix composition described at the beginning of this section.
For example, consider the following two parameterized morphisms.
\begin{align*}
    f( \theta ) &= \cos( \theta_1 ) I + i \sin( \theta_1 ) Z = \exp( i\theta_1 Z )
    &
    g( \theta ) &= \cos( \theta_2 ) I + i \sin( \theta_2 ) Z = \exp( i\theta_2 Z )
\end{align*}
Then $f$ and $g$ are $Z$-rotations by angles $\theta_1$ and $\theta_2$, respectively.
Then for each choice of $\theta \in \mathbb{R}^6$, $M( g( \theta ), f( \theta ) ) = M( \exp( i\theta_2 Z ), \exp( i\theta_1 Z ) ) = \exp( i (\theta_1 + \theta_2) Z )$.
As expected, this is also the result of $g \star f$.
\begin{equation*}
    ( g \star f )( \theta )
    =
    ( M \circ ( g \times f ) \circ d )( \theta )
    =
    M \circ ( g \times f )( \theta, \theta )
    =
    M( g( \theta ), f( \theta ) )
    =
    \exp( i ( \theta_1 + \theta_2 ) Z )
\end{equation*}
Moreover, this construction exposes the underlying structures which give rise to $\mathcal{D}$.
Since $( P, d, e )$ is coassociative and counital, then $( P, d, e )$ is a comonoid object in \textbf{Top}.
Since \textbf{Top} is Cartesian monoidal, then $( P, d, e )$ is the unique comonoid on $P$ defined by its copying and deletion morphisms.
This means that pre-composition by $d$ is the same as copying the input parameter, and post-composition by $M$ is the same as composing the morphisms when fixed at the given parameter.

Since $\mathcal{C}$ is also \textbf{Top}-enriched as a symmetric monoidal category, then the monoidal product $\otimes$ on $\mathcal{C}$ lifts to a monoidal product $\boxtimes$ on $\mathcal{D}$ in an analogous way.
In particular, $f \boxtimes g := \otimes \circ ( f \times g ) \circ d$.
If $( \alpha, \lambda, \rho, \beta )$ denotes the associator, unitors, and braiding for $\otimes$, then $( \alpha \circ e, \lambda \circ e, \rho \circ e, \beta \circ e )$ provide a canonical choice of associator, unitors, and braiding for $\boxtimes$.
One can show through tedious calculations (see~\cite{Wesley2025}) that $\boxtimes$ is a bifunctor with $( \alpha \circ e, \lambda \circ e, \rho \circ e, \beta \circ e )$ all natural transformations satisfying the coherence conditions for a symmetric monoidal category.
However, these calculations are less than ideal, since they disregard the enriched structure on $\mathcal{C}$.
A better approach is to show that $\mathcal{D}$ is obtained from $\mathcal{C}$ through a base change which preserves the symmetric monoidal structure on $\mathcal{C}$.

\begin{figure}[t]
    \begin{subfigure}[b]{0.13\textwidth}
        \centering
        \includegraphics[valign=c,scale=0.84]{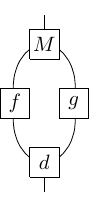}
        \caption{$g \star f$}
        \label{Fig:SetFunctors:Comp}
    \end{subfigure}
    \begin{subfigure}[b]{0.13\textwidth}
        \centering
        \includegraphics[valign=c,scale=0.84]{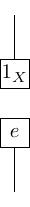}
        \caption{$i_X$}
        \label{Fig:SetFunctors:Id}
    \end{subfigure}
    \begin{subfigure}[b]{0.13\textwidth}
        \centering
        \includegraphics[valign=c,scale=0.84]{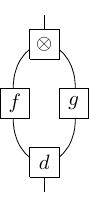}
        \caption{$f \boxtimes g$}
        \label{Fig:SetFunctors:PComp}
    \end{subfigure}
    \begin{subfigure}[b]{0.13\textwidth}
        \centering
        \includegraphics[valign=c,scale=0.84]{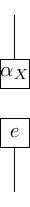}
        \caption{$a_X$}
        \label{Fig:SetFunctors:Assoc}
    \end{subfigure}
    \begin{subfigure}[b]{0.13\textwidth}
        \centering
        \includegraphics[valign=c,scale=0.84]{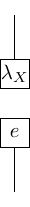}
        \caption{$l_X$}
        \label{Fig:SetFunctors:LUnitor}
    \end{subfigure}
    \begin{subfigure}[b]{0.13\textwidth}
        \centering
        \includegraphics[valign=c,scale=0.84]{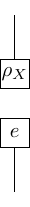}
        \caption{$r_X$}
        \label{Fig:SetFunctors:RUnitor}
    \end{subfigure}
    \begin{subfigure}[b]{0.13\textwidth}
        \centering
        \includegraphics[valign=c,scale=0.84]{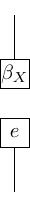}
        \caption{$b_X$}
        \label{Fig:SetFunctors:Braid}
    \end{subfigure}
    \caption{The category $( \mathcal{D}, \star, i )$ with monoidal structure $( \mathcal{D}, \boxtimes, \mathbb{I}, a, l, r )$ and braiding $\beta$.}
    \label{Fig:SetFunctors}
\end{figure}

Since each hom-object $\mathcal{C}( X, Y )$ in $\mathcal{C}$ is mapped to a new hom-object $\mathbf{Top}( P, \mathcal{C}^{\mathbf{Top}}( X, Y ) )$ in $\mathcal{D}$, then the base change functor should be $F( - ) = \mathbf{Top}( P, - )$.
Each choice of braided lax-monoidal structure on $F( - )$ determines a new base change from $\mathbf{TopCat}$ to $\mathbf{Cat}$.
This corresponds to picking a natural transformation $\eta: F( - ) \times_{\mathbf{Set}} F( - ) \to F( ( - ) \times ( - ) )$ and morphism $\epsilon: \{ \star \} \to F( \mathbb{I} )$ such that the coherence conditions hold.
Following through the base change construction, we find that the new composition in $F_*( \mathcal{C} )$ is $M \circ \eta$ and the new identities in $F_*( \mathcal{C} )$ are $1_X \circ \epsilon$.
Moreover, taking the product of topological spaces and then pre-composing with $d$ induces a natural transformation from $F( - ) \times_{\mathbf{Set}} F( - ) = \mathbf{Top}( P, - ) \times_{\mathbf{Set}} \mathbf{Top}( P, - )$ to $F( ( - ) \times ( - ) ) = \mathbf{Top}( P, ( - ) \times ( - ) )$.
This means that $\eta = \mathbf{Top}( d, - ) \circ ( \times )$ and $\epsilon = e$ work.
These morphisms are subject to the four coherence conditions for a braided lax-monoidal functor.
Since $\times_{\mathbf{Set}}$ is the Cartesian product on \textbf{Set} and $\{ \star \}$ is a separator for \textbf{Set}, then it suffices to interpret the coherence conditions as string diagrams in \textbf{Top}, as follows.
\begin{align*}
    \includegraphics[valign=c,scale=0.9]{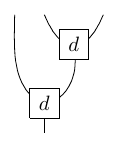}
    &=
    \includegraphics[valign=c,scale=0.9]{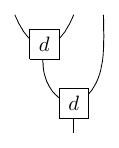}
    &
    \includegraphics[valign=c,scale=0.9]{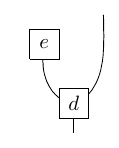}
    &=
    1_P
    =
    \includegraphics[valign=c,scale=0.9]{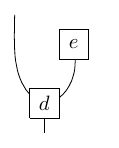}
    &
    \includegraphics[valign=c,scale=0.9]{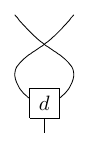}
    &=
    \includegraphics[valign=c,scale=0.9]{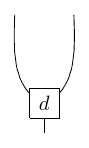}
\end{align*}
Since $( P, d, e )$ is a cocommutative comonoid, then these equations necessarily hold.
In conclusion, $\mathcal{D} = F_*( \mathcal{C} )$ admits the monoidal structure given in~\cref{Fig:SetFunctors}.

Given this definition of $\mathcal{D}$, it is now possible to interpret the parameterized rotations and the identity morphisms in $\RealAmp$.
To evaluate the controlled \texttt{not} gates, the inclusion of $\mathcal{C}$ into $\mathcal{D}$ must also be defined.
Since $1_X$ maps to $i_X$ via pre-composition by $e$, then it is reasonable to expect that all of $\mathcal{C}$ includes into $\mathcal{D}$ via pre-composition by $e$.
The section $j: \mathcal{C} \hookrightarrow \mathcal{D}$ induced by pre-composition with $e$ should be accompanied by a retraction $\ev_\theta$ for each choice of parameter $\theta \in \mathbb{R}^6$.
Since the parameters in $P$ are in correspondence with the morphisms in $\textbf{Top}( \mathbb{I}, P )$, then $\theta$ can be thought of as a morphism $\theta: \mathbb{I} \to P$.
Under this correspondence, evaluation of $g \in \mathcal{D}( X, Y )$ at $\theta$ should yield the morphism $g \circ \theta$.
In other words, evaluation at $\theta$ corresponds to pre-composition by $\theta$.
Since $\mathbb{I}$ is terminal in \textbf{Top}, then for each $f \in \mathcal{C}( X, Y )$, the evaluation of the inclusion of $f$ is $\ev_\theta( j( f ) ) = ( f \circ e ) \circ \theta = f \circ ( e \circ \theta ) = f \circ 1_{\mathbb{I}} = f$.
It remains to be shown that $j$ and $\ev_\theta$ are functorial.
It suffices to note that the underlying category $\mathcal{C}$ is also obtained via the hom-functor $\mathbf{Top}( \mathbb{I}, - )$ and the canonical comonoid on $\mathbb{I}$.
Since copying is natural in \textbf{Top}, then pre-composition by both $e$ and $\theta$ must induce braided lax-monoidal transformations between the two hom-functors.
It follows by enriched base change that both $j: \mathcal{C} \to \mathcal{D}$ and $\ev_\theta: \mathcal{D} \to \mathcal{C}$ are braided monoidal functors.
Consequently $\sem{\texttt{cnot}}_2 := j( \sem{\texttt{cnot}}_1 )$ and the following equation holds.
\begin{equation*}
    \sem{\RealAmp}_2
    =
    \left( \bigboxtimes_{j=4}^6 \exp( i \theta_j Z ) \right)
    \star
    \left( j\left( \sem{\texttt{cnot}}_1  \right) \boxtimes i_1 \right)
    \star
    \left( i_1 \boxtimes j\left( \sem{\texttt{cnot}}_1 \right) \right)
    \star
    \left( \bigboxtimes_{j=1}^3 \exp( i \theta_j Z ) \right)
\end{equation*}

Of course, there are many other properties we may wish to ask of $\mathcal{D}$.
For example, the functions illustrated from $\mathcal{D}$ have all been continuous, so it is natural to ask whether $\mathcal{D}$ admits non-trivial enrichment in topological spaces.
It turns out that this is not necessarily true since \textbf{Top} contains pathological counterexamples, but in the case of nice topological spaces this is true.
This leads to a more general construction of $\mathcal{D}$, which works for enrichment over any monoidal category $\mathcal{V}$.
However, the nice properties discussed in this section will still require that $\mathcal{V}$ is Cartesian closed.

%% file: circs/ra.tex
\scalebox{0.73}{\begin{quantikz}[row sep=0.5em]
    & \gate{R_Z( \theta_1 )} &          & \ctrl{1} & \gate{R_Z( \theta_4 )} & \\
    & \gate{R_Z( \theta_2 )} & \ctrl{1} & \targ{}  & \gate{R_Z( \theta_5 )} & \\
    & \gate{R_Z( \theta_3 )} & \targ{}  &          & \gate{R_Z( \theta_6 )} & \\
\end{quantikz}}

%% file: circs/ra_c.tex
\scalebox{0.73}{\begin{quantikz}[row sep=0.5em]
    & \gate{R_Z( \pi/1 )} &          & \ctrl{1} & \gate{R_Z( \pi/4 )} & \\
    & \gate{R_Z( \pi/2 )} & \ctrl{1} & \targ{}  & \gate{R_Z( \pi/5 )} & \\
    & \gate{R_Z( \pi/3 )} & \targ{}  &          & \gate{R_Z( \pi/6 )} & \\
\end{quantikz}}

%% file: sections/construction.tex
\section{The General Construction}
\label{Sect:Construction}

In the previous section, it was shown that parameterized families of ansatz circuits can be realized through comonoids on parameter objects and base change functors.
The construction followed the methodology in~\cite{Wesley2025}, which requires that $\mathcal{V}$ is Cartesian monoidal and always produces a locally small category.
This section provides a more general construction which no longer requires that $\mathcal{V}$ be Cartesian monoidal, and can produce categories which are enriched over some $\mathcal{W}$.
To this end, the notion of a \emph{parameterization functor} $\Omega_{(\mathcal{V},\mathcal{W})}$ is introduced in \cref{Def:Params}, which characterizes the base change functors of interest as those in the image of $\Omega_{(\mathcal{V},\mathcal{W})}$.
By allowing for $\mathcal{V}$ to be a $\mathcal{W}$-enriched monoidal category, the resulting base change functors can yield $\mathcal{W}$-enriched categories, as opposed to only locally small categories.
Note that in the motivating example, $\mathcal{V}$ was locally small so $\mathcal{W} = \mathbf{Set}$.
In \cref{Thm:Reparam}, it is shown that this assignment is a faithful functor.
Under the assumption that the monoidal unit in $\mathcal{W}$ is a separator, such as when $\mathcal{W} = \mathbf{Set}$, the functor is in fact fully faithful.
This means that each natural transformation between the base change functors of interest must arise from pre-composition by a unique comonoid homomorphism.
It will be shown later in this section that each natural transformation can be thought of as reparameterization with respect to its unique comonoid homomorphism.
To extend this construction to $\mathcal{V}$-enriched braided monoidal categories, it is then shown that $\Omega_{(\mathcal{V},\mathcal{W})}$ restricts to a full and faithful functor from cocommutative comonoids to lax braided monoidal functors.
This extends immediately to $\mathcal{V}$-enriched symmetric monoidal categories, since being symmetric is a property.
To verify that these base change functors truly generalize the motivating example, it is then shown in~\cref{Cor:UnderlyingShape} that the underlying category always exhibits the structures illustrated in \cref{Fig:SetFunctors}.

\begin{nota}
    In this section, $\mathcal{W} \in \mathbf{BrMonCat}_0$ with composition $( \bullet )$, identity $u$, and monoidal unit $\mathbb{J}$.
    Then $\mathcal{V}^{\mathcal{W}} \in \mathcal{W}\mathbf{Cat}_0$ with composition $\circ^{\mathcal{W}}$, identity $1^{\mathcal{W}}$, and monoidal data $( \mathcal{V}^{\mathcal{W}}, \mathbb{I}, \odot^{\mathcal{W}}, a^{\mathcal{W}}, \ell^{\mathcal{W}}, r^{\mathcal{W}} )$.
    Then $\mathcal{C}^{\mathcal{V}} \in \mathcal{V}\mathbf{Cat}_0$ with composition $M^{\mathcal{V}}$ and identity $I^{\mathcal{V}}$.
\end{nota}

\begin{defi}
    \label{Def:Params}
    The \emph{$\mathcal{W}$-parameterizations} for $\mathcal{V}^{\mathcal{W}}$ are determined by the assignment $\Omega_{(\mathcal{V}, \mathcal{W})}: \mathbf{Comon}^{\text{op}}( \mathcal{V} ) \to \mathbf{MonCat}( \mathcal{V}, \mathcal{W} )$ such that $\Omega_{(\mathcal{V}, \mathcal{W})}: ( P, d, e ) \mapsto ( \mathcal{V}^{\mathcal{W}}( P, - ), \eta, e )$ where $\eta$ denotes the following family of diagrams in $\mathcal{W}$ and $\Omega_{(\mathcal{V}, \mathcal{W})}: f \mapsto \mathcal{V}^{\mathcal{W}}( f, - )$.
    \begin{equation}
        \label{Eqn:ParamEta}
        \eta
        :=
        \includegraphics[valign=c,scale=0.84]{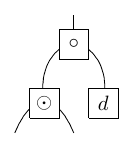}
    \end{equation}
\end{defi}

\begin{lem}
    \label{Lem:ParamEqs}
    If $P \in \mathcal{V}_0$, $d \in \mathcal{V}( P, P \odot P )$, $e \in \mathcal{V}( P, \mathbb{I} )$, $F = \mathcal{V}^{\mathcal{W}}( P, - )$, and $\eta$ is as defined in \cref{Eqn:ParamEta}, then the following equations hold.

    \noindent
    \begin{minipage}{.5\linewidth}
    \begin{equation}
        \label{Eqn:CompositorLemma1}
        \includegraphics[valign=c,scale=0.84]{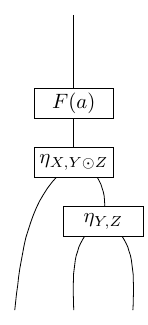}
        =
        \includegraphics[valign=c,scale=0.84]{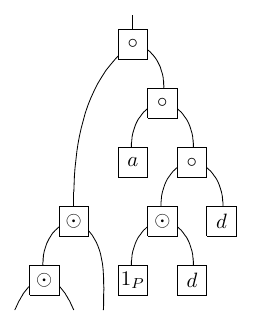}
    \end{equation}
    \end{minipage}%
    \hfill
    \begin{minipage}{.44\linewidth}
    \begin{equation}
        \label{Eqn:UnitorLemma1}
        \includegraphics[valign=c,scale=0.84]{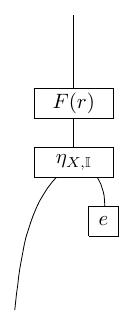}
        =
        \includegraphics[valign=c,scale=0.84]{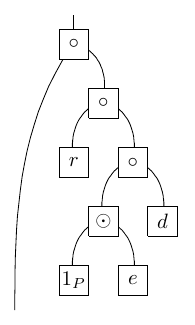}
    \end{equation}
    \end{minipage}

    \noindent
    \begin{minipage}{.5\linewidth}
    \begin{equation}
        \label{Eqn:CompositorLemma2}
        \includegraphics[valign=c,scale=0.84]{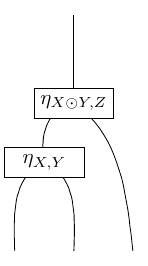}
        =
        \includegraphics[valign=c,scale=0.84]{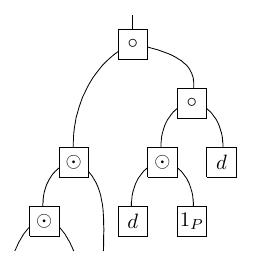}
    \end{equation}
    \end{minipage}%
    \hfill
    \begin{minipage}{.44\linewidth}
    \begin{equation}
        \label{Eqn:UnitorLemma2}
        \includegraphics[valign=c,scale=0.84]{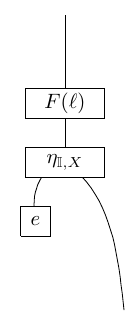}
        =
        \includegraphics[valign=c,scale=0.84]{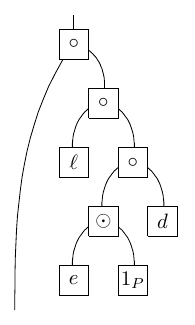}
    \end{equation}
    \end{minipage}

    \noindent
    Further assume that $Q \in \mathcal{V}_0$, $f \in \mathcal{V}( Q, P )$, $\delta \in \mathcal{V}( Q, Q \odot Q )$, and $\mu$ is as defined in \cref{Eqn:ParamEta}.
    
    \noindent
    \begin{minipage}{.56\linewidth}
    \begin{equation}
        \label{Eqn:NatTxLemma1}
        \includegraphics[valign=c,scale=0.84]{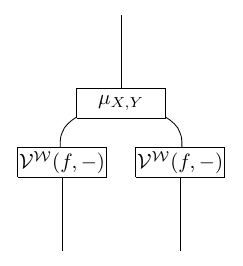}
        =
        \includegraphics[valign=c,scale=0.84]{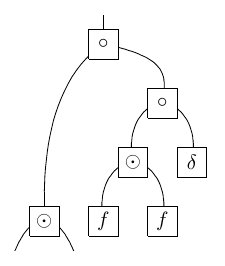}
    \end{equation}
    \end{minipage}%
    \hfill
    \begin{minipage}{.42\linewidth}
    \begin{equation}
        \label{Eqn:NatTxLemma2}
        \includegraphics[valign=c,scale=0.84]{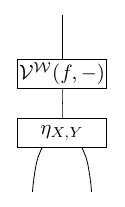}
        =
        \includegraphics[valign=c,scale=0.84]{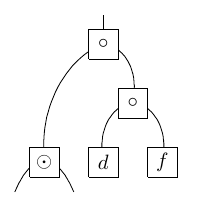}
    \end{equation}
    \end{minipage}

    \noindent
    If $\mathcal{V}^{\mathcal{W}}$ admits a $\mathcal{W}$-enriched braiding $b$, then the following equation also holds.
    \begin{align}
        \label{Eqn:BraidLemma}
        \includegraphics[valign=c,scale=0.84]{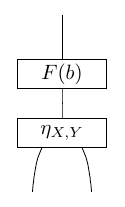}
        &=
        \includegraphics[valign=c,scale=0.84]{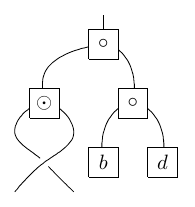}
    \end{align}
\end{lem}

\begin{proof}
    In the case of \textbf{Set}-enriched categories, the proof is trivial.
    For $\mathcal{V}$-enriched categories, the proof follows by lifting the set-theoretic reasoning to the diagrammatic setting.
    All derivations proceed in the obvious way, and can be found in the supplementary material.
\end{proof}

\begin{thm}
    \label{Thm:Reparam}
    The assignment $\Omega_{(\mathcal{V}, \mathcal{W})}$ is a faithful functor.
    Moreover, if $\mathcal{V}^{\mathcal{W}}$ admits a $\mathcal{W}$-enriched braiding, then $\Omega_{(\mathcal{V}, \mathcal{W})}$ restricts to a faithful functor $\mathbf{CComon}^{\text{op}}( \mathcal{V} ) \to \mathbf{BrCat}( \mathcal{V}, \mathcal{W} )$.
\end{thm}

\begin{proof}
    It must first be shown that this assignment is well-typed.
    \begin{itemize}
    \item \textbf{Objects}.
          Let $M = ( P, d, e ) \in \mathbf{Comon}( \mathcal{V} )_0$.
          Since enriched hom-functors are covariant in the second coordinate, then $\mathcal{V}^{\mathcal{W}}( P, - )$ is a functor from $\mathcal{V}$ to $\mathcal{W}$.
          It remains to be shown that $\Omega_{(\mathcal{V},\mathcal{W})}( M ) = ( \mathcal{V}^{\mathcal{W}}( P, - ), \eta, e )$ is a lax monoidal functor.
          Since $( P, d, e )$ is a comonoid, then the following equation holds by the coassociativity of $d$, together with \cref{Eqn:CompositorLemma1} and \cref{Eqn:CompositorLemma2} of \cref{Lem:ParamEqs}.
          \begin{equation*}
              \includegraphics[valign=c,scale=0.84]{figs/construction/compositor_lemma_1_lhs.pdf}
              =
              \includegraphics[valign=c,scale=0.84]{figs/construction/compositor_lemma_1_rhs.pdf}
              =
              \includegraphics[valign=c,scale=0.84]{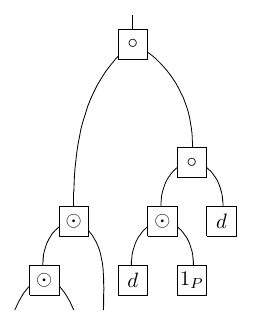}
              =
              \includegraphics[valign=c,,scale=0.84]{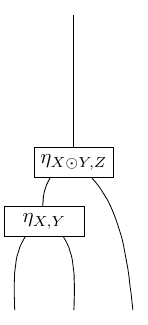}
          \end{equation*}
          Likewise, the following equation holds by the right counitality of $d$, together with \cref{Eqn:UnitorLemma1} of \cref{Lem:ParamEqs}.
          \begin{equation*}
              \includegraphics[valign=c,scale=0.84]{figs/construction/runitor_lemma_lhs.pdf}
              =
              \includegraphics[valign=c,scale=0.84]{figs/construction/runitor_lemma_rhs.pdf}
              =\;\;
              \includegraphics[valign=c,scale=0.84]{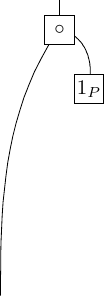}
              =\;\;
              1_{\mathcal{V}^{\mathcal{W}}( P, X )}
          \end{equation*}
          Likewise, the following also holds by the left counitality of $d$ and \cref{Eqn:UnitorLemma2} of \cref{Lem:ParamEqs}.
          \begin{equation*}
              \includegraphics[valign=c,scale=0.84]{figs/construction/lunitor_lemma_lhs.pdf}
              =
              \includegraphics[valign=c,scale=0.84]{figs/construction/lunitor_lemma_rhs.pdf}
              =\;\;
              \includegraphics[valign=c,scale=0.84]{figs/construction/obj_unitor_t2.pdf}
              =\;\;
              1_{\mathcal{V}^{\mathcal{W}}( P, X )}
          \end{equation*}
          Then $( \mathcal{V}^{\mathcal{W}}( P, - ), \eta, e )$ is associative and unital.
          Hence $\Omega_{(\mathcal{V},\mathcal{W})}( M )$ is lax monoidal.
          Since $M$ was arbitrary, then the assignment $\Omega_{(\mathcal{V},\mathcal{W})}$ is well-typed with respect to objects.
    \item \textbf{Morphisms}.
          Let $M = ( P, d_P, e_P ) \in \mathbf{Comon}( \mathcal{V} )_0$ and $N = ( Q, d_Q, e_Q ) \in \mathbf{Comon}( \mathcal{V} )_0$.
          Then by definition $\Omega_{(\mathcal{V},\mathcal{W})}( M ) = ( \mathcal{V}^{\mathcal{W}}( P, - ), \eta, e_P )$ and $\Omega_{(\mathcal{V},\mathcal{W})} = ( \mathcal{V}^{\mathcal{W}}( Q, - ), \mu, e_Q )$.
          Pick some comonoid homomorphism $f \in \mathbf{Comon}^{\text{op}}( \mathcal{V} )( M, N )$.
          Since $f \in \mathcal{V}^{\text{op}}( P, Q ) = \mathcal{V}( Q, P )$, then $\mathcal{V}^{\mathcal{W}}( f, - ): \mathcal{V}( P, - ) \Rightarrow \mathcal{V}^{\mathcal{W}}( Q, - )$.
          This means that $\Omega_{(\mathcal{V},\mathcal{W})}( f )$ is a natural transformation between the underlying functors of $\Omega_{(\mathcal{V},\mathcal{W})}( M )$ and $\Omega_{(\mathcal{V},\mathcal{W})}( N )$.
          It remains to be shown that $\Omega_{(\mathcal{V},\mathcal{W})}( f )$ respects the lax monoidal data.
          Since $f$ respects comultiplication, then the following equation holds by \cref{Eqn:NatTxLemma1} and \cref{Eqn:NatTxLemma2} of \cref{Lem:ParamEqs}.
          \begin{equation*}
              \includegraphics[valign=c,scale=0.84]{figs/construction/nat_tx_lemma_2_lhs.pdf}
              =
              \includegraphics[valign=c,scale=0.84]{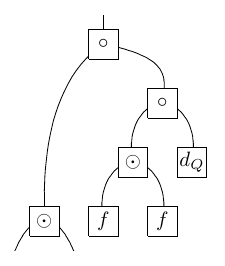}
              =
              \includegraphics[valign=c,scale=0.84]{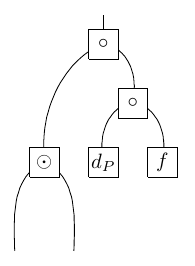}
              =
              \includegraphics[valign=c,scale=0.84]{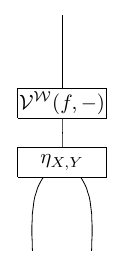}
          \end{equation*}
          Then $\mathcal{V}^{\mathcal{W}}( f, - )$ is compatible with the compositors.
          Moreover, since $f$ respects counits, then $\mathcal{V}^{\mathcal{W}}( f, - ) \bullet e_P = e_P \circ f = e_Q$.
          This means that $\mathcal{V}^{\mathcal{W}}( f, - )$ is also compatible with the unitors.
          Then $\mathcal{V}^{\mathcal{W}}( f, - ) \in \mathbf{MonCat}( \mathcal{V}, \mathcal{W} )\left( \Omega_{(\mathcal{V},\mathcal{W})}( M ), \Omega_{(\mathcal{V},\mathcal{W})}( N ) \right)$.
          Since $f$, $M$, and $N$ were arbitrary, then the assignment $\Omega_{(\mathcal{V},\mathcal{W})}$ is also well-typed with respect to morphisms.
    \end{itemize}
    Next, it must be shown that the assignment is a faithful functor.
    \begin{itemize}
    \item \textbf{Preserves Identities}.
          Let $M = ( P, d, e ) \in \mathbf{Comon}( \mathcal{V} )_0$.
          Recall that $\mathcal{V}^{\mathcal{W}}( 1_P, - )$ is the identity natural transformation on $\mathcal{V}^{\mathcal{W}}( P, - )$.
          Since $\Omega_{(\mathcal{V},\mathcal{V})}( M )$ is a lax monoidal functor with underlying functor $\mathcal{V}^{\mathcal{W}}( P, - )$ and $\Omega_{(\mathcal{V},\mathcal{V})}( 1_M ) = \mathcal{V}^{\mathcal{W}}( 1_P, - )$, then $\Omega_{(\mathcal{V},\mathcal{V})}^{\mathcal{W}}( 1_P )$ is the identity natural transformation on $\Omega_{(\mathcal{V},\mathcal{V})}( M )$.
          This identity lifts to the lax monoidal functor $\Omega_{(\mathcal{V},\mathcal{V})}( M )$.
          Since $M$ was arbitrary, then $\Omega_{(\mathcal{V},\mathcal{W})}$ preserves identities.
    \item \textbf{Preserves Composition}.
          Let $M = ( P, d_P, e_p )$, $N = ( Q, d_Q, e_Q )$ and $S = ( R, d_R, e_R )$ be comonoids in $\mathbf{Comon}( \mathcal{V} )_0$.
          If $f \in \mathbf{Comon}( \mathcal{V} )^{\text{op}}( M, N )$ and $g \in \mathbf{Comon}( \mathcal{V} )^{\text{op}}( N, S )$, then $f \in \mathcal{V}^{\text{op}}( P, Q ) = \mathcal{V}( Q, P )$ and $g \in \mathcal{V}^{\text{op}}( Q, R ) = \mathcal{V}( R, Q )$.
          Since enriched hom-functors are contravariant in the first component, then the following equation holds.
          \begin{equation*}
              \Omega_{(\mathcal{V},\mathcal{W})}( g ) \circ \Omega_{(\mathcal{V},\mathcal{W})}( f )
              =
              \mathcal{V}^{\mathcal{W}}( g, - ) \circ \mathcal{V}^{\mathcal{W}}( f, - )
              =
              \mathcal{V}^{\mathcal{W}}( f \circ g, - )
              =
              \Omega_{(\mathcal{V},\mathcal{W})}( g \circ^{\text{op}} f )
          \end{equation*}
          Since $f$ and $g$ were arbitrary, then $\Omega_{(\mathcal{V},\mathcal{W})}$ respects composition.
    \item \textbf{Faithfulness}.
           Assume that $f \in \mathbf{Comon}^{\text{op}}( \mathcal{V} )( M, N )$ and $g \in \mathbf{Comon}^{\text{op}}( \mathcal{V} )( M, N )$ with $\Omega_{(\mathcal{V},\mathcal{W})}( f, - ) = \Omega_{(\mathcal{V},\mathcal{W})}( g, - )$.
           There exists $P \in \mathcal{V}_0$ and $Q \in \mathcal{V}_0$ such that $f, g \in \mathcal{V}( P, Q )$.
           Then $\mathcal{V}( f, X )( k ) = k \circ f = \mathcal{V}^{\mathcal{W}}( f, X ) \bullet k = \mathcal{V}^{\mathcal{W}}( g, X ) \bullet k = k \circ g = \mathcal{V}( f, X )( k )$ for each $k \in \mathcal{V}( N, X )$.
           Then $\mathcal{V}( f, - ) = \mathcal{V}( g, - )$.
           Since contravariant hom-functors are faithful, then $f = g$.
           Since $f$ and $g$ were arbitrary, then $\Omega_{(\mathcal{V},\mathcal{W})}$ is faithful.
    \end{itemize}
    Now assume that $\mathcal{V}^{\mathcal{W}}$ is a $\mathcal{W}$-enriched braided monoidal category with enriched braiding $b$.
    It must be shown that $\Omega_{(\mathcal{V},\mathcal{W})}$ restricts to a functor $\mathbf{CComon}^{\text{op}}( \mathcal{V} ) \to \mathbf{BrCat}( \mathcal{V}, \mathcal{W} )$.
    Since $\mathbf{BrCat}( \mathcal{V}, \mathcal{W} )$ is a full subcategory of $\mathbf{MonCat}( \mathcal{V}, \mathcal{W} )$, then it suffices to show that $\Omega_{(\mathcal{V},\mathcal{W})}$ sends cocommutative comonoids to lax braided monoidal functors.
    \begin{itemize}
    \item \textbf{Objects Restrict}.
          Let $M = ( P, d, e ) \in \mathbf{CComon}( \mathcal{V} )_0$.
          Since $\Omega_{(\mathcal{V},\mathcal{W})}$ is functorial, then $\Omega_{(\mathcal{V},\mathcal{W})}( M ) = ( \mathcal{V}^{\mathcal{W}}( P, - ), \eta, e )$ is a lax monoidal functor.
          It remains to be shown that $\Omega_{(\mathcal{V},\mathcal{W})}( M )$ is a lax braided monoidal functor.
          Since $d$ is cocommutative, then the following equation holds by \cref{Eqn:BraidLemma} from \cref{Lem:ParamEqs}.
          \begin{equation*}
              \includegraphics[valign=c,scale=0.84]{figs/construction/braid_lemma_lhs.pdf}
              =
              \includegraphics[valign=c,scale=0.84]{figs/construction/braid_lemma_rhs.pdf}
              =
              \includegraphics[valign=c,scale=0.84]{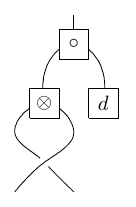}
              =
              \includegraphics[valign=c,scale=0.84]{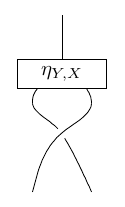}
          \end{equation*}
          Then $( \mathcal{V}^{\mathcal{W}}( P, - ), \eta, e )$ is commutative.
          This means that $\Omega_{(\mathcal{V},\mathcal{W})}( M ) \in \mathbf{BrCat}( \mathcal{V}, \mathcal{W} )$.
          Since $M$ was arbitrary, then the assignment $\Omega_{(\mathcal{V},\mathcal{W})}$ restricts to $\mathbf{Comon}( \mathcal{V} )^{\text{op}} \rightarrow \mathbf{BrCat}( \mathcal{V}, \mathcal{W} )$.
    \end{itemize}
    Since restrictions preserve faithfulness, then this restriction is also faithful.
\end{proof}

\begin{cor}
    \label{Cor:FullFunctor}
    If $\mathcal{V}$ is locally small and the monoidal unit in $\mathcal{W}$ is a separator, then $\Omega_{(\mathcal{V}, \mathcal{W})}$ is full.
    Moreover, if $\mathcal{V}^{\mathcal{W}}$ admits a $\mathcal{W}$-enriched braiding, then $\Omega_{(\mathcal{V}, \mathcal{W})}$ restricts to a full functor $\mathbf{CComon}^{\text{op}}( \mathcal{V} ) \to \mathbf{BrCat}( \mathcal{V}, \mathcal{W} )$.
\end{cor}

\begin{proof}
    By \cref{Thm:Reparam}, $\Omega_{(\mathcal{V}, \mathcal{W})}$ is a functor.
    The goal is to show that $\Omega_{(\mathcal{V},\mathcal{W})}$ is full by reduction to the classical Yoneda embedding.
    Let $M = ( P, d_P, e_P ) \in \mathbf{Comon}( \mathcal{V} )_0$ and $N = ( Q, d_Q, e_Q ) \in \mathbf{Comon}( \mathcal{V} )_0$.
    Pick some $\rho \in \mathbf{MonCat}( \mathcal{V}, \mathcal{W} )\left( \Omega_{(\mathcal{V},\mathcal{W})}( M ), \Omega_{(\mathcal{V},\mathcal{W})}( N ) \right)$, which is a natural transformation of the form $\rho: \mathcal{V}^{\mathcal{W}}( P, - ) \Rightarrow \mathcal{V}^{\mathcal{W}}( Q, - )$.
    Then define a transformation $\gamma: \mathcal{V}( P, - ) \Rightarrow \mathcal{V}( Q, - )$ such that $\gamma_X: f \mapsto \rho_X \bullet f$.
    This is well defined since $\rho_X: \mathcal{V}^{\mathcal{W}}( P, X ) \to \mathcal{V}^{\mathcal{W}}( Q, X )$ and $f \in \mathcal{V}( P, X ) = \mathcal{W}( \mathbb{J}, \mathcal{V}^{\mathcal{W}}( P, X ) )$.
    It will be shown that $\gamma$ is natural.
    Let $g \in \mathcal{V}( X, Y )$.
    Notice that if $h \in \mathcal{V}( Y, Z )$, then $(\mathcal{V}( Y, g ))( h ) = \mathcal{V}^{\mathcal{W}}( Y, g ) \bullet h$ by definition.
    Since $\rho$ is natural, the following equation holds for all $h \in \mathcal{V}( P, X )$.
    \begin{equation*}
        ( \mathcal{V}( Q, g ) \circ \gamma_X )( h )
        =
        \mathcal{V}^{\mathcal{W}}( Q, g ) \bullet \rho_X \bullet h
        =
        \rho_Y \bullet \mathcal{V}^{\mathcal{W}}( P, g ) \bullet h
        =
        ( \gamma_Y \circ \mathcal{V}( P, g ) )( h )
    \end{equation*}
    Then $\mathcal{V}( Q, g ) \circ \gamma_X = \gamma_Y \circ \mathcal{V}( P, g )$.
    Since $g$ was arbitrary, then $\gamma$ is natural.
    Then by the Yoneda embedding, there exists an $f \in \mathcal{V}( Q, P )$ such that $\mathcal{V}( f, - ) = \gamma$.
    It remains to be shown that $\rho_X = \mathcal{V}^{\mathcal{W}}( f, X )$.
    Let $X \in \mathcal{V}_0$.
    Then $\mathcal{V}^{\mathcal{W}}( f, X ) \bullet g = \mathcal{V}( f, X )( g ) = \gamma_X( g ) = \rho_X \bullet g$ for each $g \in \mathcal{V}( P, X )$.
    Since $\mathcal{V}( P, X ) = \mathcal{W}( \mathbb{J}, \mathcal{V}^{\mathcal{W}}( P, X ) )$ and $\mathbb{J}$ is a separator for $\mathcal{W}$, then $\mathcal{V}^{\mathcal{W}}( f, X ) = \rho_X$.
    Since $X$ was arbitrary, then $\mathcal{V}^{\mathcal{W}}( f, - ) = \rho$.
    It remains to be shown that $f \in \mathbf{Comon}( \mathcal{V} )^{\text{op}}( M, N )$.
    Let $( \mathcal{V}^{\mathcal{W}}( P, - ), \eta, e_P ) = \Omega_{(\mathcal{V},\mathcal{W})}( M )$ and $( \mathcal{V}^{\mathcal{W}}( Q, - ), \mu, e_Q ) = \Omega_{(\mathcal{V},\mathcal{W})}( N )$.
    By \cref{Eqn:NatTxLemma1} of \cref{Lem:ParamEqs}, the following equation holds where $W = P \odot P$.
    \begin{equation*}
        \includegraphics[valign=c,scale=0.84]{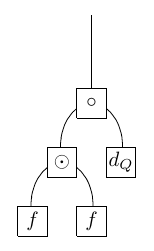}
        =
        \includegraphics[valign=c,scale=0.84]{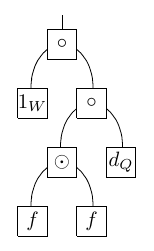}
        =
        \includegraphics[valign=c,scale=0.84]{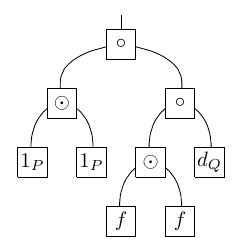}
        =
        \includegraphics[valign=c,scale=0.84]{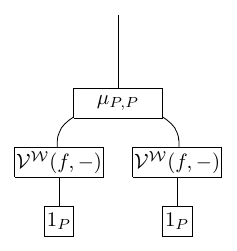}
    \end{equation*}
    By \cref{Eqn:NatTxLemma2} of \cref{Lem:ParamEqs}, the following equation also holds.
    \begin{equation*}
        \includegraphics[valign=c,scale=0.84]{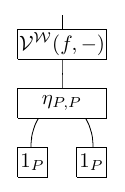}
        =
        \includegraphics[valign=c,scale=0.84]{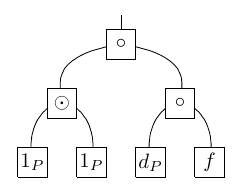}
        =
        \includegraphics[valign=c,scale=0.84]{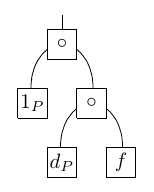}
        =
        \includegraphics[valign=c,scale=0.84]{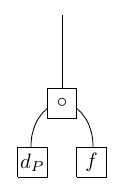}
    \end{equation*}
    Since $\mathcal{V}^{\mathcal{W}}( f, - )$ is compatible with the compositor, then the following equation also holds.
    \begin{equation*}
        \includegraphics[valign=c,scale=0.84]{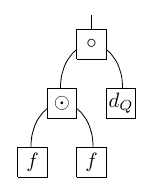}
        =
        \includegraphics[valign=c,scale=0.84]{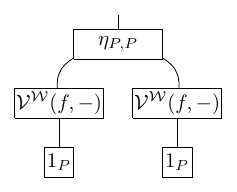}
        =
        \includegraphics[valign=c,scale=0.84]{figs/construction/full_compositor_t8.pdf}
        =
        \includegraphics[valign=c,scale=0.84]{figs/construction/full_compositor_t7.pdf}
    \end{equation*}
    In other words, $f$ respects comultiplication.
    Since $\mathcal{V}^{\mathcal{W}}( f, - )$ is compatible with the unitor, then also $e_P \circ f = \mathcal{V}^{\mathcal{W}}( f, - ) \bullet e_P = e_Q$.
    In other words, $f$ respects counits.
    Since $f$ respects comultiplication and conuits, then $f$ is a comonoid homomorphism from $N$ to $M$.
    Then $f \in \mathbf{Comon}( \mathcal{V} )^{\text{op}}( M, N )$ and $\rho = \Omega_{(\mathcal{V},\mathcal{W})}( f )$.
    Since $\rho$, $M$, and $N$ were arbitrary, then $\Omega_{(\mathcal{V},\mathcal{W})}$ is a full functor.
    Since $\mathbf{CComon}( \mathcal{V}, \mathcal{W} )$ is a full subcategory of $\mathbf{Comon}( \mathcal{V}, \mathcal{W} )$, then this restriction of $\Omega_{(\mathcal{V},\mathcal{W})}$ is also a full functor.
\end{proof}

\begin{cor}
    \label{Cor:UnderlyingShape}
    Let $( P, d, e ) = M \in \mathbf{Comon}( \mathcal{V} )_0$ and $F = \Omega_{( \mathcal{V}, \mathcal{W} )}( M )$.
    \begin{enumerate}
    \item If $\mathcal{C}^{\mathcal{V}} \in \mathcal{V}\mathbf{Cat}_0$, then $U( F_*( \mathcal{C}^{\mathcal{V}} ) )$ has the categorical structure depicted in~\cref{Fig:SetFunctors} with each diagram interpreted in $\mathcal{V}$.
    \item If $\mathcal{V}$ is admits a $\mathcal{W}$-enriched braiding and $\mathcal{C}^{\mathcal{V}} \in \mathcal{V}\mathbf{MonCat}_0$, then $U( F_*( \mathcal{C}^{\mathcal{V}} ) )$ admits the monoidal structure depicted in~\cref{Fig:SetFunctors} with each diagram interpreted in $\mathcal{V}$.
    \item If $\mathcal{V}$ is admits a $\mathcal{W}$-enriched symmetry, $\mathcal{C}^{\mathcal{V}} \in \mathcal{V}\mathbf{BrCat}_0$ and $M \in \mathbf{CComon}( \mathcal{V} )_0$, then $U( F_*( \mathcal{C}^{\mathcal{V}} ) )$ admits the braiding depicted in~\cref{Fig:SetFunctors} as a diagram in $\mathcal{V}$.
    \end{enumerate}
\end{cor}

\begin{proof}
    Let $\mathcal{D}^{\mathcal{W}} = F_*( \mathcal{C}^{\mathcal{V}} )$.
    It must be shown that the composition and identities in $U( \mathcal{D}^{\mathcal{W}} )$ are the same as those depicted in~\cref{Fig:SetFunctors}.
    \begin{itemize}
    \item \textbf{Sequential Composition}.
          If $X, Y, Z \in \mathcal{C}_0$,
          then the sequential composition of $\mathcal{D}^{\mathcal{W}}( X, Y )$ with $\mathcal{D}^{\mathcal{W}}( Y, Z )$ is given by the morphism $\mathcal{V}^{\mathcal{W}}( P, M_{X,Y,Z}^{\mathcal{V}} ) \bullet \eta_{S,T}$ where $S = \mathcal{V}^{\mathcal{W}}( P, \mathcal{C}^{\mathcal{V}}( Y, Z ) )$ and $T = \mathcal{V}^{\mathcal{W}}( P, \mathcal{C}^{\mathcal{V}}( X, Y ) )$.
          Then $g \star f = \mathcal{V}^{\mathcal{W}}( P, M_{X,Y,Z}^{\mathcal{V}} ) \bullet ( ( g \odot f ) \circ d ) =  M_{X,Y,Z}^{\mathcal{V}} \circ ( g \odot f ) \circ d$ for each $f \in U( \mathcal{D}^{\mathcal{W}} )( X, Y )$ and $g \in U( \mathcal{D}^{\mathcal{W}} )( Y, Z )$.
          This is the diagram depicted in \cref{Fig:SetFunctors:Comp}.
    \item \textbf{Identities}.
          If $X \in \mathcal{C}_0$, then the identity on $X$ in $\mathcal{D}^{\mathcal{W}}$ is defined by $\mathcal{V}^{\mathcal{W}}( P, 1_X^{\mathcal{V}} ) \bullet e = 1_X^{\mathcal{V}} \circ e$.
          Then the identity on $X$ in $U( \mathcal{D}^{\mathcal{W}} )$ is the diagram depicted in \cref{Fig:SetFunctors:Id}.
    \end{itemize}
    Further assume that $\mathcal{V}$ admits a $\mathcal{W}$-enriched braiding and $\mathcal{C}^{\mathcal{V}}$ admits a $\mathcal{V}$-monoidal structure given by the data $( \mathcal{C}^{\mathcal{V}}, \mathbb{K}, \otimes^{\mathcal{V}}, \alpha^{\mathcal{V}}, \lambda^{\mathcal{V}}, \rho^{\mathcal{V}} )$.
    It must be shown that the monoidal structure on $U( F_*( \mathcal{C}^{\mathcal{V}}, \mathbb{K}, \otimes^{\mathcal{V}}, \alpha^{\mathcal{V}}, \lambda^{\mathcal{V}}, \rho^{\mathcal{V}} ) )$ is the same as depicted in~\cref{Fig:SetFunctors}.
    The parallel composition follows analogously to sequential composition, and the coherence data follows analogously to the identities.
    Finally, assume that $\mathcal{V}$ is symmetric monoidal, $M$ is cocommutative, and $( \mathcal{C}^{\mathcal{V}}, \mathbb{K}, \otimes^{\mathcal{V}}, \alpha^{\mathcal{V}}, \lambda^{\mathcal{V}}, \rho^{\mathcal{V}} )$ admits a $\mathcal{V}$-enriched braiding $\beta^{\mathcal{V}}$.
    It follows analogously to the identities that $U( F_*( \mathcal{C}^{\mathcal{V}}, \mathbb{K}, \otimes^{\mathcal{V}}, \alpha^{\mathcal{V}}, \lambda^{\mathcal{V}}, \rho^{\mathcal{V}}, \beta^{\mathcal{V}} ) )$ admits the braiding depicted in \cref{Fig:SetFunctors}.
\end{proof}

It turns out that the underlying category functor can also understood in terms of parameterization.
For each $\mathcal{V}$-enriched category $\mathcal{C}^{\mathcal{V}}$, the underlying composition in $\mathcal{C}$ is given by $M \circ ( ( - ) \otimes ( - ) ) \circ \lambda_{\mathbb{I}}^{-1}$.
This coincides with the parameterized composition induced by the comonoid $T = ( \mathbb{I}, \lambda_{\mathbb{I}}^{-1}, 1_{\mathbb{I}} )$ in $\mathcal{V}$ with $\mathcal{W} = \mathbf{Set}$.
Since $1_{\mathbb{I}}$ is the counit in $T$, then the identities in $\mathcal{C}$ also arise through this parameterization.
These observations generalize to $\mathcal{V}$-enriched monoidal and $\mathcal{V}$-enriched braided monoidal categories, as shown in~\cref{Cor:UnderlyingFunctor}.
So far this discussion has been restricted to the case where $\mathcal{W} = \mathbf{Set}$, but the same construction works when $\Omega_{(\mathcal{V},\mathbf{Set})}( T )$ is replaced by $\Omega_{(\mathcal{V},\mathcal{W})}( T )$.
This motivates the \emph{underlying $\mathcal{W}$-category of $\mathcal{V}^{\mathcal{W}}$}, as defined in~\cref{Def:UnderlyingDCat}.
It is shown in~\cref{Thm:UnderlyingDCat}, that this new definition is compatible with the standard underlying category functor.

\begin{cor}
    \label{Cor:UnderlyingFunctor}
    The underlying category functor, the underlying monoidal category functor, and the underlying braided monoidal category functor for $\mathcal{V}$ are all induced by $F = \Omega_{(\mathcal{V},\mathbf{Set})}( T )$ where $T = ( \mathbb{I}, \lambda_{\mathbb{I}}^{-1}, 1_{\mathbb{I}} )$.
\end{cor}

\begin{proof}
    By \cref{Cor:UnderlyingShape}, the following claims hold.
    \begin{itemize}
    \item For each $X \in \mathcal{C}_0$ and $Y \in \mathcal{C}_0$, $F_*( \mathcal{C}^{\mathcal{V}} )( X, Y ) = \mathcal{V}\left( \mathbb{I}, \mathcal{C}^{\mathcal{V}}( X, Y ) \right) = \mathcal{C}( X, Y )$.
    \item For each $X \in \mathcal{C}_0$, $i_X = I^{\mathcal{V}}_X \circ 1_{\mathbb{I}} = 1^{\mathcal{V}}_X = I_X$.
    \item If $f \in F_*( \mathcal{C}^{\mathcal{V}} )( X, Y )$ and $g \in F_*( \mathcal{C}^{\mathcal{V}} )( Y, Z )$, then $g \star f = M^{\mathcal{V}} \circ ( g \odot f ) \circ \lambda_{\mathbb{I}}^{-1} = M( g, f )$.
    \end{itemize}
    Then $F_*( \mathcal{C}^{\mathcal{V}} ) = \mathcal{C}$.
    This means that the underlying category for $\mathcal{C}^{\mathcal{V}}$ is induced by $F$.
    Further assume that $\mathcal{V}$ admits a $\mathcal{W}$-enriched braiding and $\mathcal{C}^{\mathcal{V}}$ admits a $\mathcal{V}$-monoidal structure given by the data $( \mathcal{C}^{\mathcal{V}}, \mathbb{K}, \otimes^{\mathcal{V}}, A^{\mathcal{V}}, L^{\mathcal{V}}, R^{\mathcal{V}} )$.
    Then by \cref{Cor:UnderlyingShape}, the following claims hold.
    \begin{itemize}
    \item If $f \in F_*( \mathcal{C}^{\mathcal{V}} )( X, Y )$ and $g \in F_*( \mathcal{C}^{\mathcal{V}}( Z, W )$, then $f \boxtimes g = \otimes^{\mathcal{V}} \circ ( f \times g ) \circ \lambda_{\mathbb{I}}^{-1} = \otimes( f, g )$.
    \item If $X \in \mathcal{C}_0$, $Y \in \mathcal{C}_0$, $Z \in \mathcal{C}_0$, and $W \in \mathcal{C}_0$, then $a_{X,Y,Z,W} = A_{X,Y,Z,W}^{\mathcal{V}} \circ 1_{\mathbb{I}} = A_{X,Y,Z,W}$.
    \item If $X \in \mathcal{C}_0$, then $\ell_X = L_X^{\mathcal{V}} \circ 1_{\mathbb{I}} = L_X$ and $r_X = R_X^{\mathcal{V}} \circ 1_{\mathbb{I}} = R_X$.
    \end{itemize}
    Then $F_*( \mathcal{C}^{\mathcal{V}}, \mathbb{K}, \otimes^{\mathcal{V}}, A^{\mathcal{V}}, L^{\mathcal{V}}, R^{\mathcal{V}} ) = ( \mathcal{C}^{\mathcal{V}}, \mathbb{K}, \otimes, A, L, R )$, and therefore the underlying monoidal category for $\mathcal{C}^{\mathcal{V}}$ is induced by $F$.
    Further assume that $\mathcal{V}$ is symmetric monoidal and $( \mathcal{C}^{\mathcal{V}}, \mathbb{K}, \otimes^{\mathcal{V}}, A^{\mathcal{V}}, L^{\mathcal{V}}, R^{\mathcal{V}} )$ admits a $\mathcal{V}$-enriched braiding $B^{\mathcal{V}}$.
    Since $T$ is cocommutative, then by \cref{Cor:UnderlyingShape}, $b_{X,Y} = B^{\mathcal{V}}_{X,Y} \circ 1_{\mathbb{I}} = B_{X,Y}$ for each $X \in \mathcal{V}_0$ and $Y \in \mathcal{V}_0$.
    This means that the underlying braided monoidal category for $\mathcal{C}^{\mathcal{V}}$ is induced by $F$.
    Since $\mathcal{C}^{\mathcal{V}}$ was arbitrary, then this concludes the proof.
\end{proof}

\begin{defi}
    \label{Def:UnderlyingDCat}
    Let $T = ( \mathbb{I}, \lambda_{\mathbb{I}}^{-1}, 1_{\mathbb{I}} )$ and $F = \Omega_{(\mathcal{V},\mathcal{W})}( T )$.
    The \emph{underlying $\mathcal{W}$-enriched category of $\mathcal{C}^{\mathcal{V}}$} is $\mathcal{C}^{\mathcal{V}:\mathcal{W}} := F_*( \mathcal{C}^{\mathcal{W}} )$.
    This definition extends to $\mathcal{V}$-enriched monoidal and braided monoidal categories in the obvious way.
\end{defi}

\begin{thm}
    \label{Thm:UnderlyingDCat}
    If $\mathcal{D}^{\mathcal{W}} := \mathcal{C}^{\mathcal{V}:\mathcal{W}}$, then $\mathcal{C} = \mathcal{D}$.
    If $\mathcal{V}$ admits a $\mathcal{W}$-enriched braiding and $\mathcal{C}^{\mathcal{V}}$ is $\mathcal{V}$-enriched monoidal, then this extends to the monoidal structure.
    If the $\mathcal{W}$-enriched braiding is a symmetry and $\mathcal{C}^{\mathcal{V}}$ admits a $\mathcal{V}$-enriched braiding, then this extends to the braiding on $\mathcal{C}^{\mathcal{V}}$.
\end{thm}

\begin{proof}
    Let $T = ( \mathbb{I}, \lambda^{-1}_{\mathbb{I}}, 1_\mathbb{I} )$ and $F = \Omega_{(\mathcal{V},\mathcal{W})}( T )$ such that $\mathcal{D}^{\mathcal{W}} = F_*( \mathcal{C}^{\mathcal{V}} )$.
    Then by definition, $\mathcal{C}_0 = \mathcal{D}_0$.
    Next, let $X \in \mathcal{C}_0$ and $Y \in \mathcal{C}_0$.
    Then by definition, $\mathcal{D}( X, Y ) = \mathcal{W}( \mathbb{J}, \mathcal{D}^{\mathcal{W}}( X, Y ) ) = \mathcal{W}( \mathbb{J}, F( \mathcal{C}^\mathcal{W}( X, Y ) ) ) = \mathcal{W}( \mathbb{J}, \mathcal{V}^{\mathcal{W}}( \mathbb{I}, \mathcal{C}^{\mathcal{V}}( X, Y ) ) ) = \mathcal{V}( \mathbb{I}, \mathcal{C}^{\mathcal{V}}( X, Y ) ) = \mathcal{C}( X, Y )$.
    Since $X$ and $Y$ were arbitrary, then $\mathcal{C}$ and $\mathcal{D}$ have the same hom-sets as well.
    It remains to be shown that $\mathcal{C}$ and $\mathcal{D}$ agree on composition and identities.
    \begin{itemize}
    \item \textbf{Sequential Composition}.
          Let $f \in \mathcal{C}( X, Y )$ and $g \in \mathcal{C}( Y, Z )$.
          The sequential composition of $f$ and $g$ in $\mathcal{D}$ corresponds to the following diagram interpreted as a moprhism in $\mathcal{W}$.
          \begin{equation*}
              \includegraphics[valign=c,scale=0.84]{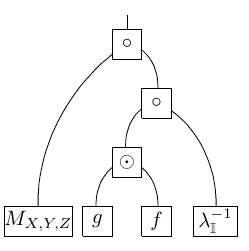}
              \in
              \mathcal{W}( \mathbb{J}, \mathcal{V}^{\mathcal{W}}( X, Z ) )
          \end{equation*}
          This is equal to the sequential composition $M^{\mathcal{V}}_{X,Y,Z} \circ ( g \odot f ) \circ \lambda_{\mathbb{I}}^{-1}$ of $f$ with $g$ in $\mathcal{C}$.
          Since $f$ and $g$ were arbitrary, then the sequential compositions agree.
    \item \textbf{Identities}.
          Let $X \in \mathcal{C}_0$.
          The identity on $X$ in $\mathcal{D}$ is $I^{\mathcal{V}}_X \circ u_\mathbb{J} = I^{\mathcal{V}}_X \in \mathcal{W}( \mathbb{J}, \mathcal{V}^{\mathcal{W}}( X, X ) )$, which equals the identity on $X$ in $\mathcal{C}$.
          Since $X$ was arbitrary, then the identities agree.
    \end{itemize}
    Then $\mathcal{C} = \mathcal{D}$.
    Further assume that $\mathcal{V}^{\mathcal{W}}$ is $\mathcal{W}$-enriched braided monoidal and $\mathcal{C}^{\mathcal{V}}$ admits a $\mathcal{V}$-monoidal structure $( \mathcal{C}^{\mathcal{V}}, \mathbb{K}, \otimes^{\mathcal{V}}, A^{\mathcal{V}}, L^{\mathcal{V}}, R^{\mathcal{V}} )$.
    By definition, $\mathbb{K}$ is the monoidal unit in both $\mathcal{C}$ and $\mathcal{D}$ when viewed as monoidal categories.
    It remains to be shown that $( \mathcal{C}, \mathbb{K}, \otimes, A, L, R ) = ( \mathcal{D}, \mathbb{K}, \boxtimes, a, \ell, r )$.
    The parallel composition follows analogously to sequential composition, and the coherence data follows analogously to the identities.
    Then $( \mathcal{C}, \mathbb{K}, \otimes, A, L, R ) = ( \mathcal{D}, \mathbb{K}, \boxtimes, a, \ell, r )$.
    Further assume that $\mathcal{V}$ is symmetric monoidal and $( \mathcal{C}^{\mathcal{V}}, \mathbb{K}, \otimes^{\mathcal{V}}, A^{\mathcal{V}}, L^{\mathcal{V}}, R^{\mathcal{V}} )$ admits a $\mathcal{V}$-enriched braiding $B^{\mathcal{V}}$.
    Since $T$ is cocommutative, the braiding in $\mathcal{C}$ can be shown to agree with the braiding in $\mathcal{D}$ by an argument analogous to the identities.
\end{proof}

It is clear from~\cref{Cor:UnderlyingFunctor}, that $\mathcal{V}^{\mathcal{W}}$-parameterizations are not limited to describing families of ansatz circuits.
However, the definition of a $\mathcal{V}^{\mathcal{W}}$-parameterization assumes that the choice of $d: P \odot P \to P$ and $e: P \to \mathbb{I}$ define a comonoid $( P, d, e )$.
This was motivated by the example in~\cref{Sect:Motivation}, which showed that if $\mathcal{W} = \mathbf{Set}$, then the coherence conditions for the lax monoidal functor are equivalent to requiring that $( P, d, e )$ is a comonoid.
However, this proof relied implicitly on the fact that $\{ * \}$ was a separator for $\mathbf{Set}$.
A natural question is whether this equivalence continues to hold for a general $\mathcal{W}$, where the monoidal unit need not be a separator.
The following theorem answers this question in the affirmative.

\begin{thm}
    Let $P \in \mathcal{V}_0$ and $F = \mathcal{V}^{\mathcal{W}}( P, - )$.
    If $d \in \mathcal{V}( P \odot P, P )$ and $( F, \eta, \epsilon )$ is a lax monoidal functor with $\eta$ defined in terms of $d$ as in \cref{Eqn:ParamEta}, then $( P, d, \epsilon )$ is a comonoid.
    Moreover, if $\mathcal{V}^{\mathcal{W}}$ admits a $\mathcal{W}$-enriched braiding and $( F, \eta, \epsilon )$ is lax braided monoidal functor, then $( P, d, \epsilon )$ is cocommutative comonoid.
\end{thm}

\begin{proof}
    Assume that $d \in \mathcal{V}( P \odot P, P )$ and $( F, \eta, \epsilon )$ is a lax monoidal functor with $\eta$ defined as in \cref{Eqn:ParamEta}.
    First, it must be shown that $d$ is coassociative.
    By \cref{Eqn:CompositorLemma1} of \cref{Lem:ParamEqs}, the following equation holds where $W = P \odot ( P \odot P )$.
    \begin{equation*}
        \includegraphics[valign=c,scale=0.84]{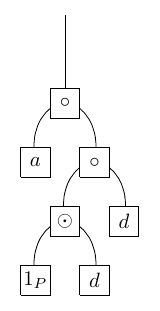}
        =
        \includegraphics[valign=c,scale=0.84]{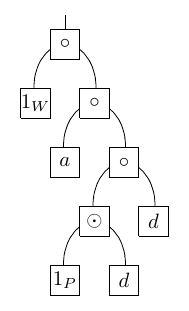}
        =
        \includegraphics[valign=c,scale=0.84]{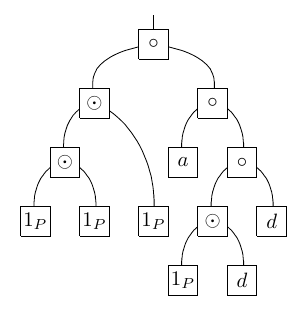}
        =
        \includegraphics[valign=c,scale=0.84]{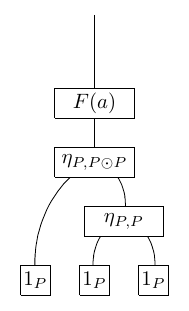}
    \end{equation*}
    By \cref{Eqn:CompositorLemma2} of \cref{Lem:ParamEqs}, the following equation also holds where $X = ( P \odot P ) \odot P$.
    \begin{equation*}
        \includegraphics[valign=c,scale=0.84]{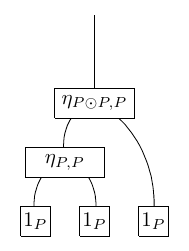}
        =
        \includegraphics[valign=c,scale=0.84]{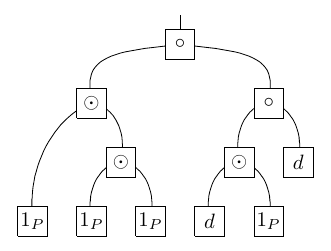}
        =
        \includegraphics[valign=c,scale=0.84]{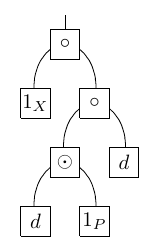}
        =
        \includegraphics[valign=c,scale=0.84]{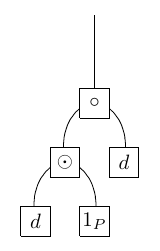}
    \end{equation*}
    These equations can be combined with the associativity of $\eta$ to obtain the following equation.
    \begin{equation*}
        \includegraphics[valign=c,scale=0.82]{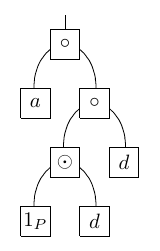}
        =
        \includegraphics[valign=c,scale=0.82]{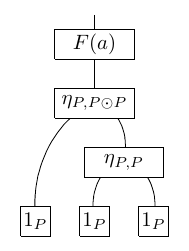}
        =
        \includegraphics[valign=c,scale=0.82]{figs/construction/comon_coassoc_t5.pdf}
        =
        \includegraphics[valign=c,scale=0.82]{figs/construction/comon_coassoc_t8.pdf}
    \end{equation*}
    This means that $d$ is coassociative.
    Next, it must be shown that $\epsilon$ is a right counit for $d$.
    Since $\epsilon$ is a right unit of $\eta$, then the following equation holds by \cref{Eqn:UnitorLemma1} of \cref{Lem:ParamEqs}
    \begin{equation*}
        1_P
        =
        \includegraphics[valign=c,scale=0.82]{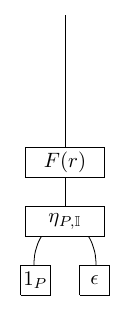}
        =
        \includegraphics[valign=c,scale=0.82]{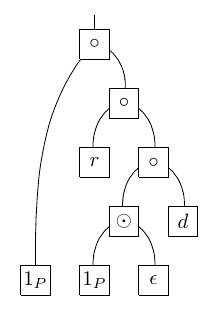}
        =
        \includegraphics[valign=c,scale=0.82]{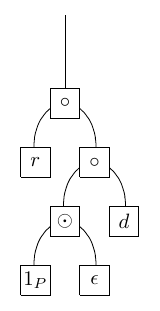}
    \end{equation*}
    Then $\epsilon$ is a right counit for $d$.
    Finally, it must be shown that $\epsilon$ is a right counit for $d$.
    Since $\epsilon$ is a left unit of $\eta$, then the following equation holds by \cref{Eqn:UnitorLemma2} of \cref{Lem:ParamEqs}
    \begin{equation*}
        1_P
        =
        \includegraphics[valign=c,scale=0.82]{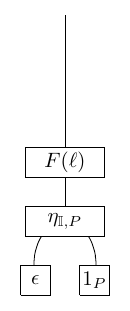}
        =
        \includegraphics[valign=c,scale=0.82]{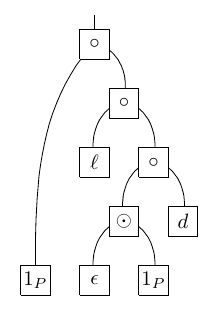}
        =
        \includegraphics[valign=c,scale=0.82]{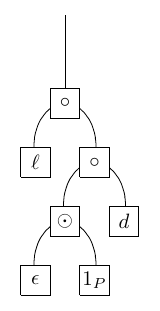}
    \end{equation*}
    Then $\epsilon$ is also a left counit for $d$.
    This means that $d$ is coassociative with counit $\epsilon$.
    In other words, $( P, d, \epsilon )$ is a comonoid.
    Now further assume that $\mathcal{V}^{\mathcal{W}}$ admits a $\mathcal{W}$-enriched braiding $b$ and $( F, \eta, \epsilon )$ is a lax braided monoidal functor.
    It must be shown that $d$ is cocommutative.
    The following equation holds by \cref{Eqn:BraidLemma} of \cref{Lem:ParamEqs}.
    \begin{equation*}
        \includegraphics[valign=c,scale=0.84]{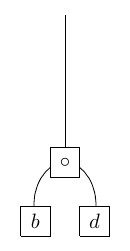}
        =
        \includegraphics[valign=c,scale=0.84]{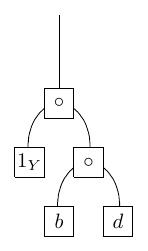}
        =
        \includegraphics[valign=c,scale=0.84]{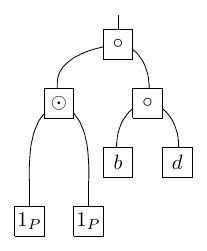}
        =
        \includegraphics[valign=c,scale=0.84]{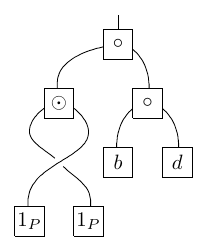}
        =
        \includegraphics[valign=c,scale=0.84]{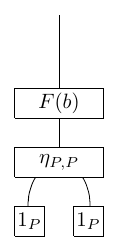}
    \end{equation*}
    The following equation also holds by definition, where $Y = P \odot P$.
    \begin{equation*}
        \includegraphics[valign=c,scale=0.82]{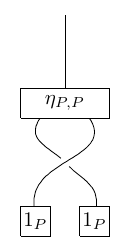}
        =
        \includegraphics[valign=c,scale=0.82]{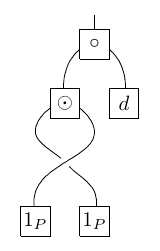}
        =
        \includegraphics[valign=c,scale=0.82]{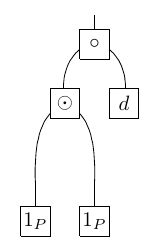}
        =
        \includegraphics[valign=c,scale=0.82]{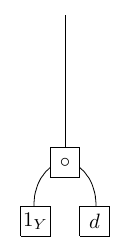}
        =
        d
    \end{equation*}
    Since $( F, \eta, \epsilon )$ is a lax braided monoidal functor, then the following equation holds.
    \begin{equation*}
        \includegraphics[valign=c,scale=0.84]{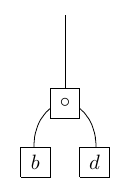}
        =
        \includegraphics[valign=c,scale=0.84]{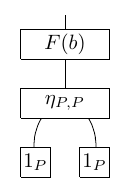}
        =
        \includegraphics[valign=c,scale=0.84]{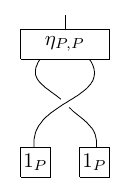}
        =
        d
    \end{equation*}
    Then $d$ is cocommutative.
    In other words, $( P, d, \epsilon )$ is a cocommutative comonoid.
\end{proof}

The remainder of this section will consider the natural transformations between the parameterization functors.
Assume that $\mathcal{V}^{\mathcal{W}}$ is a $\mathcal{W}$-enriched monoidal category, with $M$ and $N$ comonoid objects in $\mathcal{V}$.
It follows from~\cref{Thm:Reparam} that every natural transformation $\eta: F \Rightarrow G$ with $F = \Omega_{(\mathcal{V},\mathcal{W})}( M )$ and $G = \Omega_{(\mathcal{V},\mathcal{W})}( N )$ will correspond to pre-composition by some comonoid homomorphism $\tau \in \mathbf{Comon}( \mathcal{V} )^{\text{op}}( M, N )$.
Given a $\mathcal{V}$-enriched category $\mathcal{C}^{\mathcal{V}}$, the natural transformation $\eta$ will induce a functor $\eta_*( \mathcal{C}^{\mathcal{V}} ): F_*( \mathcal{C}^{\mathcal{V}} ) \to G_*( \mathcal{C}^{\mathcal{V}} )$ which acts by pre-composing each parameterized morphism $f \in \mathcal{V}( M, \mathcal{C}^{\mathcal{V}}( X, Y ) )$ with $\tau: N \to M$.
Clearly, these natural transformations capture some notion of reparameterization.
As an illustration, \cref{Ex:Reparam} shows how the correct choice of comonoid homomorphism can encode the fact that the $Z$-rotations from~\cref{Sect:Motivation} are periodic functions.

\begin{exa}[Z-Rotations Are Periodic]
    \label{Ex:Reparam}
    An abstract rotation type~\cite{Wesley2024} is analogous to an abstract data type, in that an abstract rotation type is a symbolic rotation gate $R( - )$ together with a list of abstract specifications it satisfies.
    One specification discussed in~\cite{Wesley2024} was $p$-periodicity, which states that that $R( \theta ) = R( \theta + p )$ for all $\theta \in \mathbb{R}$.
    The notion of $p$-periodicity can be understood topologically as saying that $R( \theta )$ factors through the circle.
    For concreteness, the circle $S = [0,2\pi] / \{ 0 = 2\pi \}$ will be considered with canonical projection $\tau: \mathbb{R} \to S$.
    It follows that a continuous function of type $f: \mathbb{R} \to \mathbb{C}\mathbf{Vect}( n, m )$ is $(2\pi)$-period if and only if there exists a continuous function of type $g: S \to \mathbb{C}\mathbf{Vect}( n, m )$ satisfying $f = g \circ \tau$.
    Since \textbf{Top} is Cartesian monoidal, then $S$ admits a unique comonoid structure $M$.
    If $F = \Omega_{(\textbf{Top},\textbf{Set})}( M )$, then $\mathcal{D} = F_*( \mathbb{C}\textbf{FVect} )$ denotes the category of ($2\pi$)-periodic rotations.
    Since $R_Z( - )$ is $(2\pi)$-periodic, then there exists a morphism $R( - ) \in \mathcal{D}( 2, 2 )$ such that $R_Z( \theta ) = R( \tau( \theta ) )$.
    Moreover, $\eta := \Omega_{(\textbf{Top},\textbf{Set})}( \tau )$ induces a functor $\eta_*( \mathbb{C}\mathbf{FVect} )$ which would send $R( - )$ to $R_Z( - )$.
    In other words, $\eta_*( \mathbb{C}\mathbf{FVect} )$ picks out $(2\pi)$-periodic rotations.
\end{exa}

Two special examples of reparameterization are the inclusion of constant morphisms from an underlying category into a category of parameterized morphisms, and the evaluation of a parameterized morphism at a given parameter.
For a choice of comonoid $M$ with $F = \Omega_{(\mathcal{V},\mathcal{W})}( M )$, the functor which picks out constant morphisms $j: \mathcal{C}^{\mathcal{V}:\mathcal{W}} \to F_*( \mathcal{C}^{\mathcal{V}} )$ should be induced by a natural transformation $\eta: \Omega_{(\mathcal{V},\mathcal{W})}( T ) \rightarrow F$.
An obvious choice for $\eta$ is pre-composition by the counit of $M$, as outlined in the motivating example.
It is shown in \cref{Cor:UniqueConstants} that this choice of $\eta$ is unique provided that $\mathcal{V}$ is locally small and the monoidal unit in $\mathcal{W}$ is a separator.
In the case where $\mathcal{W} = \mathbf{Set}$, this means that the \emph{constant} morphisms in $F_*( \mathcal{C}^{\mathcal{V}} )$ are precisely those obtained from $\mathcal{C}$ via pre-composition with the counit of $M$.
Dually, evaluation of a morphism in $F_*( \mathcal{C}^{\mathcal{V}} )$ at a point in $M$ should be given by a comonoid homomorphism $\theta \in \mathbf{Comon}^{\text{op}}( M, T )$.
As one would expect, if $M$ is a comonoid on $P$ then $\theta$ is a generalized element of $P$, since the morphisms in $\mathbf{Comon}^{\text{op}}( M, T )$ are precisely the morphisms in $\mathcal{V}( \mathbf{I}, P )$ which respect the comonoid structure of $M$.
The $\mathcal{W}$-functors these induce under $\Omega_{(\mathcal{V},\mathcal{W})}$ are the evaluation functors $\ev_\theta: F_*( \mathcal{C}^{\mathcal{V}} ) \to \mathcal{C}^{\mathcal{V}:\mathcal{W}}$.
When $\mathbf{Comon}^{\text{op}}( M, T )$ is inhabited by a morphism $\theta$, then $\ev_\theta$ is always a retraction to $j$ as proven in~\cref{Thm:Retract}.
However, $\mathbf{Comon}^{\text{op}}( M, T )$ need not be inhabited by any generalized elements, as illustrated in \cref{Ex:NoParams}.

\begin{cor}
    \label{Cor:UniqueConstants}
    Let $T = ( \mathbb{I}, \lambda_{\mathbb{I}}^{-1}, 1_{\mathbb{I}} )$ and $M = ( P, d, e ) \in \mathbf{Comon}( \mathcal{V} )_0$.
    \begin{enumerate}
    \item $e \in \mathbf{Comon}( \mathcal{V} )^{\text{op}}( T, M )$.
    \item If $M$ is cocommutative, then $e \in \mathbf{CComon}( \mathcal{V} )^{\text{op}}( T, M )$.
    \item If $\mathcal{V}$ is locally small and the monoidal unit in $\mathcal{W}$ is a separator, then $\Omega_{(\mathcal{V},\mathcal{W})}( e )$ is the unique natural transformation from $\Omega_{(\mathcal{V},\mathcal{W})}( T )$ to $\Omega_{(\mathcal{V},\mathcal{W})}( M )$
    \end{enumerate}
\end{cor}

\begin{proof}
    The proof proceeds as follows.
    \begin{enumerate}
    \item Since $( P, d, e )$ is a comonoid, then $( e \otimes e ) \circ d = \lambda_{\mathbb{I}}^{-1} \circ e$.
          Then $e$ maps the comultiplication of $M$ to the comultiplication of $T$.
          Also, $1_\mathbb{I} \circ e = e$ by definition.
          Then $e$ maps the counit of $M$ to the counit of $T$ as well.
          Then $e \in \mathbf{Comon}( \mathcal{V} )( M, T ) = \mathbf{Comon}( \mathcal{V} )^{\text{op}}( T, M )$.
    \item Assume that $M$ is a cocommutative comonoid.
          Since $T$ is also a cocommutative comonoid, then $e \in \mathbf{Comon}( \mathcal{V} )( M, T ) = \mathbf{CComon}( \mathcal{V} )( M, T ) = \mathbf{CComon}( \mathcal{V} )^{\text{op}}( T, M )$.
    \item Assume that $\mathcal{V}$ is locally small and the monoidal unit in $\mathcal{W}$ is a separator.
          Then $\Omega_{(\mathcal{V},\mathcal{W})}$ is a full functor by~\cref{Cor:FullFunctor}.
          Let $\rho: \Omega_{(\mathcal{V},\mathcal{W})}( T ) \Rightarrow \Omega_{(\mathcal{V},\mathcal{W})}( M )$.
          Since $\Omega_{(\mathcal{V},\mathcal{W})}$ is full, then there exists an $f \in \mathbf{Comon}( \mathcal{V} )^{\text{op}}( T, M ) = \mathbf{Comon}( \mathcal{V} )( M, T )$ such that $\rho = \Omega_{(\mathcal{V},\mathcal{W})}( f )$.
          Since $f$ is a comonoid homomorphism, then $f$ maps the counit of $M$ to the counit of $T$.
          This means that $f = 1_\mathbb{I} \circ f = e$, which implies that $\rho = \Omega_{(\mathcal{V},\mathcal{W})}( e )$.
          Since $\rho$ was arbitrary, then $\Omega_{(\mathcal{V},\mathcal{W})}( e )$ is the unique natural transformation from $\Omega_{(\mathcal{V},\mathcal{W})}( T )$ to $\Omega_{(\mathcal{V},\mathcal{W})}( M )$.
          \endproof
    \end{enumerate}
\end{proof}

\begin{thm}
    \label{Thm:Retract}
    Let $M \in \mathbf{Comon}( \mathcal{V} )_0$ and $F = \Omega_{(\mathcal{V},\mathcal{W})}( M )$.
    The following hold for each generalized element $\theta \in \mathbf{Comon}( \mathcal{V} )^{\text{op}}( M, T )$.
    \begin{enumerate}
    \item The $\mathcal{W}$-enriched functor $\ev_\theta: F_*( \mathcal{C}^{\mathcal{V}} ) \to \mathcal{C}^{\mathcal{V}:\mathcal{W}}$ is a retraction to $j: \mathcal{C}^{\mathcal{V}:\mathcal{W}} \to F_*( \mathcal{C}^{\mathcal{V}} )$.
    \item If $\mathcal{V}^{\mathcal{W}}$ admits a $\mathcal{W}$-enriched braiding and $\mathcal{C}^{\mathcal{V}}$ is $\mathcal{V}$-enriched monoidal, then the $\mathcal{W}$-enriched monoidal functor $\ev_\theta: F_*( \mathcal{C}^{\mathcal{V}} ) \to \mathcal{C}^{\mathcal{V}:\mathcal{W}}$ is a retraction to $j: \mathcal{C}^{\mathcal{V}:\mathcal{W}} \to F_*( \mathcal{C}^{\mathcal{V}} )$.
    \item If $\mathcal{V}^{\mathcal{W}}$ admits a $\mathcal{W}$-enriched symmetry, $\mathcal{C}^{\mathcal{V}}$ is $\mathcal{V}$-enriched braided monoidal, and $M$ is cocommutative, then the $\mathcal{W}$-enriched braided monoidal functor $\ev_\theta: F_*( \mathcal{C}^{\mathcal{V}} ) \to \mathcal{C}^{\mathcal{V}:\mathcal{W}}$ is a retraction to $j: \mathcal{C}^{\mathcal{V}:\mathcal{W}} \to F_*( \mathcal{C}^{\mathcal{V}} )$.
    \end{enumerate}
\end{thm}

\begin{proof}
    Let $\theta \in \mathbf{Comon}( \mathcal{V} )^{\text{op}}( M, T ) = \mathbf{Comon}( \mathcal{V} )( T, M )$.
    Then $e \circ \theta \in \mathbf{Comon}( \mathcal{V} )( T, T )$ by \cref{Cor:UniqueConstants} part (1).
    Since $1_{\mathbb{I}}$ is the counit of $T$, then $1_{\mathbb{I}} \circ ( e \circ \theta ) = 1_{\mathbb{I}}$.
    Since $\Omega_{(\mathcal{V},\mathcal{W})}$ is contravariant, then $\Omega_{(\mathcal{V},\mathcal{W})}( \theta ) \circ \Omega_{(\mathcal{V},\mathcal{W})}( e ) = \Omega_{(\mathcal{V},\mathcal{W})}( 1_{\mathbb{I}} )$.
    There are three cases to consider.
    \begin{enumerate}
    \item Let $\mathcal{C}^{\mathcal{V}}$ be a $\mathcal{V}$-enriched category.
          Then $j = \left( \Omega_{(\mathcal{V},\mathcal{W})}( e ) \right)_*( \mathcal{C}^{\mathcal{V}} ): \mathcal{C}^{\mathcal{V}:\mathcal{W}} \rightarrow F_*( \mathcal{V}^{\mathcal{W}} )$ and $\ev_\theta = \left( \Omega_{(\mathcal{V},\mathcal{W})}( \theta ) \right)_*( \mathcal{C}^{\mathcal{V}} ): F_*( \mathcal{C}^{\mathcal{V}} ) \to \mathcal{C}^{\mathcal{V}:\mathcal{W}}$ are both $\mathcal{W}$-enriched functors by change of base.
          It follows that $\ev_\theta \circ j = \left( \Omega_{(\mathcal{V},\mathcal{W})}( 1_{\mathbf{I}} ) \right)_*( \mathcal{C}^{\mathcal{V}} ) = 1$.
          Therefore $\ev_\theta$ is a retraction to $j$.
    \item Further assume that $\mathcal{V}^{\mathcal{W}}$ admits a $\mathcal{W}$-enriched braiding and $\mathcal{C}^{\mathcal{V}}$ is a $\mathcal{V}$-enriched monoidal.
          Then $j$ and $\ev_\theta$ are $\mathcal{W}$-enriched monoidal functors by change of base.
    \item Further assume that $M$ is cocommutative, the $\mathcal{W}$-enriched braiding is invertible, and $\theta \in \mathbf{CComon}( \mathcal{V} )^{\text{op}}( M, T )$.
          Then by \cref{Cor:UniqueConstants} part (2), $e \in \mathbf{CComon}( \mathcal{V} )^{\text{op}}( T, M )$.
          Then $\mathcal{C}^{\mathcal{V}}$ may be a $\mathcal{V}$-enriched braided monoidal category, in which case $j$ and $\ev_\theta$ are $\mathcal{W}$-enriched braided monoidal functors by change of base.
          \endproof
    \end{enumerate}
\end{proof}

\begin{exa}[Comonoid Without Parameters]
    \label{Ex:NoParams}
    Recall that every meet-semilattice defines a Cartesian monoidal category.
    For example, the category $\mathbbm{2}$ is trivially a meet-semilattice, with enrichment over $\mathbbm{2}$ giving rise to posetal categories.
    However, enrichment over any meet-semilattice $\mathcal{V}$ will always yield an underlying posetal category $\mathcal{C}$ such that the hom-set $\mathcal{C}( X, Y )$ is inhabited if and only if $\mathcal{C}^{\mathcal{V}}( X, Y ) = \top$.
    It was shown in~\cite{Wesley2025} that parameterization could be used to weaken the requirement that $\mathcal{C}^{\mathcal{V}}( X, Y ) = \top$.
    In short, each object $P \in \mathcal{V}_0$ defines a posetal category $\mathcal{D}$ such that $\mathcal{D}( X, Y )$ is inhabited if and only if $\mathcal{C}^{\mathcal{V}}( X, Y ) \ge P$.
    Since $\mathcal{V}$ is Cartesian monoidal, then the category $\mathcal{D}$ is obtained via parameterization with the canonical comonoid structure $M$ on $P$.
    It is easy to show that the functor $j: \mathcal{C} \to \mathcal{D}$ is faithful.
    In spite of this, $\mathbf{Comon}^{\text{op}}( M, T )$ is not inhabited.
    This is because $P < \top$, which implies that $\mathcal{V}( \top, P ) = \varnothing$.
    Consequently, $j$ does not admit a retraction of the form $\ev_\theta$.
\end{exa}

To conclude this section, two examples from quantum information will be described in which $\mathcal{W} \ne \mathbf{Set}$ and $\mathcal{V}$ is not Cartesian monoidal.
These examples serve to justify the level of generality at which \cref{Def:Params} is stated.
It will be shown in \cref{Sect:CartMon} and \cref{Sect:Closed} how this general construction can be used to recover non-trivial results from quantum information, with an emphasis on control functors.

\begin{exa}[Generalized Controls]
    \label{Ex:GenCtrl}
    Let $k = 2^n$ for some $n \in \mathbb{N}$.
    Then the unit vectors in $P = \mathbb{C}^k$ form the state space of an $n$-qubit quantum system.
    Let $\{ \ket{0}, \ket{1}, \ldots, \ket{k - 1} \}$ denote the computational basis states in $P$.
    By~\cref{Prop:BasisComonoid}, there exists a unique commutative special $\dagger$-Frobenius structure $A = ( P, d^\dagger, d, e^\dagger, e )$ whose copyable states are precisely the basis vectors $\{ \ket{0}, \ket{1}, \ldots, \ket{k - 1} \}$.
    Recall that $d: \ket{j} \to \ket{j} \otimes \ket{j}$ and $e: \ket{j} \to 1$ for each $j \in \{ 0, 1, \ldots, k - 1 \}$.
    Since $A$ is commutative, then $d$ is necessarily cocommutative.
    Since $\mathcal{V} = \mathbb{C}\mathbf{FVect}$ is $\mathcal{V}$-enriched, then there exists a functor $F = \Omega_{\mathcal{V},\mathcal{V}}( M )$ which induces a new category $\mathcal{C}^{\mathcal{V}} = F_*( \mathcal{V}^{\mathcal{V}} )$.
    It can be shown that the morphisms in this category correspond to the generalized $n$-qubit control operations used in Quipper.
    To see this correspondence, first notice that if $V$ and $W$ are finite-dimensional vector spaces, then the hom-object $\mathcal{C}^{\mathcal{V}}( V, W )$ is defined to be the vector space $\mathcal{V}( P, [ V, W ] )$.
    If $f \in \mathcal{C}^{\mathcal{V}}( V, W )$, then by linearity there exists a family of matrices $\{ f_0, f_1, \ldots, f_{k-1} \} \in [ V, W ]$ such that $f: \ket{j} \mapsto f_j$.
    Then by the hom-tensor adjunction, $f$ corresponds to a map $g: P \otimes V \to W$ such that $g: \ket{j} \otimes v \mapsto f_j( v )$ for all $j \in \{ 0, 1, \ldots, k - 1 \}$.
    At first glance, the map $g$ appears to have a variation of generalized controls, with the exception that the control qubits are consumed by $g$.
    However, precomposing with $d \otimes 1_V$ yields a new morphism $h: P \otimes V \to P \otimes W$ such that $h: \ket{j} \otimes v \to \ket{j} \otimes f_j( v )$ for each $j \in \{ 0, 1, \ldots, k - 1 \}$.
    Note that since $\mathcal{V}$ is in fact compact closed, then $h$ corresponds to the following string diagram in $\mathcal{V}$.
    \begin{equation*}
        h
        =
        \includegraphics[valign=c,scale=0.84]{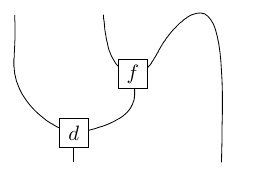}
    \end{equation*}
    Since $h$ has generalized controls, then it should follow that $h$ is an $A$-module homomorphism from the free $A$-module $P \otimes U$ to the free $A$-module $P \otimes V$.
    This follows from the Frobenius law and bifunctoriality, as illustrated in the following equation of diagrams in $\mathcal{V}$.
    \begin{equation*}
        \includegraphics[valign=c,scale=0.84]{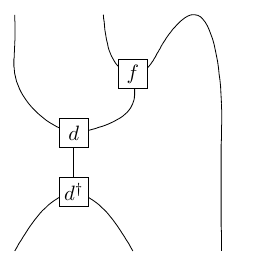}
        =
        \includegraphics[valign=c,scale=0.84]{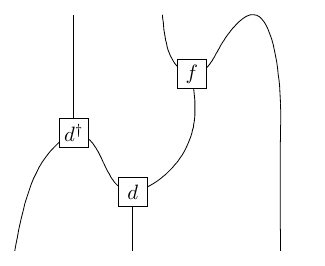}
    \end{equation*}
    Moreover, the composition $( \star )$ in $\mathcal{C}^{\mathcal{V}}$ is the expected composition of gates with generalized controls.
    For example, if $s \in \mathcal{C}^{\mathcal{V}}( W, L )$, then $s \star f: \ket{j} \mapsto s_j \circ f_j$ for each $j \in \{ 0, 1, \ldots, k - 1 \}$, as expected.
    Before concluding this example, it should be noted that $h$ is a controlled operation even if $\mathcal{V} \ne \mathbb{C}\mathbf{FVect}$.
    Indeed, $\mathcal{V}$ need not even be a linear category.
    Note that the construction of $h$ from $f$ can be performed in any compact closed category $\mathcal{V}$, and that the proof that $h$ is an $A$-module homomorphism holds given any Frobenius structure $A$.
    In the case where $\mathcal{V}$ is a dagger-compact closed category and $A$ is commutative special $\dagger$-Frobenius structure, then the $A$-module homomorphism $h$ defines a controlled operation in $\mathcal{V}$.
\end{exa}

\begin{exa}[Shared Entanglement]
    Following the setup of the previous example, let $A$ be the unique commutative special $\dagger$-Frobenius structure $A$ on associated with the computational basis $\{ \ket{0}, \ket{1} \}$.
    As shown in~\cref{Fig:Frob}, the Frobenius structures $A$ on $\mathbb{C}^2$ induces an adjunction between $\mathbb{C}^2$ and itself, with the following unit and counit.
    \begin{align*}
        \eta &:=
        \includegraphics[valign=c,scale=0.84]{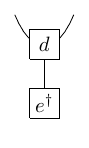}
        : 1 \mapsto \ket{00} + \ket{11}
        &
        \epsilon &:=
        \includegraphics[valign=c,scale=0.84]{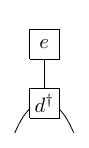}
        : \ket{ij} \mapsto \begin{cases}
            1 & i = j \\
            0 & \text{otherwise}
        \end{cases}
    \end{align*}
    Recall from~\cref{Sect:Background}, that the unit of this adjunction induces a comonoid $M = ( P, \delta, \epsilon )$ on $P = \mathbb{C}^2 \otimes \mathbb{C}^2$, commonly referred to as the pair-of-pants comonoid.
    Up to rescaling, the comultiplication $d$ sends $P$ to $P \otimes P$ by injecting a Bell-pair between the two qubits in $P$.
    When parameterizing with respect to $M$, these injected Bell-pairs recover a simple model of communicating quantum processes.
    Since $\mathcal{V} = \mathbb{C}\mathbf{FVect}$ is $\mathcal{V}$-enriched, then there exists a functor $F = \Omega_{\mathcal{V},\mathcal{V}}( M )$ which induces a new category $\mathcal{C}^{\mathcal{V}} = F_*( \mathcal{V}^{\mathcal{V}} )$.
    If $V$ and $W$ are finite-dimension vector spaces, then the hom-object $\mathcal{C}^{\mathcal{V}}( V, W )$ is defined to be the vector space $\mathcal{V}( P, [ V, W ] )$.
    If $f \in \mathcal{C}^{\mathcal{V}}( V, W )$, then by the hom-tensor adjunction, $f$ corresponds to a map of the shape $P \otimes V \to W$.
    Intuitively, $f$ is a quantum process which maps inputs of type $V$ to outputs of type $W$, while using a pair of qubits as a resource.
    When composing $f$ with another morphism $g \in \mathcal{C}^{\mathcal{V}}( W, L )$, it becomes clear that the resource qubits are used to inject Bell-pairs.
    This is illustrated in~\cref{Fig:Network:Comp}.

    \begin{figure}[t]
        \begin{subfigure}[b]{0.26\textwidth}
            \centering
            \includegraphics[valign=c,scale=0.84]{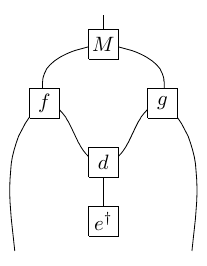}
            \caption{The composite $f \star g$.}
            \label{Fig:Network:Comp}
        \end{subfigure}
        \begin{subfigure}[b]{0.36\textwidth}
            \centering
            \includegraphics[valign=c,scale=0.84]{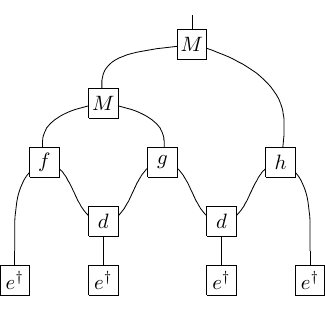}
            \caption{The line topology for $( f, g, h )$.}
            \label{Fig:Network:Line}
        \end{subfigure}
        \begin{subfigure}[b]{0.36\textwidth}
            \centering
            \includegraphics[valign=c,scale=0.84]{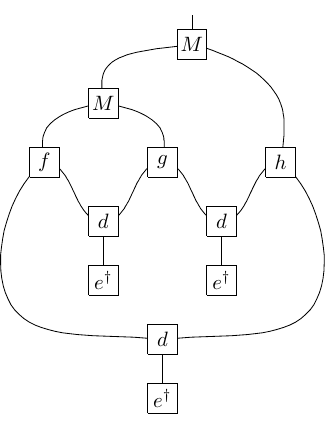}
            \caption{The ring topology for $( f, g, h )$.}
            \label{Fig:Network:Ring}
        \end{subfigure}
        \caption{Network topologies obtained from the pair-of-pants comonoids.}
        \label{Fig:Network}
    \end{figure}

    After composing two morphisms in $\mathcal{C}^{\mathcal{V}}$, a new morphism is obtained which also consumes two resource qubits.
    However, an executable process should be a element of $[ V, L ]$ as opposed to a morphism $\mathcal{V}( P, [ V, L ] )$.
    There are two canonical ways to obtain a generalized element of $[ V, L ]$ from a morphism in $\mathcal{V}( P, [ V, L ] )$.
    The first choice, denoted \texttt{line}, is to precompose with $e^\dagger \otimes e^\dagger$.
    This corresponds to injecting a uniform superposition in place of the two resource qubits, and can be thought of as composing the processes with a linear network topology.
    The second choice, denoted \texttt{ring} is to precompose with $\eta$.
    This corresponds to injecting a Bell-pair between the first and last process, and can be thought of as composing processes with a circular network topology.
    Both choices are illustrated in~\cref{Fig:Network}, for a network with three parallel processes $( f, g, h )$.

    It should be noted that $( P, \delta, \epsilon )$ is not a commutative comonoid.
    Intuitively, applying $\sigma_{P,P}$ to the output of $d$ moves the injected Bell-state from the inner-most resource qubits to the outer-most resource qubits.
    Since the morphisms in $\mathcal{V}( P, [ V, W ] )$ need not be symmetric with respect to their resource qubits, then there is no reason to expect $\mathcal{C}^{\mathcal{V}}$ to inherit the symmetry from $\mathcal{V}$.
    This is illustrated by the following equation of string diagrams.
    \begin{equation*}
        \sigma_{P,P} \circ \delta
        =
        \includegraphics[valign=c,scale=0.84]{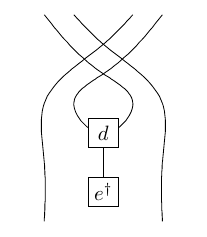}
        =
        \includegraphics[valign=c,scale=0.84]{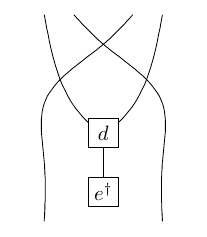}
        \ne
        \includegraphics[valign=c,scale=0.84]{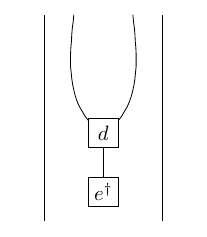}
        =
        \delta
    \end{equation*}
    To generalize beyond $\mathbb{C}\mathbf{FVect}$, assume that $X^*$ is the right-dual of $X \in \mathcal{V}_0$ with unit $\eta$ and counit $\epsilon$ in a category $\mathcal{C}^{\mathcal{V}}$.
    Then the pair-of-pants comonoid induces a parameterized category $\mathcal{C}^{\mathcal{V}}$ as described above.
    In this category, a ring network topology can be achieved without any further assumptions.
    Moreover, if $X = X^*$ and $( \eta, \epsilon )$ are induced by some Frobenius algebra $A$ on $X$, then a linear network topology is also achievable in some canonical way.
    Of course, if $A$ is not $\dagger$-Frobenius, then the unit of the algebra may not be equal to $e^\dagger$.
\end{exa}

%% file: sections/cart.tex
\section{Cartesian Monoidal Categories: A Nice Context for Parameterization}
\label{Sect:CartMon}
The goal of this section is to understand nice contexts for parameterization, such as the examples studied in \cref{Sect:Motivation}.
To this end, let $\mathcal{V}^{\mathcal{W}}$ be a braided monoidal category.
If $\mathcal{V}$ is a nice context for parameterization, then it would be reasonable to assume that each object $P \in \mathcal{V}_0$ can be used as the parameter object in at least one way.
This means that there must exists at least one comonoid $M_P = ( P, d_p, e_P )$ in $\mathcal{V}$ from which to obtain a parameterization $\Omega_{(\mathcal{V},\mathcal{W})}( M_P )$.
The comonoid associated with $\mathbb{I}$ should be $( \mathbb{I}, \lambda^{-1}, \mathbb{I} )$, since this is the comonoid used to obtain the underlying $\mathcal{W}$-categories.
It would also be nice if this choice of comonoids respected the monoidal structure on $\mathcal{V}$, in the sense that $M_{P \otimes Q} = M_P \otimes M_Q$.
This makes sense since $\mathcal{V}$ is braided monoidal, from which a monoidal structure can be induced on $\mathbf{Comon}( \mathcal{V} )$.
Now, consider the problem of finding a reparameterization from $\Omega_{(\mathcal{V},\mathcal{W})}( M_Q )$ to $\Omega_{(\mathcal{V},\mathcal{W})}( M_P )$.
In general, this is hard since not every morphism $f \in \mathcal{V}( P, Q )$ lifts to comonoid homomorphism from $M_P$ to $M_Q$.
In a nice context for reparameterization, such a lifting would always exist.
These assumptions can be summarized by requiring the existence of a monoidal section to the forgetful functor from $\mathbf{Comon}( \mathcal{V} )$ to $\mathcal{V}$.

\begin{defi}
    A braided monoidal category $\mathcal{V}$ is a \emph{nice context for parameterization} if there exists a strict monoidal section to the forgetful functor $U : \mathbf{Comon}( \mathcal{V} ) \to \mathcal{V}$.
\end{defi}

\begin{lem}
    \label{Lem:NiceParam}
    A braided monoidal category $\mathcal{V}$ is a nice context for parameterization if and only if the following conditions hold.
    \begin{enumerate}
    \item Each object $P \in \mathcal{V}_0$ admits a comonoid $M_P$.
    \item If $P \in \mathcal{V}_0$ and $Q \in \mathcal{V}_0$, then $M_{P \otimes Q} = M_P \otimes M_Q$.
    \item If $f \in \mathcal{V}( P, Q )$, then $f$ is a comonoid homomorphism from $M_P$ to $M_Q$.
    \item $M_{\mathbb{I}} = ( \mathbb{I}, \lambda^{-1}_{\mathbb{I}}, 1_{\mathbb{I}} )$.
    \end{enumerate}
\end{lem}

\begin{proof}
    First, assume that $\mathcal{V}$ is a nice context for parameterization.
    Then there exists a strict monoidal section $S: \mathcal{V} \to \mathbf{Comon}( \mathcal{V} )$ to the forgetful functor $U: \mathbf{Comon}( \mathcal{V} ) \to \mathcal{V}$.
    For each $P \in \mathcal{V}_0$, define $M_P = S_0( P )$.
    The four properties are satisfied as follows.
    \begin{enumerate}
    \item Let $P \in \mathcal{V}_0$.
          Since $S$ is a section to $U$, then $U_0( M_P ) = U_0( S_0( P ) ) = P$.
          This means that $M_P$ is a comonoid on $P$.
    \item Let $P \in \mathcal{V}_0$ and $Q \in \mathcal{V}_0$.
          Since $S$ is a strict monoidal functor, then by definition of a strict tensorator $M_{P \otimes Q} = S_0( P \otimes Q ) = S_0( P ) \otimes S_0( Q ) = M_P \otimes M_Q$.
    \item Let $f \in \mathcal{V}( P, Q )$.
          Then $S( f )$ is a comonoid homomorphism from $M_P$ to $M_Q$.
          This means that $U( S( f ) )$ lifts to a comonoid homomorphism from $M_P$ to $M_Q$.
          Since $S$ is a section to $U$, then $U( S( f ) ) = f$.
          Then $f$ is a comonoid homomorphism from $M_P$ to $M_Q$.
    \item Since $S$ is strict monoidal, then $M_{\mathbb{I}} = S_0( \mathbb{I} ) = ( \mathbb{I}, \lambda^{-1}_{\mathbb{I}}, 1_{\mathbb{I}} )$.
    \end{enumerate}
    This concludes the first direction of the proof.
    Next, assume that $\mathcal{V}$ is a braided monoidal category satisfying properties (1) through to (4).
    By property (1), there exists an assignment $S_0: \mathcal{V}_0 \to \mathbf{Comon}( \mathcal{V} )_0$ such that $S_0: P \mapsto M_P$.
    By property (3), this extends to an assignment on morphisms, such that $S_{P,Q}: f \mapsto f$.
    Let $S: \mathcal{V} \to \mathbf{Comon}( \mathcal{V} )$ denote the full assignment from $\mathcal{V}$ to $\mathbf{Comon}( \mathcal{V} )$.
    This is clearly a functor since function composition is the same in $\mathbf{Comon}( \mathcal{V} )$ as in $\mathcal{V}$.
    It can then be shown that $S$ is a section to $U: \mathbf{Comon}( \mathcal{V} ) \to \mathcal{V}$.
    \begin{itemize}
    \item \textbf{Objects}.
          Let $P \in \mathcal{V}_0$.
          Since $M_P$ is a comonoid on $P$, then $U_0( S_0( P ) ) = U_0( M_P ) = P$.
    \item \textbf{Morphisms}.
          If $f \in \mathcal{V}( P, Q )$, then $U( S( f ) ) = U( f ) = f$.
    \end{itemize}
    Then $U \circ S = 1_{\mathcal{V}}$.
    It remains to be shown that $S$ is a strict monoidal functor.
    \begin{itemize}
    \item \textbf{Preserves Products}.
          If $f \in \mathcal{V}( P, Q )$ and $g \in \mathcal{V}( R, S )$, then $S( f \otimes g ) = f \otimes g = S( f ) \otimes S( g )$.
    \item \textbf{Strict Coherence}.
          We first show that $S$ preserves the left unitor.
          If $\lambda$ is the left unitor for $\mathcal{V}$, then $\lambda$ is also the left unitor for $\mathbf{Comon}( \mathcal{V} )$ and $S( \lambda_P ) = \lambda_P$ for each $P \in \mathcal{V}_0$.
          A similar argument holds for the right unitor of $\mathcal{V}$ and the associator of $\mathcal{V}$.
    \item \textbf{Strict Tensorator}.
          By property (2), $S( P \otimes Q ) = M_{P \otimes Q} = M_P \otimes M_Q = S( P ) \otimes S( Q )$ for each $P \in \mathcal{V}_0$ and $Q \in \mathcal{V}_0$.
    \item \textbf{Strict Unitor}.
          By property (4), $S( \mathbb{I} ) = M_{\mathbb{I}} = ( \mathbb{I}, \lambda^{-1}_{\mathbb{I}}, 1_{\mathbb{I}} )$.
    \end{itemize}
    Then $S$ is a strict monoidal section to the forgetful functor.
\end{proof}

It turns out that a symmetric monoidal category $\mathcal{V}$ is a nice context for parameterization if and only if $\mathcal{V}$ is Cartesian monoidal.
This can be seen as a variation on Fox's theorem.
To prove the first direction, it suffices to prove that the strictness of $S: \mathcal{V} \to \mathbf{Comon}( \mathcal{V} )$ ensures that $\mathbb{I}$ is terminal, and that the comonoids chosen by $S$ can be used to define uniform copying and deleting.
To prove the second direction, it suffices to show that the uniform copying and deleting exhibited by Cartesian monodial categories pick out comonoids satisfying the properties in \cref{Lem:NiceParam}.

\begin{thm}
    \label{Thm:NiceParam}
    A symmetric monoidal category $\mathcal{V}$ is a nice context for parameterization if and only if $\mathcal{V}$ is Cartesian monoidal category.
\end{thm}

\begin{proof}
    First, assume that $\mathcal{V}$ is a nice context for parameterization.
    Then there exist a strict monoidal section $S: \mathcal{V} \to \mathbf{Comon}( \mathcal{V} )$ to the forgetful functor $U: \mathbf{Comon}( \mathcal{V} ) \to \mathcal{V}$.
    It must first be shown that $\mathbb{I}$ is terminal.
    To this end, let $P \in \mathcal{V}$ and $f \in \mathcal{V}( P, \mathbb{I} )$.
    Since $S$ is strict monoidal, then $S_0( \mathbb{I} ) = ( \mathbb{I}, \lambda_{\mathbb{I}}^{-1}, 1_{\mathbb{I}} )$.
    Then $f = U( S( f ) )$ is a comonoid homomorphism from $S_0( P ) = ( P, d, e )$ to $S_0( \mathbb{I} ) = ( \mathbb{I}, \lambda_{\mathbb{I}}^{-1}, 1_{\mathbb{I}} )$.
    Then by definition of a comonoid homomorphism, $e = 1_{\mathbb{I}} \circ f = f$.
    Since $f$ was arbitrary, then $\mathcal{V}( P, \mathbb{I} ) = \{ e \}$.
    Since $P$ was arbitrary, then $\mathbb{I}$ is terminal.
    Next, it must be shown that $\mathcal{V}$ has uniform copying and uniform deleting.
    \begin{itemize}
    \item \textbf{Uniform Copying}.
          For each $P \in \mathcal{V}_0$, let $\Delta_P$ denote the comultiplication in $S_0( P )$.
          Then $\Delta: \mathcal{V} \Rightarrow \mathcal{V} \otimes \mathcal{V}$ defines an unnatural transformation.
          To show that $\Delta$ is natural, let $f \in \mathcal{V}( P, Q )$ for some $P \in \mathcal{V}_0$ and $Q \in \mathcal{V}_0$.
          Then $S( f )$ is a comonoid homomorphism.
          In particular, $\Delta_Q \circ f = ( f \otimes f ) \circ \Delta_P$.
          Since $f$ was arbitrary, then $\Delta$ is natural.
          Since $S$ is a strict monoidal functor, then $\Delta_\mathbb{I} = \lambda_{\mathbb{I}}^{-1}$.
          It remains to be shown that $\Delta$ defines uniform copying on $\mathcal{V}$.
          Let $P \in \mathcal{V}_0$ and $Q \in \mathcal{V}_0$.
          Since $S$ is a strict monoidal functor, then $S( P \otimes Q ) = S( P ) \otimes S( Q )$.
          Since $\Delta_{P \otimes Q}$ is the comultiplication for $S( P ) \otimes S( Q )$, then \cref{Eqn:MonoidProd} holds in $\mathcal{V}$.
          Since $P$ and $Q$ were arbitrary, then $\Delta$ defines uniform copying on $\mathcal{V}$.
    \item \textbf{Uniform Deleting}.
          Since $\mathbb{I}$ is terminal, then there exists a unique natural transformation $\epsilon: \mathcal{V} \Rightarrow \mathbb{I}$.
          Moreover, if $P \in \mathcal{V}_0$ and $( P, d, e ) = S_0( P )$, then $\epsilon_P = e$ by uniqueness.
    \end{itemize}
    By construction, $( P, \Delta_P, \epsilon_P ) = S_0( P )$ is a comonoid for each $P \in \mathcal{V}_0$.
    Then $\mathcal{V}$ is Cartesian monoidal by \cref{Prop:Cartesian}.
    This concludes the first direction of the proof.
    Next, assume that $\mathcal{V}$ is Cartesian monoidal.
    Then $\mathcal{V}$ admits uniform copying $\Delta: \mathcal{V} \Rightarrow \mathcal{V} \otimes \mathcal{V}$ and uniform deleting $\epsilon: \mathcal{V} \Rightarrow \mathbb{I}$.
    Since $\Delta$ and $\epsilon$ are compatible, then each object $P \in \mathcal{V}_0$ admits a comonoid $( P, \Delta_P, \epsilon_P )$.
    These comonoids satisfy the properties of \cref{Lem:NiceParam}.
    \begin{enumerate}
    \item For each $P \in \mathcal{V}_0$, define $M_P = ( P, \Delta_P, \epsilon_P )$.
    \item Let $P \in \mathcal{V}_0$ and $Q \in \mathcal{V}_0$.
          Since $\mathcal{V}$ is Cartesian monoidal, then $\mathbb{I}$ is terminal.
          It follows by uniqueness that $\epsilon_{P \otimes Q} = \lambda_{\mathbb{I}}^{-1} \circ ( \epsilon_P \otimes \epsilon_Q )$.
          Moreover, since $\Delta$ defines uniform copying, then \cref{Eqn:MonoidProd} holds in $\mathcal{V}$.
          Then by definition, $M_{P \otimes Q} = M_P \otimes M_Q$.
    \item Let $f \in \mathcal{V}( P, Q )$.
          Since $\Delta$ is a natural transformation, then $( f \otimes f ) \circ \Delta_P = \Delta_Q \circ f$.
          Since $\mathcal{V}$ is Cartesian monoidal, then $\mathbb{I}$ is terminal.
          Then $\epsilon_P = \epsilon_Q \circ f$ by uniqueness.
          Then $f$ is a comonoid homomorphism from $M_P$ to $M_Q$.
    \item By definition, $\Delta_{\mathbb{I}} = \lambda_{\mathbb{I}}^{-1}$ and $\epsilon_{\mathbb{I}} = 1_\mathbb{I}$.
          Then $M_{\mathbb{I}} = ( \mathbb{I}, \lambda_{\mathbb{I}}^{-1}, 1_{\mathbb{I}} )$.
    \end{enumerate}
    Then $\mathcal{V}$ is a nice context for parameterization by \cref{Lem:NiceParam}.
\end{proof}

It turns out that even if $\mathcal{V}$ is a nice context for parameterization, it may not be the case that evaluation is possilbe.
It was shown in~\cref{Ex:NoParams} that every meet-semilattice with at least two elements defines a Cartesian monoidal category in which parameter objects never have elements.
In contrast, topological spaces define a category in which parameters always have elements.
Similar, the category of sub-lattices is also Cartesian monoidal and admits evaluation, giving rise to interesting examples from classical program analysis.
It should be noted that this is not the standard monoidal closed structure on \textbf{SupLat}, since \textbf{SupLat} is not closed as a Cartesian monoidal category.

\begin{exa}[Parameterized Rotations]
    \label{Ex:CartRot}
    Since $\mathbf{Top}$ is a Cartesian monoidal category, then by \cref{Thm:NiceParam} the category $\mathbf{Top}$ is a nice context for parameterization.
    This means that the categories described in the motivating example inherit many nice properties from $\mathbf{Top}$.
    To this end, let $\mathcal{V} = \mathbf{Top}$, $\mathcal{W} = \mathbf{Set}$, and $\mathcal{C}^{\mathcal{V}} = \mathbb{C}\mathbf{FVect}^{\mathcal{V}}$.
    Then for each $P \in \mathbf{Top}$, the category $\mathcal{C}_P := \Omega_{(\mathcal{V},\mathcal{W})}( M_P )( \mathcal{C}^{\mathcal{V}} )$ is the ordinary category whose objects are vector spaces with morphisms from $V$ to $W$ given by the continuous maps from $P$ into $\mathcal{C}^{\mathcal{V}}( V, W )$.
    Since $\mathcal{V}$ is a nice context for parameterization, then all moprhisms in $\mathcal{V}$ induce reparameterizations.
    In particular, the structural maps which make $\mathcal{V}$ Cartesian will also induce reparameterizations.
    \begin{itemize}
    \item \textbf{Symmetries}.
          Let $P = X \times Y$, $Q = Y \times X$, and $\sigma_{X,Y}$ denote the symmetry from $P$ to $Q$ which sends each tuple $( x, y )$ to $( y, x )$.
          Then $F = \Omega_{(\mathcal{V},\mathcal{W})}( \sigma_{X,Y} )( \mathcal{C} ): \mathcal{C}_Q \to \mathcal{C}_P$ is the functor which maps each morphsim $f \in \mathcal{C}_Q( V, W )$ to the morphism $( x, y ) \mapsto f( y, x )$ in $\mathcal{C}_P( V, W )$.
          This provides a framework to discuss whether a morphism is symmetric in its parameters.
          For example, if $X = Y$, then $f( x, y ) = f( y, x )$ if and only if $f = F( f )$.
    \item \textbf{Projections}.
          Let $P = X \times Y$.
          Then there exists a projection map $\pi_X: P \to X$ such that $\pi_X( x, y ) = x$.
          Then $F = \Omega_{(\mathcal{V},\mathcal{W})}( \pi_X )( \mathcal{C} ): \mathcal{C}_X \to \mathcal{C}_P$ is the functor which maps each morphism $f \in \mathcal{C}_X( V, W )$ to the morphism $( x, y ) \to f( x )$ in $\mathcal{C}_P( V, W )$.
          In other words, $F$ picks out the subcategory of $\mathcal{C}_P$ consisting of the parameterized morphisms in $\mathcal{C}_P$ which do not depend on $Y$.
    \item \textbf{Copying}.
          Let $X \in \mathcal{V}_0$ and $P = X \times X$.
          Then there exists a copying map $\Delta_X: X \to P$ such that $\Delta_X( x ) = ( x, x )$.
          Then $\Omega_{(\mathcal{V},\mathcal{W})}(\Delta_X )( \mathcal{C} ): \mathcal{C}_P \to \mathcal{C}_X$ is the functor which maps each morphsim $f \in \mathcal{C}_P( V, W )$ to the morphism $x \mapsto f( x, x )$ in $\mathcal{C}_X( V, W )$.
          In the case of parameterized quantum circuits, this provides a framework to interpret circuits in which a rotation gate depends on the same parameter more than once.
    \end{itemize}
    It should be noted that reparameterizations already appear in the existing work on parameterized quantum circuit analysis, though these results were not stated in the language of enriched category theory.
    For example, the authors of~\cite{RossWesley2025} consider circuits parameterized by $P = \mathbb{R}^k$.
    These circuits are then precomposed with a linear automorphism $f: P \to P$ to reduce parameterized equivalence checking with rational coefficients to parameterized equivalence checking with integer coefficients.
    Since $f$ is a linear map between finite-dimensional vector space, then clearly $f$ is continuous.
    This means that the reparameterization induced by $f$ is indeed the isomorphism $\Omega_{(\mathcal{V},\mathcal{W})}( f )( \mathcal{C} ): \mathcal{C}_P \to \mathcal{C}_P$.
\end{exa}

\begin{exa}[Discrete-Time Processes]
    Let $\mathcal{P}: \mathbf{Set} \to \mathbf{SupLat}$ denote the covariant power-set functor.
    It was shown in~\cite{Little2025}, that non-determinism in program semantics could be understood in terms of \textbf{SupLat}-enrichment.
    For example, assume that the semantics of a deterministic programming language are given by some locally small category $\mathcal{C}$.
    This means that $\mathcal{C}( X, Y )$ is the set of deterministic programs which convert inputs of type $X$ into outputs of type $Y$.
    Given two deterministic programs $f \in \mathcal{C}( X, Y )$ and $g \in \mathcal{C}( X, Y )$, a non-deterministic program is allowed to choose between $f$ and $g$, independent of the input data~\cite{10.1145/321420.321422}.
    Non-determinism is often used as a programming abstraction to model interactions with an external environment, such as requesting an input from a user.
    As described in~\cite{Little2025}, non-determinism can be introduced to $\mathcal{C}$ via change-of-base through the power-set functor.
    The interpretation of $\mathcal{D} = \mathcal{P}_*( \mathcal{C} )$ is as follows.
    \begin{itemize}
    \item \textbf{Objects}.
          Since $\mathcal{D}_0 = \mathcal{C}_0$, then the collection of types remains unchanged.
    \item \textbf{Morphisms}.
          If $\widetilde{f} \in \mathcal{D}( X, Y )$, then $\widetilde{f} \subseteq \mathcal{C}( X, Y )$.
          This means that $\widetilde{f}$ is a set of deterministic programs.
          Each $f \in \widetilde{f}$ is a possible outcome when $\widetilde{f}$ is executed.
          For example, $\widetilde{f} = \{ f, g \}$ is the non-deterministic program which chooses between executing either $f$ or $g$.
    \item \textbf{Determinism}.
          There is an inclusion $j: \mathcal{C}( X, Y ) \hookrightarrow \mathcal{D}( X, Y )$ such that $j: f \mapsto \{ f \}$.
          This inclusion maps each program in $\mathcal{C}( X, Y )$ to a singleton set, and can be seen as picking out the deterministic programs in $\mathcal{D}( X, Y )$.
    \item \textbf{Composition}.
          If $\widetilde{f} \in \mathcal{D}( X, Y )$ and $\widetilde{g} \in \mathcal{D}( Y, Z )$, then $\widetilde{g} \circ \widetilde{f} = \{ g \circ f : f \in \widetilde{f} \land g \in \widetilde{g} \}$.
          Then the composition of $\widetilde{f}$ with $\widetilde{g}$ is composition of all possible outcomes from $\widetilde{f}$ with all possible outcomes from $\widetilde{g}$.
    \end{itemize}
    An abstraction of discrete-time systems can then be obtained by parameterizing $\mathcal{D}$ by the extended natural numbers $L = \mathbb{N} \cup \{ \infty \}$.
    To see how this works, let $F = \Omega_{(\mathbf{SupLat},\mathbf{Set})}( L )$ and $\mathcal{D}_L = F( \mathcal{D} )$.
    Then each $\widetilde{f} \in \mathcal{D}_L( X, Y )$ is a monotonic function from $L$ to $\mathcal{D}( X, Y )$.
    Each $f \in \widetilde{f}( n )$ is a possible outcome from executing $\widetilde{f}$ at some time between $0$ and $n$.
    This can be used to model behaviours such as error recovery, as illustrated in the following toy model.
    In this example, assume that $a \in \mathcal{D}( X, Y )$ is the desired behaviour of a system, and $b \in \mathcal{D}( X, y )$ is a possible error.
    Further assume that $c \in \mathcal{D}( Y, Y )$ is an error-recovery protocol, such that $c \circ a = c \circ b = a$.
    We will model the system by some $\widetilde{f} \in \mathcal{D}_L( X, Y )$ such that $\widetilde{f}( n ) \subseteq \{ a, b \}$ for all $n \in L$, $\widetilde{f}( t ) = \{ a, b \}$ for some $t \in \mathbb{N}$.
    This means that $\widetilde{f}$ will fail after some unknown time $t$.
    If we assume that assume that $c$ is fault-tolerant (i.e., it will never fail), then we may lift $c$ to $\widetilde{g} = j( c ) \circ \epsilon_L$.
    That is, $\widetilde{g}( n ) = \{ c \}$ for all $n \in L$.
    Since $( \widetilde{g} \star \widetilde{f} )( n ) = \widetilde{g}( n ) \circ \widetilde{f}( n )$, then it can be proven that $( \widetilde{g} \star \widetilde{f} )( n ) = j( a ) = \{ a \}$ for all $n \in L$.
    In other words, the process $\widetilde{g} \star \widetilde{f}$ can recover from all of the faults encountered by $\widetilde{f}$.
\end{exa}

%% file: sections/closed.tex
\section{Monoidal Closed Enrichment and Parameterized Constructions}
\label{Sect:Closed}

Recall that if $\mathcal{V}$ is monoidal closed, then $\mathcal{V}$ is self-enriched.
In particular, the hom-tensor adjunction provides a canonical choice for the $\mathcal{V}$-enriched identity and composition maps.
Using these canonical choices, there then exists a canonical parameterizations of the form $\Omega_{(\mathcal{V},\mathcal{V})}( - )$.
In some sense, $\Omega_{(\mathcal{V},\mathcal{V})}( - )$ is the \emph{best} possible parameterization for a self-enriched category.
For example, \textbf{Top} can be viewed as a self-enriched category using the discrete topology, but this construction fails to satisfy the hom-tensor adjunction.
In contrast, when \textbf{Top} is restricted to nice topological spaces, it is possible to view \textbf{Top} as self-enriched without sacrificing the topological data encoded by the function spaces.
This can be used to further refine the semantics of parameterized quantum circuits, as illustrated in~\cref{Ex:NiceTop}.

\begin{exa}[Continuous Rotations]
    \label{Ex:NiceTop}
    Let $\mathcal{V}$ be the category of compactly generated topological spaces and continuous maps.
    This category is a nice context for parameterization, since $\mathcal{V}$ is a Cartesian closed category~(see~\cite{Steenrod1967}).
    As shown in \cref{Sect:CartMon}, the Cartesian monoidal structure on $\mathcal{V}$ ensures that each object $P \in \mathcal{V}_0$ admits a unique comonoid $M_P$.
    Since $[0,2\pi]$, $[0,2\pi] / \{ 0 = 2\pi \}$, and $\mathbb{R}^k$ are compactly generated topological spaces, then this subsumes the topological comonoids discussed thus far.
    Moreover, since $\mathcal{V}$ is closed, then $\mathcal{V}$ is self-enriched.
    This means that there exists a parameterization function $F = \Omega_{( \mathcal{V}, \mathcal{V} )}( M_P )$ for each $P \in \mathcal{V}_0$.
    Next, notice that the enrichment of $\mathcal{C} = \mathbb{C}\mathbf{FVect}$ over \textbf{Top} also makes sense over $\mathcal{V}$, since every vector space is a metric space, and every metric space is a compactly generated topological space.
    Then for each $P \in \mathcal{V}_0$, $\mathcal{C}_P = \Omega_{(\mathcal{V},\mathcal{V})}( M_P )( \mathcal{C} )$ is a well-defined $\mathcal{V}$-enriched category.
    An immediate consequence of this construction is that each hom-object $\mathcal{C}_P( V, W )$ now enjoys a compact-open topology.
    A more interesting consequence is that each pair of objects $P$ and $Q$ in $\mathcal{V}_0$ admit an evaluation map $\app_{P,Q}: [ Q, P ] \otimes Q \to P$.
    Since $\mathcal{V}$ is Cartesian monoidal, then $\app_{P,Q}$ is also a comonoid homomorphism.
    This induces a reparameterization $F = \Omega_{( \mathcal{V}, \mathcal{V} )}( \app_{P,Q} )( \mathcal{C} ): \mathcal{C}_P \to \mathcal{C}_R$ where $R = [ Q, P ] \otimes Q$.
    While this reparameterization may look complicated upon first glance, it recovers a fairly intuitive notion.
    In particular, if $f \in \mathcal{C}_P( V, W )$, then $F( f ) \in \mathcal{C}_R( V, W )$ is the parameterized function which takes a pair $( \varphi, q )$ and then evaluated $f$ at $\varphi( q )$.
    In other words, $F( f )( \varphi, q ) = f( \varphi( q ) )$.
\end{exa}

Of course, $\mathcal{V}$ need not be Cartesian closed for for this parameterization to be well-defined.
For each $M \in \mathbf{Comon}( \mathcal{V} )_0$, let $\Lambda_M( \mathcal{V} )$ denote the category $F_*( \mathcal{V} )$ where $F = \Omega_{(\mathcal{V},\mathcal{V})}( M )( \mathcal{V} )$.
When $\mathcal{V}$ is Cartesian closed, the only comonoids are the diagonal comonoids.
In contrast, a strictly non-monoidal category such as $\mathbb{C}\mathbf{FVect}$ can admit many interesting comonoids on the same object, such as the linear copying comonoid and the pair-of-pants comonoid.
We saw in \cref{Ex:GenCtrl} that the monoidal closed structure on $\mathcal{V} = \mathbb{C}\mathbf{FVect}$ allowed for the parameterized morphisms in $\Lambda_M( \mathcal{V} )$ to be rewritten as generalized controls in $\mathcal{V}$.
This construction did not depend on the Frobenius algebra associated with $M$, and can be generalized to any parameterization of the form $\Lambda_M( \mathcal{V} )$.
From a categorical point-of-view, this construction can be described as a generally non-monoidal functor $\Xi_M: \Lambda_M( \mathcal{V} ) \to \mathcal{V}$.
This means that each comonoid $M \in \mathbf{Comon}( \mathcal{V} )_0$ defines an internal construction in $\mathcal{V}$.

\begin{nota}
    In this section, $\mathcal{V}$ will denote a monoidal closed category with monoidal product $\otimes$, $\mathcal{U}$ a subcategory of $\mathcal{V}$, and $M = ( P, d, e ) \in \mathbf{Comon}( \mathcal{V} )_0$.
\end{nota}

\begin{defi}
    \label{Def:ClosedIncl}
    The \emph{parameterized $M$-construction} is the assignment $\Xi_M: \Lambda_M( \mathcal{V} ) \to \mathcal{V}$ with $\left( \Xi_M \right)_0( X ) = P \otimes X$ and $\left( \Xi_M \right)_{X,Y}: \mathcal{V}( P, [ X, Y ] ) \to \mathcal{V}( P \otimes X, P \otimes Y )$ defined as follows.
    \begin{equation*}
        \left( \Xi_M \right)_{X,Y}:
        f
        \mapsto
        ( 1_P \otimes \app_{X,Y} ) \circ ( 1_P \otimes f \otimes 1_X ) \circ ( d \otimes 1_X ) =
        \includegraphics[valign=c,scale=0.84]{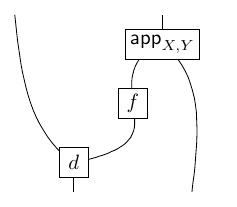}
    \end{equation*}
\end{defi}

\begin{lem}
    \label{Lem:XiFunc}
    If $f \in \mathcal{V}( X, Y )$, then $\Xi_M( j( f ) ) = 1_P \otimes f$.
\end{lem}

\begin{proof}
    Let $f \in \mathcal{V}( X, Y )$.
    Then the following equation holds.
    \begin{equation*}
        \Xi_M( j( f ) )
        =
        \includegraphics[valign=c,scale=0.84]{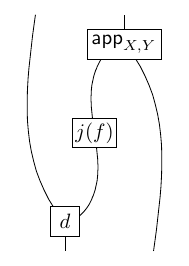}
        =
        \includegraphics[valign=c,scale=0.84]{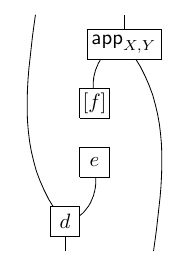}
        =
        \includegraphics[valign=c,scale=0.84]{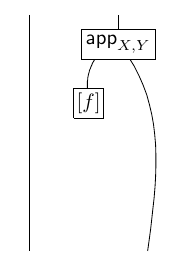}
        =
        \includegraphics[valign=c,scale=0.84]{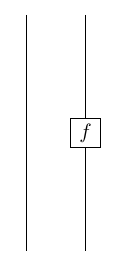}
        =
        1_P \otimes f
    \end{equation*}
    Since $f$ was arbitrary, then $\Xi_M( j( f ) ) = 1_P \otimes f$ for each $f \in \mathcal{V}( X, Y )$.
\end{proof}

\begin{thm}
    \label{Thm:XiFunc}
    The parameterized $M$-construction is a functor.
\end{thm}

\begin{proof}
    It must be shown that $\Xi_M$ preserves identities and composition.
    \begin{itemize}
    \item \textbf{Identities}.
          Let $X \in \mathcal{V}_0$.
          The $\Xi_M( i_X ) = \Xi_M( j( 1_X ) ) = 1_P \otimes 1_X = 1_{P \otimes X}$  by \cref{Lem:XiFunc}.
          Since $( \Xi_M )_0( X ) = P \otimes X$, then $\Xi_M$ maps the identity on $X$ to the identity on $( \Xi_M )_0( X )$.
          Since $X$ was arbitrary, then $\Xi_M$ preserves identities.
    \item \textbf{Composition}.
          Let $f \in \Lambda_M( \mathcal{V} )( X, Y )$ and $g \in \Lambda_M( \mathcal{V} )( Y, Z )$.
          Then the following holds.
          \begin{equation*}
            \Xi_M( g \star f )
            =
            \includegraphics[valign=c,scale=0.84]{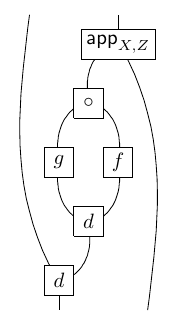}
            =
            \includegraphics[valign=c,scale=0.84]{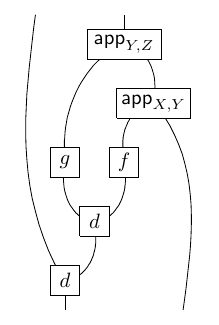}
            =
            \includegraphics[valign=c,scale=0.84]{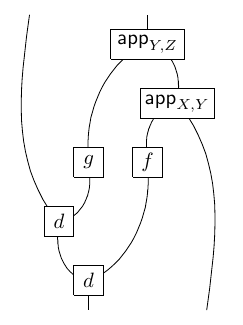}
          \end{equation*}
          It then follows by a deformation of the right-hand side that $\Xi_M( g \star f ) = \Xi_M( g ) \circ \Xi_M( f )$.
          Since $f$ and $g$ were arbitrary, then $\Xi_M$ respects composition.
          \endproof
    \end{itemize}
\end{proof}

The functor $\Xi_M$ enjoys many algebraic properties which are usually associated with controlled gates in quantum computation.
For example, applying $\Xi_M$ to the image of $j$ is the same as tensoring with $P$, since the parameter $P$ is unused.
Of particular interest is the fact that if $\mathcal{V}$ is symmetric monoidal, then $\Xi_M( f \boxtimes g )$ can always be rewritten in terms of $\Xi_M( f )$ and $\Xi_M( g )$, even though $\Xi_M( f \boxtimes g )$ is typically not monoidal.
These algebraic properties are summarized by the following two theorems.
All relations are depicted in \cref{Fig:XiCircs}.

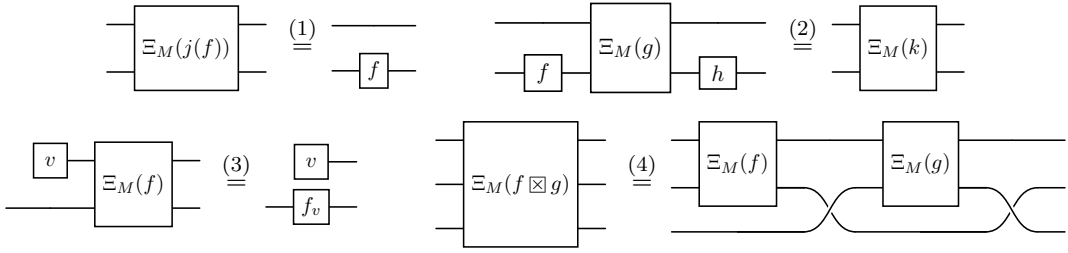
\begin{figure}[t]
    \begin{equation*}
        \scalebox{0.73}{\input{circs/xij_lhs}}
        \overset{(1)}{=}
        \scalebox{0.73}{\input{circs/xij_rhs}}
        \qquad
        \scalebox{0.77}{\input{circs/xicomm_lhs}}
        \overset{(2)}{=}
        \scalebox{0.73}{\input{circs/xicomm_rhs}}
    \end{equation*}
    \begin{equation*}
        \scalebox{0.73}{\input{circs/xicopy_lhs}}
        \overset{(3)}{=}
        \scalebox{0.73}{\input{circs/xicopy_rhs}}
        \qquad
        \scalebox{0.73}{\input{circs/xitensor_lhs}}
        \overset{(4)}{=}
        \scalebox{0.73}{\input{circs/xitensor_rhs}}
    \end{equation*}
    \caption{Circuit equations satisfied by $\Xi_M$. Relations (1) through to (3) correspond to the properties in \cref{Thm:XiProps}, whereas relation (4) is \cref{Thm:XiTensor}.}
    \label{Fig:XiCircs}
\end{figure}

\begin{thm}
    \label{Thm:XiProps}
    The parameterized $M$-construction satisfy the following properties.
    \begin{enumerate}
    \item $\Xi_M \circ j = P \otimes ( - )$.
    \item If $f \in \mathcal{V}( X, Y )$, $g \in \Lambda_M( \mathcal{V} )( Y, Z )$, $h \in \mathcal{V}( Z, W )$, and $k \in \Lambda_M( \mathcal{V} )( X, W )$ satisfy the equation $j( h ) \star g \star j( f ) = k$, then $( 1_P \otimes h ) \circ \Xi_M( g ) \circ ( 1_P \otimes f ) = \Xi_M( k )$.
    \item If $d$ copies $v \in \mathcal{V}( \mathbb{I}, P )$ and $f \in \Lambda_M( \mathcal{V} )( X, Y )$, then $\Xi_M( f ) \circ ( v \otimes 1_X ) = v \otimes f_v$ where $f_v = \app_{X,Y} \circ ( ( f \circ v ) \otimes 1_X ) \circ \rho_X^{-1}$.
    \end{enumerate}
\end{thm}

\begin{proof}
    There are three cases to consider.
    \begin{enumerate}
    \item By definition, $( \Xi_M \circ j )_0 ( X ) = \Xi_M( X )_0 = P \otimes X$ for each $X \in \mathcal{V}_0$.
          Then $\Xi_M \circ j$ and $P \otimes ( - )$ agree on objects.
          It remains to be shown that $\Xi_M \circ j$ and $P \otimes ( - )$ agree on morphisms.
          Let $f \in \mathcal{V}( X, Y )$.
          Then $\Xi_M( j( f ) ) = 1_P \otimes f = P \otimes ( f )$ by \cref{Lem:XiFunc}.
          Since $f$ was arbitrary, then $\Xi_M \circ j = P \otimes ( - )$.
    \item The following equation holds.
          {\small\begin{align*}
              ( 1_P \otimes h ) \circ \Xi_M( g ) \circ ( 1_P \otimes f )
              &=
              \Xi_M( j( h ) ) \circ \Xi_M( g ) \circ \Xi_M( j( f ) )
              && \emph{(by part (1))}
              \\
              &=
              \Xi_M( j( h ) \star g \star j( f ) )
              && \emph{(by \cref{Thm:XiFunc})}
              \\
              &= \Xi_M( k )
              && \emph{(by assumption)}
          \end{align*}}
    \item Since $d$ copies $v$, then the following equation holds.
          \begin{equation*}
            \includegraphics[valign=c,scale=0.84]{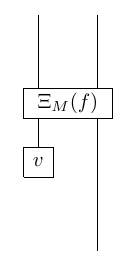}
            =
            \includegraphics[valign=c,scale=0.84]{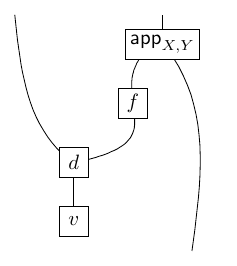}
            =
            \includegraphics[valign=c,scale=0.84]{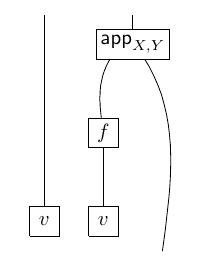}
          \end{equation*}
          Then $\Xi_M( f ) \circ ( v \otimes 1_X ) = v \otimes f_v$ where $f_v = \app_{X,Y} \circ ( ( f \circ v ) \otimes 1_X ) \circ \rho_X^{-1}$.
          \endproof
    \end{enumerate}
\end{proof}

\begin{thm}
    \label{Thm:XiTensor}
    Assume that $\mathcal{V}$ admits a symmetry $\beta$ and $M$ is cocommutative.
    If $f \in \Lambda_M( \mathcal{V} )( X, X' )$ and $g \in \Lambda_M( \mathcal{V} )( Y, Y' )$, then $\Xi_M( f \boxtimes g ) = ( 1_P \otimes \beta_{X',Y'} ) \circ ( \Xi_M( g ) \otimes 1_{X'} ) \circ ( 1_P \otimes \beta_{X',Y} ) \circ ( \Xi_M( f ) \otimes 1_Y )$.
\end{thm}

\begin{proof}
    We first prove the special case in which $g = i_Y$.
    Let $f \in F_*( \mathcal{V} )( X, X' )$ and $Y \in \mathcal{V}_0$.
    Then the following equation holds.
    \begin{equation*}
        \Xi_M( f \boxtimes i_Y )
        =
        \includegraphics[valign=c,scale=0.84]{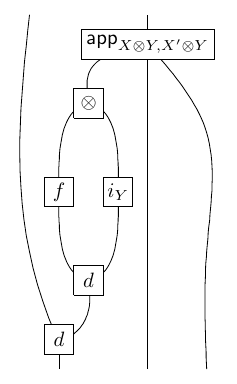}
        =
        \includegraphics[valign=c,scale=0.84]{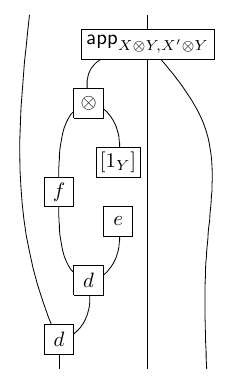}
        =
        \includegraphics[valign=c,scale=0.84]{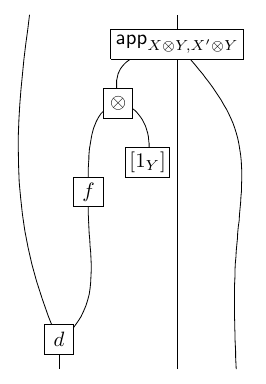}
    \end{equation*}
    The enriched tensor product can then be expanded as follows.
    \begin{equation*}
        \Xi_M( f \boxtimes i_Y )
        =
        \includegraphics[valign=c,scale=0.84]{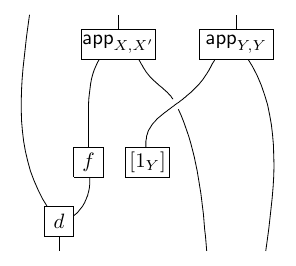}
        =
        \includegraphics[valign=c,scale=0.84]{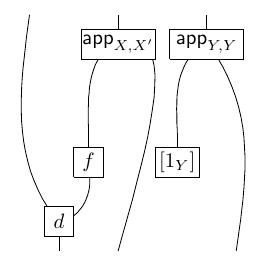}
        =
        \includegraphics[valign=c,scale=0.84]{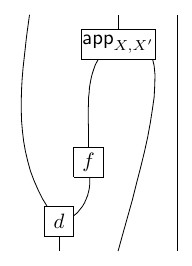}
    \end{equation*}
    Then $\Xi_M( f \boxtimes i_Y ) = \Xi_M( f ) \otimes 1_Y$.
    Since $f$ and $Y$ were arbitrary, then this special case of the theorem is established.
    Using the special case, it is then possible to prove the general case.
    Let $f \in \Lambda_M( \mathcal{V} )( X, X' )$ and $g \in \Lambda_M( \mathcal{V} )( Y, Y' )$.
    Then the following equation holds.
    {\par\small\begin{align*}
        \Xi_M( f \boxtimes g )
        &= \Xi_M( ( i_{X'} \boxtimes g ) \star ( f \boxtimes i_Y ) )
        \\
        &= \Xi_M( b_{X',Y'}^{-1} \star ( g \boxtimes i_{X'} ) \star b_{X',Y} \star ( f \boxtimes i_Y ) )
        \\
        &= \Xi_M( b_{X',Y'}^{-1} ) \circ \Xi_M( g \boxtimes i_{X'} ) \circ \Xi_M( b_{X',Y}) \circ \Xi_M( f \boxtimes i_Y )
            && \textit{(by \cref{Thm:XiFunc})}
            \\
        &= \Xi_M( b_{X',Y'}^{-1} ) \circ ( \Xi_M( g ) \otimes 1_{X'} ) \circ \Xi_M( b_{X',Y} ) \circ ( \Xi_M( f ) \otimes 1_Y )
            && \textit{(by special case)}
            \\
        &= \Xi_M( j( \beta_{X',Y'}^{-1} ) ) \circ ( \Xi_M( g ) \otimes 1_{X'} ) \circ \Xi_M( j( \beta_{X',Y} ) ) \circ ( \Xi_M( f ) \otimes 1_Y )
        \\
        &= ( 1_P \otimes \beta_{X',Y'}^{-1} ) \circ ( \Xi_M( g ) \otimes 1_{X'} ) \circ ( 1_P \otimes \beta_{X',Y} ) \circ ( \Xi_M( f ) \otimes 1_Y )
            && \textit{(by \cref{Thm:XiProps} (1))}
    \end{align*}}%
    Since $f$ and $g$ were arbitrary, then this concludes the proof.
\end{proof}

\subsection{Parameterized Constructions and Control Functors}
The circuit diagrams obtains from \cref{Thm:XiProps} and \cref{Thm:XiTensor} are of particular interest, as they seemingly recover many of the properties expected of controlled operations, despite lacking the Frobenius structure required by controlled operations.
This is best illustrated by appeal to the axioms of a control functor.
For example, equation (2) can be seen as a generalization of the \textbf{Conjugation Axiom} satisfied by a $\dagger$-control functor, equation (3) can be seen as a generalization of the \textbf{Activation Axiom} satisfied by a pointed control functor, and equation (4) can be seen as a generalization of the \textbf{Tensor Axiom} satisfied by any control functor.
Of course, $\Xi_M$ is not a control functor, since $\Xi_M$ is not an endofunctor.
However, $\Xi_M$ can be composed with any other functor $G: \mathcal{V} \to \Lambda_M( \mathcal{V} )$ to obtain an endofunctor on $\mathcal{V}$.
It is not hard to see that $\Xi_M$ always composes with $j$ to obtain the trivial control functor $P \otimes ( - ) = \Xi_M \circ j: \mathcal{V} \to \mathcal{V}$.
In contrast, it is not possible to extend the standard control functor used in quantum computing to all of $\mathbb{C}\mathbf{Vect}$.

This motivates the study of comomoids in a symmetric monoidal closed category $\mathcal{V}$ which induce nontrivial endofunctors on subcategories of $\mathcal{V}$.
It should be noted that this generalization to subcategories is already present in the literature on control functors.
For example, the control functors considered in~\cite{Delorme2026} are restricted to the full subcategory on endomorphisms, and the control functions in~\cite{FuKishida2022} are defined with respect to a \emph{controlled subcategory} of endomorphisms.
To better understand the control functors obtained through parameterized constructions, we will first study the endofunctors that factor through $\Lambda_M( \mathcal{V} )$.
First, it is shown that these are precisely the endofunctors which give rise to $M$-comodule homomorphisms, in the same sense as standard control functors in \cref{Ex:GenCtrl}.

\begin{defi}
    A functor $F: \mathcal{U} \to \mathcal{U}$ is \emph{$M$-compatible} if for each $f \in \mathcal{U}( X, Y )$, the morphism $F( f )$ is an $M$-comodule homomorphism from $d \otimes 1_{X}$ to $d \otimes 1_{Y}$.
\end{defi}

\begin{lem}
    \label{Thm:CompDef}
    If $F: \mathcal{U} \to \mathcal{U}$ is an endofunctor and there exists a functor $G: \mathcal{U} \to \Lambda_M( \mathcal{V} )$ satisfying $F = \Xi_M \circ G$, then the following equation holds holds for each $f \in \mathcal{U}( X, Y )$.
    \begin{equation*}
        \label{Eqn:CompDef}
        \includegraphics[valign=c,scale=0.84]{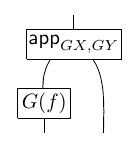}
        =
        \includegraphics[valign=c,scale=0.84]{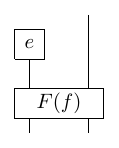}
    \end{equation*}
\end{lem}

\begin{proof}
    Let $f \in \mathcal{U}( X, Y )$.
    Since $F( f ) = \Xi_M( G( f ) )$, then the following equation holds.
    \begin{equation*}
        \includegraphics[valign=c,scale=0.84]{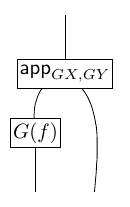}
        =
        \includegraphics[valign=c,scale=0.84]{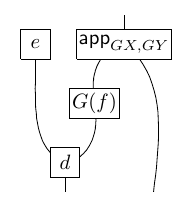}
        =
        \includegraphics[valign=c,scale=0.84]{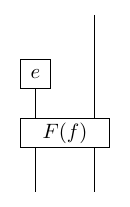}
    \end{equation*}
    Since $f$ was arbitrary, then this concludes the proof.
    \endproof
\end{proof}

\begin{thm}
    \label{Thm:CompFunc}
    A functor $F: \mathcal{U} \to \mathcal{U}$ is $M$-compatible if and only if there exists an identity-on-objects functor $G: \mathcal{U} \to \Lambda_M( \mathcal{V} )$ satisfying $F = \Xi_M \circ G$.
\end{thm}

\begin{proof}
    There are two cases to consider.
    \begin{enumerate}[align=left]
    \item[$\Leftarrow$]
          Assume there exists an identity-on-objects functor $G: \mathcal{U} \to \Lambda_M( \mathcal{V} )$ such that $F = \Xi_M \circ G$.
          Let $f \in \mathcal{U}( X, Y )$.
          Since $F( f ) = \Xi_M( G( f ) )$, then the following equation holds.
          \begin{equation*}
            \includegraphics[valign=c,scale=0.84]{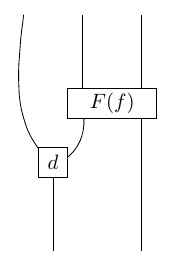}
            =
            \includegraphics[valign=c,scale=0.84]{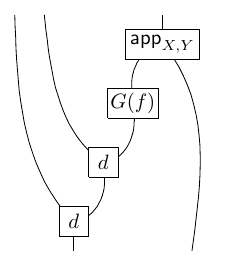}
            =
            \includegraphics[valign=c,scale=0.84]{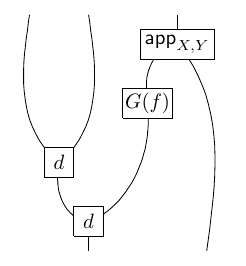}
            =
            \includegraphics[valign=c,scale=0.84]{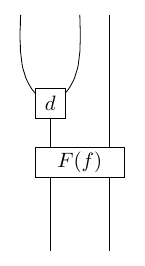}
          \end{equation*}
          Then $F( f )$ is a comodule homomorphism from $d \otimes 1_X$ to $d \otimes 1_Y$.
          Since $f$ was arbitrary, then $F$ is an $M$-compatible functor.
    \item[$\Rightarrow$]
          Assume that $F$ is $M$-compatible.
          Then the following equation holds for each $f \in \mathcal{U}( X, Y )$.
          \begin{equation}
            \includegraphics[valign=c,scale=0.84]{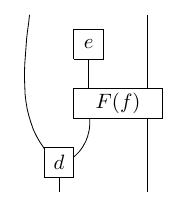}
            =
            \includegraphics[valign=c,scale=0.84]{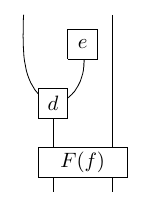}
            =
            \includegraphics[valign=c,scale=0.84]{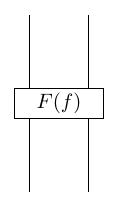}
            \label{Eqn:GCopyDel}
          \end{equation}
          Define an identity-on-objects assignment $G: \mathcal{U} \to \Lambda_M( \mathcal{V} )$ such that $G$ sends each morphism $f \in \mathcal{U}( X, Y )$ to the image of $\rho_X \circ ( e \circ 1_X ) \circ F( f )$ under the hom-tensor adjunction.
          It must be shown that $G$ defines an identity-on-objects functor.
          It suffices to show that $G$ preserves identities and composition.
          First, let $X \in \mathcal{V}_0$.
          Since $F$ is a functor, then the following equation holds.
          \begin{equation*}
            \includegraphics[valign=c,scale=0.84]{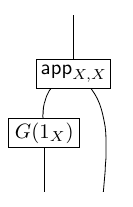}
            =
            \includegraphics[valign=c,scale=0.84]{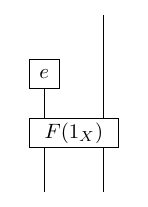}
            =
            \includegraphics[valign=c,scale=0.84]{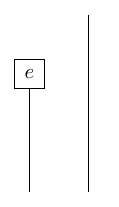}
            =
            \includegraphics[valign=c,scale=0.84]{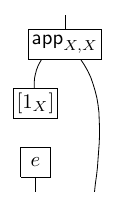}
            =
            \includegraphics[valign=c,scale=0.84]{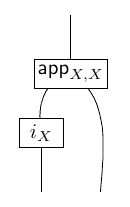}
          \end{equation*}
          Then $G( 1_X ) = i_X$ by uniqueness.
          Since $X$ was arbitrary, then $G$ preserves identities.
          Next, let $f \in \mathcal{U}( X, Y )$ and $g \in \mathcal{U}( Y, Z )$.
          By \cref{Eqn:GCopyDel}, the following equation holds.
          \begin{equation*}
            \includegraphics[valign=c,scale=0.84]{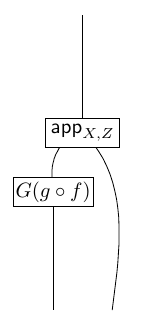}
            =
            \includegraphics[valign=c,scale=0.84]{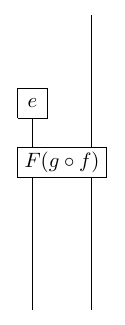}
            =
            \includegraphics[valign=c,scale=0.84]{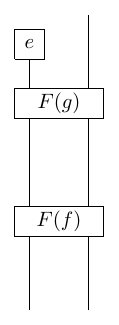}
            =
            \includegraphics[valign=c,scale=0.84]{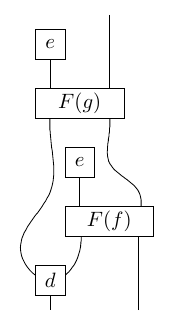}
          \end{equation*}
          It then follows from the definitions of $G( f )$ and $G( g )$ that the following equation holds.
          \begin{equation*}
            \includegraphics[valign=c,scale=0.84]{figs/special/g_mult_4.pdf}
            =
            \includegraphics[valign=c,scale=0.84]{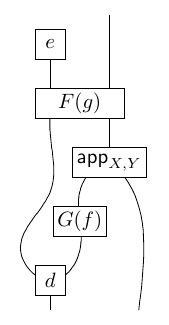}
            =
            \includegraphics[valign=c,scale=0.84]{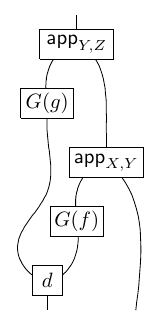}
            =
            \includegraphics[valign=c,scale=0.84]{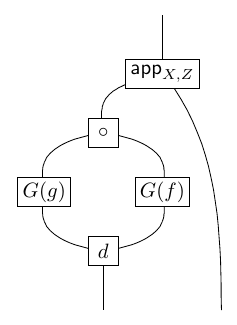}
          \end{equation*}
          Then $G( g \circ f ) = G( g ) \star G( f )$ by uniqueness.
          Since $f$ and $g$ were arbitrary, then $G$ respects composition.
          Then $G$ is a functor.
          Next, it must be shown that $F = \Xi_M \circ G$.
          To do this, let $f \in \mathcal{U}( X, Y )$.
          Then the following equation holds by \cref{Eqn:GCopyDel}.
          \begin{equation*}
            \includegraphics[valign=c,scale=0.84]{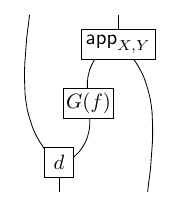}
            =
            \includegraphics[valign=c,scale=0.84]{figs/special/cstruct_lhs.pdf}
            =
            \includegraphics[valign=c,scale=0.84]{figs/special/cstruct_rhs.pdf}
          \end{equation*}
          Then $\Xi_M( G( f ) ) = F( f )$.
          Since $f$ was arbitrary, then $F = \Xi_M \circ G$.
          \endproof
    \end{enumerate}
\end{proof}

Now assume that $F: \mathcal{U} \to \mathcal{U}$ factors through $\Lambda_M( \mathcal{V} )$ by $G: \mathcal{U} \to \Lambda_M( \mathcal{V} )$.
The properties that $F$ must satisfy to be a control functor are of one of two flavours.
\begin{itemize}
\item \textbf{Conjugation Action}.
      $F( h \circ g \circ f ) = ( 1_P \otimes h ) \circ F( g ) \circ ( 1_P \otimes f )$.
\item \textbf{Right Tensor Action}.
      $F( f \otimes g ) = F( f ) \otimes g$.
\end{itemize}
For example, the \textbf{Conjugation Axiom} is a conjugation action and the \textbf{Tensor Axiom} is a left tensor action.
It turns out that both actions can be rephrased in terms of analogous actions on $G$.
These definitions are equivalent, as shown in the following two theorems.

\begin{thm}
    \label{Thm:ConjAction}
    Let $G: \mathcal{U} \to \Lambda_M( \mathcal{V} )$ be an identity-on-objects functor and $F = \Xi_M \circ G$.
    For each $f \in \mathcal{U}( X, Y )$, $g \in \mathcal{U}( Y, Z )$, and $h \in \mathcal{U}( Z, W )$, $G( h \circ g \circ f ) = j( h ) \star G( g ) \star j( f )$ if and only if $F( h \circ g \circ f ) = ( 1_P \otimes h ) \circ F( g ) \circ ( 1_P \otimes f )$.
\end{thm}

\begin{proof}
    Let $f \in \mathcal{U}( X, Y )$, $g \in \mathcal{U}( Y, Z )$, and $h \in \mathcal{U}( Z, W )$.
    If $G( h \circ g \circ f ) = j( h ) \star G( g ) \star j( f )$, then $F( h \circ g \circ f ) = ( 1_P \otimes h ) \circ F( g ) \circ ( 1_P \otimes f )$ by part (2) of \cref{Thm:XiProps}.
    Assume instead that $F( h \circ g \circ f ) = ( 1_P \otimes h ) \circ F( g ) \circ ( 1_P \otimes f )$.
    Then the following equation holds by \cref{Thm:CompDef}.
    \begin{center}
        \includegraphics[valign=c,scale=0.84]{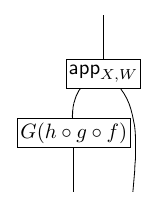}
        =
        \includegraphics[valign=c,scale=0.84]{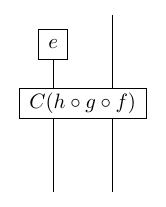}
        =
        \includegraphics[valign=c,scale=0.84]{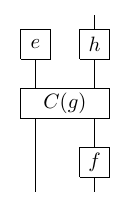}
        =
        \includegraphics[valign=c,scale=0.84]{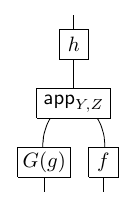}
        =
        \includegraphics[valign=c,scale=0.84]{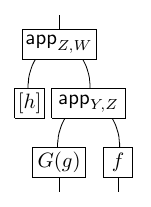}
    \end{center}
    The enriched composition can then be introduced as follows.
    \begin{equation*}
        \includegraphics[valign=c,scale=0.84]{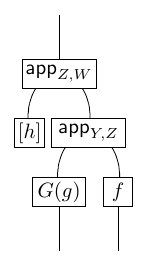}
        =
        \includegraphics[valign=c,scale=0.84]{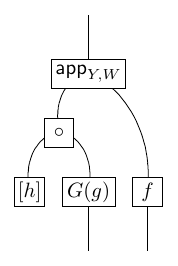}
        =
        \includegraphics[valign=c,scale=0.84]{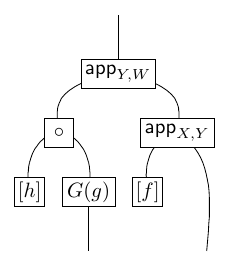}
        =
        \includegraphics[valign=c,scale=0.84]{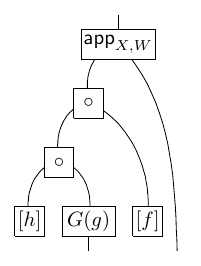}
    \end{equation*}
    Counits can then be introduced as followed.
    \begin{equation*}
        \includegraphics[valign=c,scale=0.84]{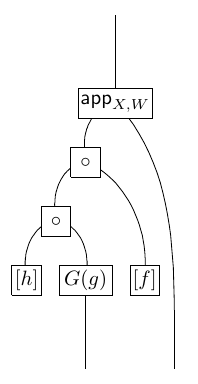}
        =
        \includegraphics[valign=c,scale=0.84]{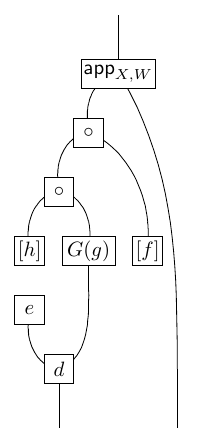}
        =
        \includegraphics[valign=c,scale=0.84]{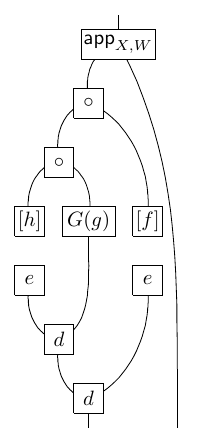}
        =
        \includegraphics[valign=c,scale=0.84]{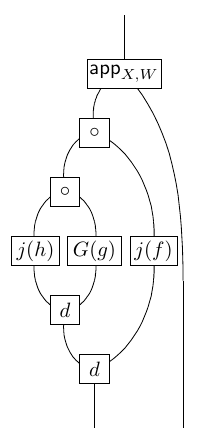}
    \end{equation*}
    Then $G( h \circ g \circ f ) = j( h ) \star G( g ) \star j( f )$ by the uniqueness of $G( h \circ g \circ f )$.
\end{proof}

\begin{thm}
    \label{Thm:TensorAction}
    Assume that $\mathcal{V}$ admits a symmetry, $\mathcal{U}$ is a monoidal subcategory, and $M$ is cocommutative.
    Let $G: \mathcal{U} \to \Lambda_M( \mathcal{V} )$ be an identity-on-objects functor and $F = \Xi_M \circ G$.
    For each $f \in \mathcal{U}( X, Y )$ and $g \in \mathcal{U}( Z, W )$, $G( f \otimes g ) = G( f ) \boxtimes j( g )$ if and only if $F( f \otimes g ) = F( f ) \otimes g$.
\end{thm}

\begin{proof}
    If $G( f \otimes g ) = G( f ) \boxtimes j( g )$, then the following circuit equation holds.
    \begin{align*}
        \input{circs/comptensor_1}
        &=
        \input{circs/comptensor_2}
        &&
        \emph{(by \cref{Thm:XiTensor})}
        \\
        &=
        \input{circs/comptensor_3}
        &&
        \emph{(by \cref{Thm:XiProps} part (1))}
        \\
        &=
        \input{circs/comptensor_4}
        =
        \input{circs/comptensor_5}
    \end{align*}
    Assume instead that $F( f \otimes g ) = F( f ) \otimes g$.
    Then the following equation holds by \cref{Thm:CompDef}.
    \begin{equation*}
        \includegraphics[valign=c,scale=0.84]{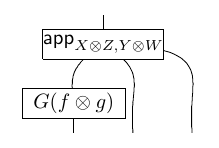}
        =
        \includegraphics[valign=c,scale=0.84]{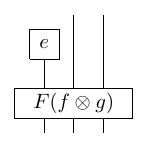}
        =
        \includegraphics[valign=c,scale=0.84]{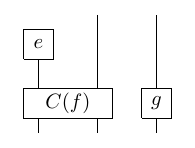}
        =
        \includegraphics[valign=c,scale=0.84]{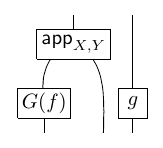}
    \end{equation*}
    Then by braiding with $\mathbb{I}$, the following equality holds.
    \begin{equation*}
        \includegraphics[valign=c,scale=0.84]{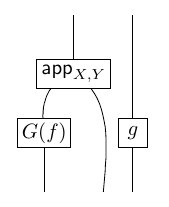}
        =
        \includegraphics[valign=c,scale=0.84]{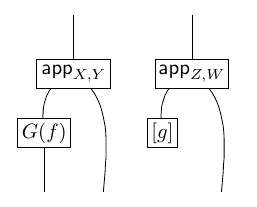}
        =
        \includegraphics[valign=c,scale=0.84]{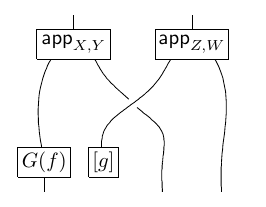}
    \end{equation*}
    The enriched tensor product can then be introduced as follows.
    \begin{equation*}
        \includegraphics[valign=c,scale=0.84]{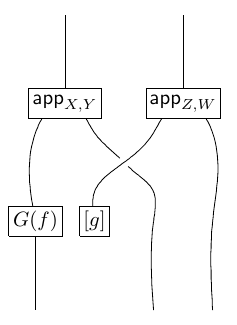}
        =
        \includegraphics[valign=c,scale=0.84]{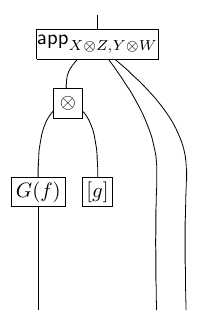}
        =
        \includegraphics[valign=c,scale=0.84]{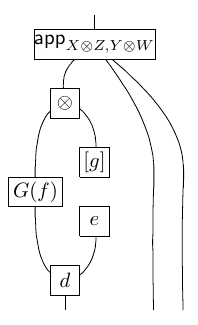}
        =
        \includegraphics[valign=c,scale=0.84]{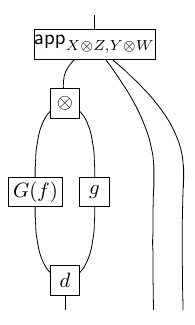}
    \end{equation*}
    Then $G( f \otimes 1_Y ) = G( f ) \boxtimes i_Y$ by the uniqueness of $G( f \otimes 1_Y )$.
\end{proof}

It follows from \cref{Thm:CompFunc} that a control functor factors through $\Lambda_M( \mathcal{V} )$ if and only if it is $M$-compatible.
Using \cref{Thm:ConjAction} and \cref{Thm:TensorAction}, it is then possible to characterize $M$-compatible control functors in terms of their factorizations.
This leads to the notion of an $M$-control structure, as defined in \cref{Def:ControlStruct}.
It is shown in \cref{Thm:ControlStructs} that every $M$-control structure induces an $M$-compatible control functor.
Conversely, \cref{Thm:CompatibleControls} shows that every $M$-compatible control functor arises from a control structure.
The trivial control functor and the standard control functors both arise in this way.

\begin{defi}
    \label{Def:ControlStruct}
    Assume that $\mathcal{V}$ admits a symmetry $\beta$, $\mathcal{U}$ is a symmetric monoidal subcategory of $\mathcal{V}_{\mathbf{endo}}$, and $M$ is cocommutative.
    An \emph{$M$-control structure} is a functor $F: \mathcal{U} \to \Lambda_M( \mathcal{V} )$ such that the following properties hold.
    \begin{enumerate}
    \item The image of $\Xi_M \circ F$ is contained in $\mathcal{U}$.
    \item $F$ is an identity-on-objects functor.
    \item $F( f \otimes 1_Y ) = F( f ) \boxtimes i_Y$ for each $f \in \mathcal{U}( X, X )$ and $Y \in \mathcal{U}_0$.
    \item $F( ( 1_X \otimes \beta_{Y,Z}^{-1} \otimes 1_W ) \circ f \circ ( 1_X \otimes \beta_{Y,Z} \otimes 1_W ) ) = ( i_X \boxtimes b_{Y,Z}^{-1} \boxtimes i_W ) \star F( f ) \star ( i_X \boxtimes b_{Y,Z} \otimes i_W )$ for all $X, Y, Z, W \in \mathcal{U}_0$ and $f \in \mathcal{U}( X \otimes Y \otimes Z \otimes W, X \otimes Y \otimes Z \otimes W )$.
    \item If $C = \Xi_M \circ G$, then $C( C( f ) ) \circ ( \beta_{P,P} \otimes 1_X ) = ( \beta_{P,P} \otimes 1_Y ) \circ C( C( f ) )$ for each $f \in \mathcal{U}( X, X )$.
          In other words, the following equation of circuits is valid.
        \begin{center}
            \input{circs/ctrl_lhs}
            =
            \input{circs/ctrl_rhs}
        \end{center}
    \end{enumerate}
\end{defi}

\begin{thm}
    \label{Thm:ControlStructs}
    Assume that $\mathcal{V}$ admits a symmetry $\beta$, $\mathcal{U}$ is a symmetric monoidal subcategory of $\mathcal{V}_{\mathbf{endo}}$, and $M$ is cocommutative.
    If $F: \mathcal{U} \to \Lambda_M( \mathcal{V} )$ is an $M$-control structure, then $\Xi_M \circ F$ is a control functor.
\end{thm}

\begin{proof}
    Assume that $F: \mathcal{U} \to \Lambda_M( \mathcal{V} )$ is a control structure and let $C = \Xi_M \circ F$.
    Then $C$ is an endofunctor on $\mathcal{U}$ by property (1) of a control structure.
    It remains to be shown that $C$ is a control functor on $\mathcal{U}$.
    \begin{enumerate}
    \item Let $X \in \mathcal{U}_0$.
          Then $F_0( X ) = X$ by property (2) of a control structure.
          It follows that $C_0( X ) = ( \Xi_M )_0( F_0( X ) ) = ( \Xi_M )_0( X ) = P \otimes X$ by part (1) of \cref{Thm:XiProps}.
          Since $X$ was arbitrary, then $C_0( X ) = P \otimes X$ for all $X \in \mathcal{U}_0$.
    \item Let $f \in \mathcal{U}( X, X )$ and $Y \in \mathcal{U}_0$.
          Then $F( f \otimes 1_Y ) = F( f ) \boxtimes i_Y$ by property (3) of a control structure.
          It follows by \cref{Thm:ConjAction} that $C( f \otimes 1_Y ) = C( f ) \otimes 1_Y$.
          Since $f$ and $Y$ were arbitrary, then $C( f \otimes 1_Y ) = C( f ) \otimes 1_Y$ for each $f \in \mathcal{U}( X, X )$ and $Y \in \mathcal{U}_0$.
    \item Let $X, Y, Z, W \in \mathcal{U}_0$ and $f \in \mathcal{U}( X \otimes Y \otimes Z \otimes W, X \otimes Y \otimes Z \otimes W )$.
          Then define $s = i_X \boxtimes b_{Y,Z} \boxtimes i_W$ which satisfies $j( s ) = 1_X \otimes \beta_{Y,Z} \otimes 1_W$.
          It follows by property (4) of a control structure that $F( j( s^{-1} ) \circ f \circ j( s ) ) = s^{-1} \star F( f ) \star s$.
          Then by \cref{Thm:TensorAction}, it follows that $C( j( s^{-1} ) \circ f \circ j( s ) ) = ( 1_P \otimes j( s^{-1} ) ) \circ C( f ) \circ ( 1_P \otimes j( s ) )$.
          Since $f$ was arbitrary, then $C( ( 1_X \otimes \beta_{Y,Z} \otimes 1_W ) \circ f \circ ( 1_X \otimes \beta_{Y,Z} \otimes 1_W ) ) = ( 1_{P \otimes X} \otimes \beta_{Y,Z} \otimes 1_W ) \circ C( f ) \circ ( 1_{P \otimes X} \otimes \beta_{Y,Z} \otimes 1_W )$ for each $X, Y, Z, W \in \mathcal{U}_0$ and $f \in \mathcal{U}( X \otimes Y \otimes Z \otimes W, X \otimes Y \otimes Z \otimes W )$.
    \item If $X \in \mathcal{U}_0$ and $f \in \mathcal{U}( X, X )$, then $C( C( f ) ) \circ ( \beta_{P,P} \otimes 1_X ) = ( \beta_{P,P} \otimes 1_Y ) \circ C( C( f ) )$ by property (5) of a control functor.
    \end{enumerate}
    In conclusion, $C$ is a control functor.
\end{proof}

\begin{thm}
    \label{Thm:CompatibleControls}
    Assume that $\mathcal{V}$ admits a symmetry $\beta$, $\mathcal{U}$ is a symmetric monoidal subcategory of $\mathcal{V}_{\mathbf{endo}}$, and $M$ is cocommutative.
    If $C: \mathcal{U} \to \mathcal{U}$ is a control functor, then there exists an $M$-control structure $F$ such that $C = \Xi_M \circ F$ if and only if $C$ is $M$-compatible.
\end{thm}

\begin{proof}
    If there exists an $M$-control structure $F: \mathcal{U} \to \Lambda_M( \mathcal{V} )$ such that $C = \Xi_M \circ F$, then $C$ is $M$-compatible by \cref{Thm:CompFunc}.
    Assume instead that $C$ is $M$-compatible.
    Then by \cref{Thm:CompFunc}, there exists an identity-on-objects functor $F: \mathcal{U} \to \Lambda_M( \mathcal{V} )$ such that $C = \Xi_M \circ F$.
    It remains to be shown that $F$ is a control structure.
    \begin{enumerate}
    \item Since $C = \Xi_M \circ F$, then the image of $\Xi_M \circ F$ is contained in $\mathcal{U}$.
    \item By definition, $F_0( X ) = X$ for each $X \in \mathcal{U}_0$.
    \item Let $f \in \mathcal{U}( X, X )$ and $Y \in \mathcal{U}_0$.
          Then $C( f \otimes 1_Y ) = C( f ) \otimes 1_Y$ by the \textbf{Tensor Axiom} of a control functor.
          It follows by \cref{Thm:ConjAction} that $F( f \otimes 1_Y ) = F( f ) \boxtimes i_Y$.
          Since $f$ and $Y$ were arbitrary, then $F( f \otimes 1_Y ) = F( f ) \boxtimes 1_Y$ for each $f \in \mathcal{U}( X, X )$ and $Y \in \mathcal{U}_0$.
    \item Let $X, Y, Z, W \in \mathcal{U}_0$ and $f \in \mathcal{U}( X \otimes Y \otimes Z \otimes W, X \otimes Y \otimes Z \otimes W )$.
          Then define $s = i_X \boxtimes b_{Y,Z} \boxtimes i_W$ such that $j( s ) = 1_X \otimes \beta_{Y,Z} \otimes 1_W$.
          It follows by the \textbf{Target Symmetry Axiom} of a control functor that $C( j( s^{-1} ) \circ f \circ j( s ) ) = j( s^{-1} ) \circ C( f ) \circ j( s )$.
          Then by \cref{Thm:ConjAction}, it follows that $F( j( s^{-1} ) \circ f \circ j( s ) ) = s^{-1} \star F( f ) \star s$.
          Since $f$ was arbitrary, then $F( ( 1_X \otimes \beta_{Y,Z}^{-1} \otimes 1_W ) \circ f \circ ( 1_X \otimes \beta_{Y,Z} \otimes 1_W ) = ( i_X \boxtimes b_{Y,Z}^{-1} \boxtimes i_W ) \circ F( f ) \circ ( i_X \boxtimes b_{Y,Z} \boxtimes i_W )$ for each $X, Y, Z, W \in \mathcal{U}_0$ and $f \in \mathcal{U}( X \otimes Y \otimes Z \otimes W, X \otimes Y \otimes Z \otimes W )$. 
    \item If $X \in \mathcal{U}_0$ and $f \in \mathcal{U}( X, X )$, then $C( C( f ) ) \circ ( \beta_{P,P} \otimes 1_X ) = ( \beta_{P,P} \otimes 1_Y ) \circ C( C( f ) )$ by the \textbf{Control Symmetry Axiom} of a control functor.
    \end{enumerate}
    In conclusion, $F$ is a control structure.
\end{proof}

\begin{exa}[The Trivial Control Structure]
    \label{Ex:TrivCtrl}
    Assume $\mathcal{V}$ admits a symmetry $\beta$ and $M$ is cocommutative.
    It is easy to see that $j$ is an $M$-control structure.
    Let $C = \Xi_M \circ j$.
    \begin{enumerate}
    \item Clearly the image of $C$ is contained within $\mathcal{V}_{\mathbf{endo}}$ when $C$ is restricted to $\mathcal{V}_{\mathbf{endo}}$.
    \item By construction, $j$ is the identity on objects.
    \item $j( f \otimes 1_Y ) = j( f ) \boxtimes j( 1_Y ) = C( f ) \boxtimes i_Y$ for each $f \in \mathcal{V}( X, X )$ and $Y \in \mathcal{V}_0$.
    \item By a similar argument, $j$ satisfies property (4) of a control structure.
    \item $( \beta_{P,P} \otimes 1_X ) \circ C( C( f ) ) = \beta_{P,P} \otimes f = C( C( f ) ) \circ ( \beta_{P,P} \otimes 1_X )$ for each $f \in \mathcal{V}( X, X )$.
    \end{enumerate}
    Then by \cref{Thm:ControlStructs}, $P \otimes ( - ) = \Xi_M \circ j$ is a control functor.
    It follows from \cref{Thm:CompatibleControls} that $P \otimes ( - )$ is $M$-compatible.
\end{exa}

\begin{exa}[The Standard Control Structure]
    \label{Ex:StdCtrlStruct}
    Let $\mathcal{U}$ be the monoidal subcategory of unitary endomorphisms in $\mathbb{C}\mathbf{FVect}$.
    For each choice of orthonormal basis $\{ \ket{u}, \ket{v} \}$ in $P = \mathbb{C}^2$, there exists a control functor $C: \mathcal{U} \to \mathcal{U}$ such that $C( f ) \circ ( \ket{u} \otimes 1_X ) = \ket{u} \otimes 1_X$ and $C( f ) \circ ( \ket{v} \otimes 1_X ) = \ket{v} \otimes f$.
    That is, each $C( f )$ is the controlled version of $f$ which applies $f$ to the target qubits if and only if the control qubit is in state $\ket{v}$.
    As explained in~\cref{Prop:BasisComonoid}, there exists a comonoid $M = ( P, d, e )$ which copies the basis $\{ \ket{u}, \ket{v} \}$.
    It will be shown that $C( - )$ factors through $\Lambda_M( \mathcal{V} )$ via an $M$-control structure.
    As explained in \cref{Thm:CompatibleControls}, it suffices to show that $C$ is $M$-compatible.
    Then let $f \in \mathcal{U}( X, X )$.
    By linearity, it suffices to show that $( 1_P \otimes C( f ) ) \circ ( d \otimes 1_X ) \circ ( \ket{u} \otimes 1_X ) = ( d \otimes 1_X ) \circ C( f ) \circ ( \ket{u} \otimes 1_X )$ and also $( 1_P \otimes C( f ) ) \circ ( d \otimes 1_X ) \circ ( \ket{v} \otimes 1_X ) = ( d \otimes 1_X ) \circ C( f ) \circ ( \ket{v} \otimes 1_X )$.
    The first equation follows from the following sequence of equalities.
    \begin{equation*}
        \input{circs/compatible_u_1}
        =
        \input{circs/compatible_u_2}
        =
        \input{circs/compatible_u_3}
        =
        \input{circs/compatible_u_4}
        =
        \input{circs/compatible_u_5}
    \end{equation*}
    Then the second equation follows from the following sequence of equalities.
    \begin{equation*}
        \input{circs/compatible_v_1}
        =
        \input{circs/compatible_v_2}
        =
        \input{circs/compatible_v_3}
        =
        \input{circs/compatible_v_4}
        =
        \input{circs/compatible_v_5}
    \end{equation*}
    Then $C$ is $M$-compatible.
    Then by \cref{Thm:CompatibleControls}, there exists a control structure $F: \mathcal{U} \to \Lambda_M( \mathcal{V} )$ such that $C = \Xi_M \circ F$.
    Then by \cref{Thm:ControlStructs}, the functor $F$ is uniquely defined by the family of equations $F( f ) = [ 1_X ] \circ \bra{u} + [ f ] \circ \bra{v}$.
\end{exa}

\subsection{Parameterization and Pointed Control Functors}
It follows from the previous two examples that control structures provide a more restricted generalization of control functors, but still fail to rule out the example of the trivial control functor.
One of the key properties that separate the trivial control functor from the standard control functor is that the trivial control functor is not pointed.
This motivates the study of pointed control structures.
We saw in \cref{Thm:XiProps} that for each $f \in \Lambda_M( \mathcal{V} )( X, Y )$, if $d$ copies $v$, then $\Xi_M( f ) \circ ( v \otimes 1_X ) = v \otimes f_v$.
It turns out that an $M$-control structure $F: \mathcal{U} \to \Lambda_M( \mathcal{V} )$ can induce a $( u, v )$-pointed control functor, even if $u$ and $v$ are not copied by $d$.
However, it will always be the case that for each $f \in \mathcal{U}( X, X )$, the morphism $F( f ) \circ u$ is a scalar multiple of $[ 1_X ]$ and the morphism $F( g ) \circ v$ is a scalar multiple of $[ f ]$ as shown below in \cref{Thm:PointedStruct}.
In the case where both $e \circ u$ and $e \circ v$ are the identity, it follows that $F( f ) \circ u = [ 1_X ]$ and $F( f ) \circ v = [ f ]$.

\begin{nota}
    In this subsection, we further assume that $\mathcal{V}$ is symmetric, $\mathcal{U}$ is a symmetric monoidal subcategory of $\mathcal{V}_{\mathbf{endo}}$, and $M$ is cocommutative.
\end{nota}

\begin{thm}
    \label{Thm:PointedStruct}
    If $F: \mathcal{U} \to \Lambda_M( \mathcal{V} )$ is an $M$-control structure and $\Xi_M \circ F$ is a $( u, v )$-pointed control functor, then the following equations hold for each $f \in \mathcal{U}( X, X )$.
    
    \noindent
    \begin{minipage}{.48\linewidth}
    \begin{equation}
        \label{Eqn:UPoint}
        \includegraphics[valign=c,scale=0.84]{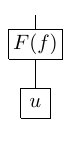}
        =
        \includegraphics[valign=c,scale=0.84]{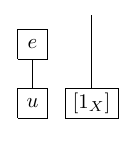}
    \end{equation}
    \end{minipage}%
    \hfill
    \begin{minipage}{.48\linewidth}
    \begin{equation}
        \label{Eqn:VPoint}
        \includegraphics[valign=c,scale=0.84]{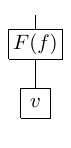}
        =
        \includegraphics[valign=c,scale=0.84]{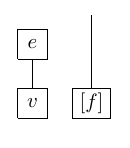}
    \end{equation}
    \end{minipage}

    \noindent
    Moreover, the following statements are equivalent.
    \begin{enumerate}
    \item $e \circ u = 1_{\mathbb{C}} = e \circ v$.
    \item $F( f ) \circ u = [ 1_X ]$ and $F( f ) \circ v = [ f ]$ for each $f \in \mathcal{U}( X, X )$.
    \end{enumerate}
\end{thm}

\begin{proof}
    Assume that $F: \mathcal{U} \to \Lambda_M( \mathcal{V} )$ is an $M$-control structure and $\Xi_M \circ F$ is a $( u, v )$-pointed control functor.
    Let $( \psi, g )$ denote either $( u, [ 1_X ] )$ or $( v, [ f ] )$.
    Since $u$ deactivates $F$ and $v$ activates $F$, then the following equation holds.
    \begin{equation*}
        \includegraphics[valign=c,scale=0.84]{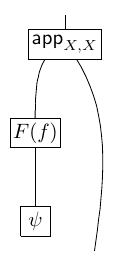}
        =
        \includegraphics[valign=c,scale=0.84]{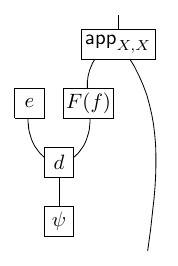}
        =
        \includegraphics[valign=c,scale=0.84]{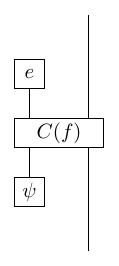}
        =
        \includegraphics[valign=c,scale=0.84]{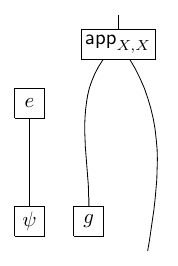}
    \end{equation*}
    Then \cref{Eqn:UPoint} and \cref{Eqn:VPoint} hold by uniqueness.
    Clearly, if $e \circ u = 1_{\mathbb{C}} = e \circ v$, then $F( f ) \circ u = [ 1_X ]$ and $F( f ) \circ v = [ f ]$ for each $f \in \mathcal{U}( X, X )$.
    Assume instead that $F( f ) \circ u = [ 1_X ]$ and $F( f ) \circ v = [ f ]$ for each $f \in \mathcal{U}( X, X )$.
    Then in particular, $F( 1_{\mathbb{C}} ) \circ u = [ 1_{\mathbb{C}} ]$ and $F( 1_{\mathbb{C}} ) \circ v = [ 1_{\mathbb{C}} ]$.
    Then the following equation holds for each $\psi \in \{ u, v \}$.
    \begin{align*}
        1_{\mathbb{C}}
        \;=
        \includegraphics[valign=c,scale=0.84]{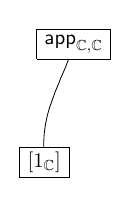}
        =\;
        \includegraphics[valign=c,scale=0.84]{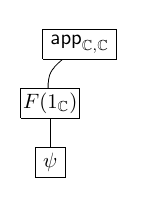}
        =
        \includegraphics[valign=c,scale=0.84]{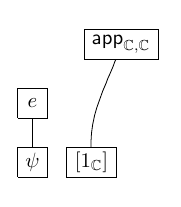}
        =
        e \circ \psi
    \end{align*}
    In conclusion, $e \circ u = 1_{\mathbb{C}} = e \circ v$.
\end{proof}

\begin{exa}[Points for the Standard Control Functor]
    As in \cref{Ex:StdCtrlStruct}, let $\mathcal{U}$ be the subcategory of unitary endomorphisms in $\mathbb{C}\mathbf{FVect}$, $\{ \ket{u}, \ket{v} \}$ be a choice of orthonormal basis for $P = \mathbb{C}^2$, $M$ be the comonoid which copies the basis $\{ \ket{u}, \ket{v} \}$, and $C: \mathcal{U} \to \mathcal{U}$ be the standard control functor with respect to the basis $\{ \ket{u}, \ket{v} \}$.
    As shown in \cref{Ex:StdCtrlStruct}, there exists a control structure $F: \mathcal{U} \to \Lambda_M( \mathcal{V} )$ such that $C = \Xi_M \circ F$.
    It is easily checked that $e \circ \ket{u} = 1 = e \circ \ket{v}$.
    Then by \cref{Thm:PointedStruct}, $F( f ) \circ u = [ 1_X ]$ and $F( f ) \circ v = [ f ]$ for each $f \in \mathcal{U}( X, X )$.
    This is easily verified as follows.
    \begin{align*}
        F( f ) \ket{u} &= [ 1_X ] \braket{u|u} + [ f ] \braket{v|u} = [ 1_X ]
        &
        F( f ) \ket{v} &= [ 1_X ] \braket{u|v} + [ f ] \braket{v|v} = [ f ]
    \end{align*}
\end{exa}

\begin{exa}[The Zero Point]
    If $\mathcal{V}$ has a zero object $\mathbb{O}$, then the notion of being $( u, v )$-pointed can be somewhat problematic.
    In particular, every control functor $F: \mathcal{U} \to \mathcal{U}$ is $( u, v )$ pointed with $u = 0$ and $v = 0$.
    In the case of $\mathbb{C}\mathbf{FVect}$, $u$ and $v$ would both be the linear map sending every vector to zero.
    If $F$ arises from a control structure $G: \mathcal{U} \to \Lambda_M( \mathcal{V} )$, then one could rule out such examples by asking that $F( f ) \circ u = [ 1_X ]$ and $F( f ) \circ v = [ f ]$ for each $f \in \mathcal{U}( X, X )$.
    As shown in \cref{Thm:PointedStruct}, this is true precisely when $e \circ u = 1_{\mathbb{C}} = e \circ v$.
    Moreover, these properties are satisfied by all standard control functors.
\end{exa}

\subsection{Parameterization and Conjugated Control Functors}
A commonly used property of the standard control functor is that it is conjugated.
It should be clear how to rewrite the conjugation axiom in terms of $\Lambda_M( \mathcal{V} )$ using~\cref{Thm:ConjAction}.
However, some care must be taken when generalizing from an endofunctor on $\mathcal{V}_{\mathbf{endo}}$ to an endofunctor on $\mathcal{U}$.
In particular, $\mathcal{V}_{\mathbf{endo}}$ is fixed under conjugation whereas this is not true for a general $\mathcal{U}$.
To this end, we first introduce the notion of a \emph{conjugated $M$-control structure}.
We then show in the maximal case, where $\mathcal{U} = \mathcal{V}_{\mathbf{endo}}$, these two notions coincide.

\begin{nota}
    In this subsection, we further assume that $\mathcal{V}$ is symmetric, $\mathcal{U}$ is a symmetric monoidal subcategory of $\mathcal{V}_{\mathbf{endo}}$, and $M$ is cocommutative.
\end{nota}

\begin{defi}
    If $f \in \mathcal{U}( Y, Y )$, then a \emph{conjugate factorization} of $f$ is a choice of morphisms $g \in \mathcal{V}( X, Y )$, $h \in \mathcal{U}( Y, Y )$, and $k \in \mathcal{V}( Y, X )$ satisfying $f = k \circ h \circ g$ and $k \circ g = 1_X$.
\end{defi}

\begin{defi}
    An $M$-control structure $F: \mathcal{U} \to \Lambda_M( \mathcal{V} )$ is \emph{conjugated} if for each conjugate factorization $( g, h, k )$ of each $f \in \mathcal{U}( X, X )$, the equation $F( f ) = j( k ) \star F( h ) \star j( g )$ holds.
\end{defi}

\begin{lem}
    \label{Lem:ConjFunctor}
    Let be $F: \mathcal{U} \to \Lambda_M( \mathcal{V} )$ a conjugated $M$-control structure and $C = \Xi_M \circ F$.
    If $( g, h, k )$ is a conjugate factorization of $f \in \mathcal{U}( X, X )$, then $C( f ) = ( 1_P \otimes k ) \circ C( h ) \circ ( 1_P \otimes g )$.
\end{lem}

\begin{proof}
    Assume that $( g, h, k )$ is a conjugate factorization of $f \in \mathcal{U}( X, X )$.
    This means that $F( k \circ h \circ g ) = F( f ) = j( k ) \star F( h ) \star F( g )$.
    Then $C( f ) = C( k \circ h \circ g ) = ( 1_P \otimes k ) \circ C( h ) \circ ( 1_P \otimes g )$ by~\cref{Thm:ConjAction}.
\end{proof}

\begin{thm}
    \label{Thm:ConjStruct}
    The functor $F: \mathcal{V}_{\mathbf{endo}} \to \Lambda_M( \mathcal{V} )$ is a conjugated $M$-control structure if and only if $\Xi_M \circ F$ is a conjugated control functor.
\end{thm}

\begin{proof}
    Let $C = \Xi_M \circ F$.
    There are two cases to consider.
    \begin{enumerate}[align=left]
    \item[$\Rightarrow$]
          Assume that $F$ is a conjugated $M$-control structure.
          Let $f \in \mathcal{C}( X, Y )$, $g \in \mathcal{C}( Y, Y )$, and $h \in \mathcal{C}( Y, X )$ satisfying $h \circ f = 1_X$.
          Then $( f, g, h )$ is a conjugate factorization of $k \circ h \circ g$.
          Then $C( h \circ h \circ f ) = ( 1_P \otimes h ) \circ C( h ) \circ ( 1_P \otimes f )$ by \cref{Lem:ConjFunctor}.
          Since $f$, $g$, and $h$ were arbitrary, then $C$ is a conjugated control functor.
    \item[$\Leftarrow$]
          Assume that $C$ is a conjugated control structure.
          Assume that $( g, h, k )$ is a conjugate factorization of $f \in \mathcal{U}( X, X )$.
          Then $C( k \circ h \circ g ) = ( 1_P \otimes k ) \circ C( h ) \circ C( 1_P \otimes g )$ by the conjugation axiom.
          Then $F( f ) = F( k \circ h \circ g ) = j( k ) \star F( h ) \star j( g )$ by~\cref{Thm:ConjAction}.
          Since $f$ and $( g, h, k )$ were arbitrary, then $F$ is a conjugated $M$-control structure.
          \endproof
    \end{enumerate}
\end{proof}

\subsection{Intrinsic Dagger Structures and Dagger Control Functors}
Another important property of the standard control functor is that it also sends unitary maps to unitary maps.
That is to say, if $f \in \mathcal{V}( X, X )$ and $f \circ f^\dagger = 1_X = f^\dagger \circ f$, then $C( f ) \circ C( f )^\dagger = 1_X = C( f )^\dagger \circ C( f )$.
This follows from the more general fact that the standard control functor is a dagger-functor.
It turns out that if $M$ is the comonoid of a $\dagger$-Frobenius algebra, then $\Lambda_M( \mathcal{V} )$ also inherits a dagger structure $( \ddagger )$ defined in terms of $( \dagger )$ as shown in \cref{Thm:FrobDag}.
Moreover, if a functor $C: \mathcal{U} \to \mathcal{U}$ is compatible with the comonoid of a $\dagger$-Frobenius algebra $M$, then $C$ is a dagger-functor if and only if its factorization through $\Lambda_M( \mathcal{V} )$ is a dagger-functor.
This provides a characterization of $\dagger$-control functors in terms of their factorizations (see \cref{Cor:ConjControl}).

\begin{nota}
    In this subsection, we further assume that $\mathcal{V}$ is symmetric with monoidal dagger structure $\dagger: \mathcal{V} \to \mathcal{V}^{\text{op}}$, $\mathcal{U}$ is a dagger-subcategory of $\mathcal{V}$, and $M$ is cocommutative.
\end{nota}

\begin{defi}
    The \emph{$M$-intrinsic dagger assignment} is the identity-on-objects contravariant assignment $\ddagger: \Lambda_M( \mathcal{V} ) \to \Lambda_M( \mathcal{V} )^{\text{op}}$ that maps each $f \in \Lambda_M( \mathcal{V} )( X, Y )$ to the unique morphism in $\Lambda_M( X, Y )$ for which the following equation of diagrams holds in $\mathcal{V}$.
    \begin{center}
        \includegraphics[valign=c,scale=0.84]{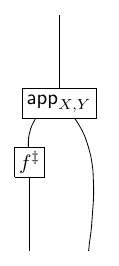}
        =
        \includegraphics[valign=c,scale=0.84]{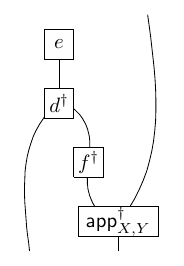}
    \end{center}
\end{defi}

\begin{lem}
    \label{Lem:DagId}
    If $f \in \mathcal{V}( X, Y )$, then $j( f )^\ddagger = j( ( f )^\dagger )$.
\end{lem}

\begin{proof}
    Since $j( f ) = [ f ] \circ e$ and $( \dagger )$ is functorial, then $j( f )^\dagger = e^\dagger \circ [ f ]^\dagger$.
    Since $( \dagger )$ is monoidal, then also $\app_{Y,X} \circ ( [f^\dagger] \otimes 1_X ) = ( [f]^\dagger \otimes 1_X ) \circ \app_{X,Y}^\dagger$.
    Then the following equation holds.
    \begin{center}
        \includegraphics[valign=c,scale=0.84]{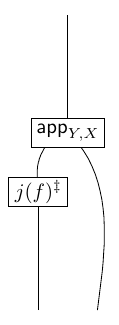}
        =
        \includegraphics[valign=c,scale=0.84]{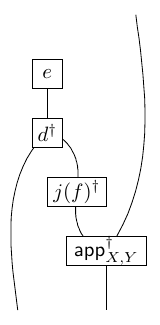}
        =
        \includegraphics[valign=c,scale=0.84]{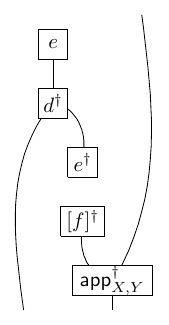}
        =
        \includegraphics[valign=c,scale=0.84]{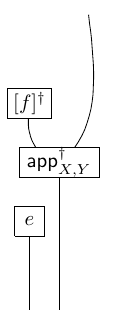}
        =
        \includegraphics[valign=c,scale=0.84]{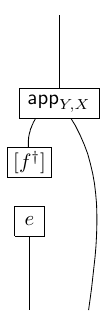}
        =
        \includegraphics[valign=c,scale=0.84]{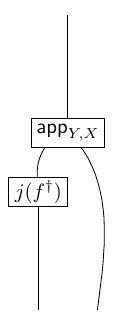}
    \end{center}
    Then by uniqueness, $j( f )^\ddagger = j( f^\dagger )$.
\end{proof}

\begin{lem}
    \label{Lem:DagInv}
    The $M$-intrinsic dagger assignment defines an involution if and only if $P$ is self-adjoint with unit $d \circ e^\dagger$ and counit $e \circ d^\dagger$.
\end{lem}

\begin{proof}
    There are two cases to consider.
    \begin{enumerate}[align=left]
    \item[$\Rightarrow$]
          Assume that $( \ddagger )$ is involutive.
          Let $f: P \to [ \mathbb{C}, P ]$ denote the image of $\rho_P: P \otimes \mathbb{C} \to P$ under the hom-tensor adjunction.
          Since $( \ddagger )$ is involutive, $f = f^{\ddagger\ddagger}$.
          Since $( \dagger )$ is monoidal, then the following equation holds by the definition of $( \ddagger )$ and its vertical reflection.
          \begin{equation*}
            \includegraphics[valign=c,scale=0.84]{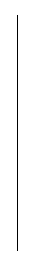}
            =
            \includegraphics[valign=c,scale=0.84]{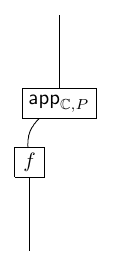}
            =
            \includegraphics[valign=c,scale=0.84]{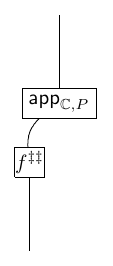}
            =
            \includegraphics[valign=c,scale=0.84]{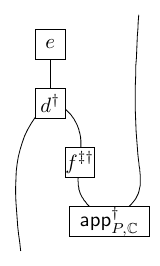}
            =
            \includegraphics[valign=c,scale=0.84]{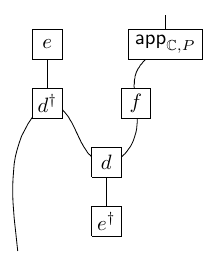}
            =
            \includegraphics[valign=c,scale=0.84]{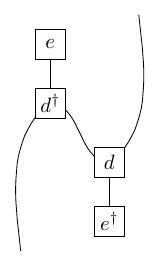}
          \end{equation*}
          Since $( \dagger )$ is monoidal, then the vertical reflection of this equation also holds.
          In conclusion, $P$ is self-adjoint with unit $d \circ e^\dagger$ and counit $e \circ d^\dagger$.
    \item[$\Leftarrow$]
          Assume that $P$ is self-adjoint with unit $d \circ e^\dagger$ and counit $e \circ d^\dagger$.
          Let $f \in \Lambda_M( \mathcal{V} )( X, Y )$.
          Since $( \dagger )$ is monoidal, then the following equation holds by the definition of $( \ddagger )$ and its vertical reflection.
          \begin{center}
            \includegraphics[valign=c,scale=0.84]{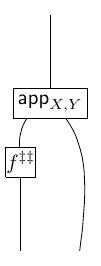}
            =
            \includegraphics[valign=c,scale=0.84]{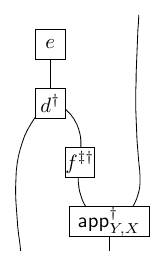}
            =
            \includegraphics[valign=c,scale=0.84]{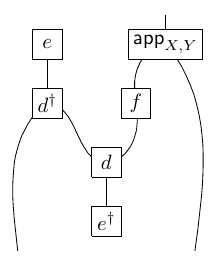}
            =
            \includegraphics[valign=c,scale=0.84]{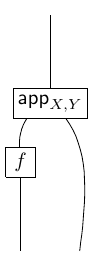}
          \end{center}
          Then by uniqueness, $f = f^{\ddagger\ddagger}$.
          Since $f$ was arbitrary, then $( \ddagger )$ is an involution.
          \endproof
    \end{enumerate}    
\end{proof}

\begin{lem}
    \label{Lem:SimpDag}
    If $f: P \to [ X, Y ]$ and $g: P \to [ Y, Z ]$, then the following equation hold.
    \begin{align*}
        \includegraphics[valign=c,scale=0.84]{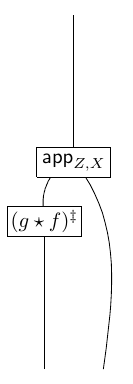}
        &=
        \includegraphics[valign=c,scale=0.84]{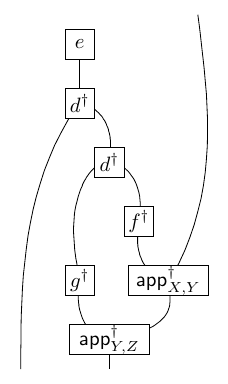}
        &
        \includegraphics[valign=c,scale=0.84]{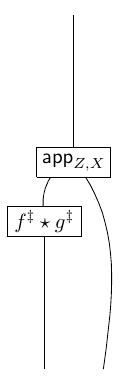}
        &=
        \includegraphics[valign=c,scale=0.84]{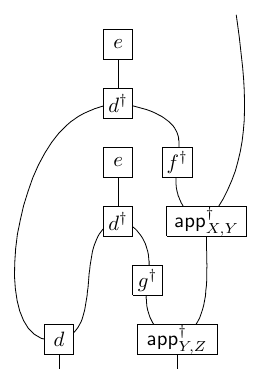}
    \end{align*}
\end{lem}

\begin{proof}
    Since $( \dagger )$ is a control functor, then $( g \star f )^\dagger$ can be evaluated by a vertical reflection.
    Then the following equations holds by the definitions of $( \ddagger )$ and $( \star )$.
    \begin{align*}
        \includegraphics[valign=c,scale=0.84]{figs/special/ddag_1_orig.pdf}
        =
        \includegraphics[valign=c,scale=0.84]{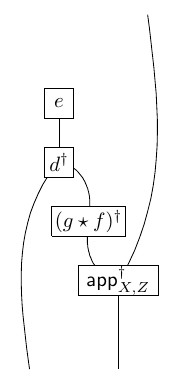}
        =
        \includegraphics[valign=c,scale=0.84]{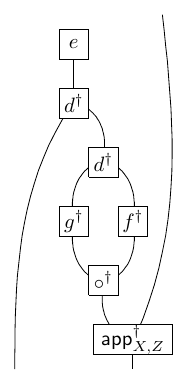}
        =
        \includegraphics[valign=c,scale=0.84]{figs/special/ddag_1_simp.pdf}
    \end{align*}
    \begin{align*}
        \includegraphics[valign=c,scale=0.84]{figs/special/ddag_2_orig.pdf}
        =
        \includegraphics[valign=c,scale=0.84]{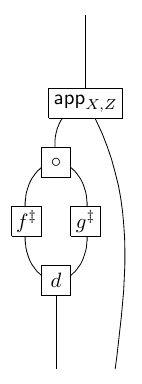}
        =
        \includegraphics[valign=c,scale=0.84]{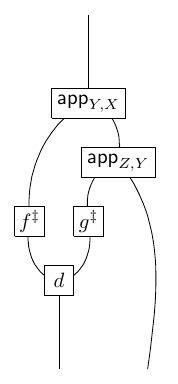}
        =
        \includegraphics[valign=c,scale=0.84]{figs/special/ddag_2_simp.pdf}
    \end{align*}
    This concludes the proof.
\end{proof}

\begin{thm}
    \label{Thm:FrobDag}
    The $M$-intrinsic dagger assignment is a dagger structure on $\Lambda_M( \mathcal{V} )$ if and only if $M$ is the comonoid of a $\dagger$-Frobenius algebra.
\end{thm}

\begin{proof}
    There are two cases to consider.
    \begin{enumerate}[align=left]
    \item[$\Rightarrow$]
          Assume that $( \ddagger )$ defines a dagger structure on $\Lambda_M( \mathcal{V} )$.
          Since $( \ddagger )$ is a dagger structure, then by \cref{Lem:DagInv} $P$ is self-adjoint with unit $d \circ e^\dagger$ and counit $e \circ d^\dagger$.
          Then by~\cref{Prop:NonDegenForm}, there exists a unique comonoid $( P, \delta, e )$ together with which $( P, d^\dagger, e^\dagger )$ forms a Frobenius algebra.
          Since $d \circ e^\dagger$ is the unit to $e \circ d^\dagger$, then $\delta$ is defined by the following equation.
          \begin{equation}
            \includegraphics[valign=c,scale=0.84]{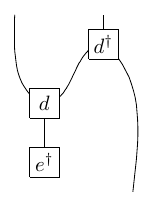}
            =
            \includegraphics[valign=c,scale=0.84]{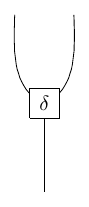}
            =
            \includegraphics[valign=c,scale=0.84]{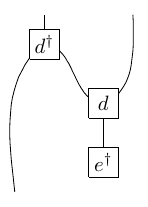}
            \label{Eqn:DeltaDef}
          \end{equation}
          Next, let $f: P \to [ P, P ]$ denote the adjoint to $d^\dagger$ and $g: P \to [ P, P ]$ denote the adjoint to $1_P \otimes e$.
          Then by \cref{Lem:SimpDag}, then following pair of equations hold.
          \begin{align*}
            \includegraphics[valign=c,scale=0.84]{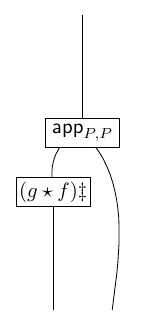}
            &=
            \includegraphics[valign=c,scale=0.84]{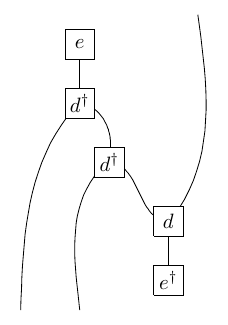}
            &
            \includegraphics[valign=c,scale=0.84]{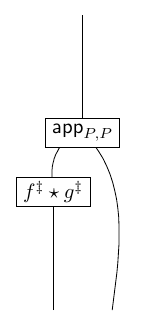}
            &=
            \includegraphics[valign=c,scale=0.84]{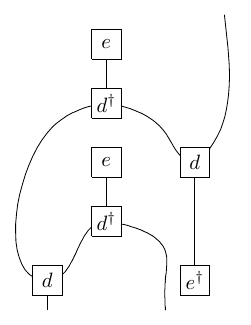}
          \end{align*}
          Since $( \ddagger )$ is a dagger structure, then $( g \star f )^\ddagger = f^\ddagger \star g^\ddagger$ and the following equation holds.
          \begin{equation*}
            \includegraphics[valign=c,scale=0.84]{figs/special/ddag_p1_1.pdf}
            =
            \includegraphics[valign=c,scale=0.84]{figs/special/ddag_p1_2.pdf}
            =
            \includegraphics[valign=c,scale=0.84]{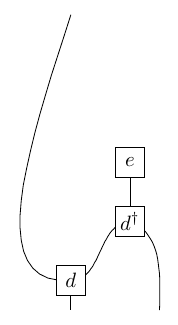}
            =
            \includegraphics[valign=c,scale=0.84]{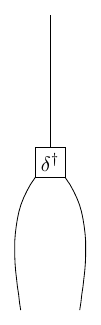}
          \end{equation*}
          Combining this with \cref{Eqn:DeltaDef}, the following equation holds.
          \begin{equation*}
            \includegraphics[valign=c,scale=0.84]{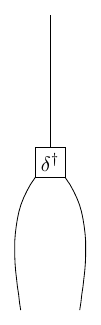}
            =
            \includegraphics[valign=c,scale=0.84]{figs/special/ddag_p1_1.pdf}
            =
            \includegraphics[valign=c,scale=0.84]{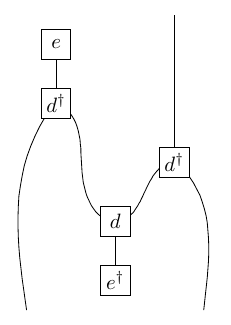}
            =
            \includegraphics[valign=c,scale=0.84]{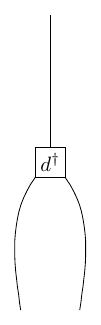}
          \end{equation*}
          Then $\delta^\dagger = d^\dagger$, which implies $\delta = d$.
          Hence $M$ is the comonoid of a $\dagger$-Frobenius algebra.
    \item [$\Leftarrow$]
          Assume that $M$ is the comonoid of a $\dagger$-Frobenius algebra.
          First, it must be shown that $( \ddagger )$ is a functor.
          Let $X \in \mathcal{V}_0$.
          Then by \cref{Lem:DagId}, $j( 1_X ) = j( 1_X^\dagger ) = j( 1_X )^\ddagger$.
          Since $X$ was arbitrary and $j( 1_X ) = i_X$, then $( \ddagger )$ preserves identities.
          It remains to be shown that $( \ddagger )$ reverses composition.
          Let $f \in \Lambda_M( \mathcal{V} )( X, Y )$ and $g \in \Lambda_M( \mathcal{V} )( Y, Z )$.
          Since $( P, d, e )$ is the comonoid of a $\dagger$-Frobenius algebra, then the following equation holds.
          \begin{equation*}
            \includegraphics[valign=c,scale=0.84]{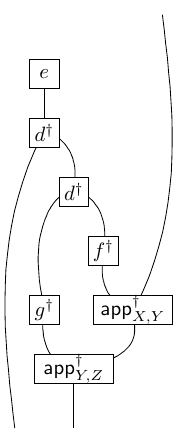}
            =
            \includegraphics[valign=c,scale=0.84]{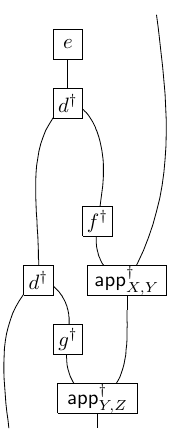}
            =
            \includegraphics[valign=c,scale=0.84]{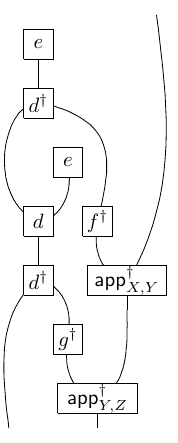}
            =
            \includegraphics[valign=c,scale=0.84]{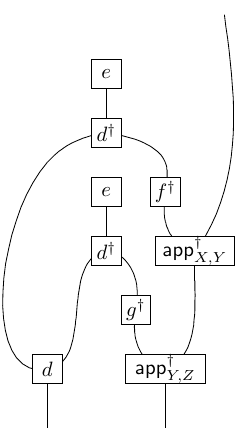}
          \end{equation*}
          Then by \cref{Lem:SimpDag}, the following equation holds.
          \begin{equation*}
            \includegraphics[valign=c,scale=0.84]{figs/special/ddag_1_orig.pdf}
            =
            \includegraphics[valign=c,scale=0.84]{figs/special/ddag_1_simp.pdf}
            =
            \includegraphics[valign=c,scale=0.84]{figs/special/ddag_2_simp.pdf}
            =
            \includegraphics[valign=c,scale=0.84]{figs/special/ddag_2_orig.pdf}
          \end{equation*}
          Then by uniqueness, $( g \star f )^\ddagger = f^\ddagger \star g^\ddagger$.
          Since $f$ and $g$ were arbitrary, then $( \ddagger )$ preserves composition.
          Then $( \ddagger )$ is a functor.
          Moreover, since $M$ is a Frobenius algebra, then as in~\cref{Fig:Frob}, $P$ is self-adjoint with unit $d \circ e^\dagger$ and counit $e \circ d^\dagger$.
          Then $( \ddagger )$ is an involution by \cref{Lem:DagInv}.
          In conclusion, $( \ddagger )$ is a dagger structure.
          \endproof
    \end{enumerate}
\end{proof}

\begin{lem}
    \label{Lem:XiDag}
    If $M$ is the comonoid of a commutative $\dagger$-Frobenius algebra, then $\Xi_M$ is a dagger functor with respect to the $M$-intrinsic dagger assignment.
\end{lem}

\begin{proof}
    Assume that $M$ is a $\dagger$-Frobenius algebra.
    Then the following equation holds.
    \begin{equation*}
        \includegraphics[valign=c,scale=0.84]{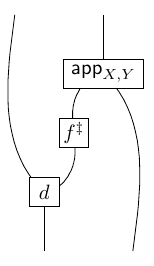}
        =
        \includegraphics[valign=c,scale=0.84]{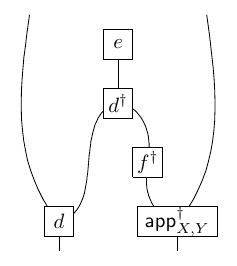}
        =
        \includegraphics[valign=c,scale=0.84]{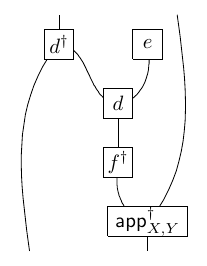}
        =
        \includegraphics[valign=c,scale=0.84]{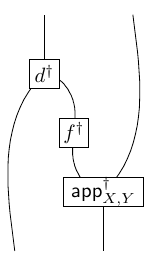}
    \end{equation*}
    Since $( \dagger )$ is a monoidal contravariant functor, then $\Xi( f^\ddagger ) = \Xi( f )^\dagger$.
    Since $f$ was arbitrary, then $\Xi_M$ is a dagger-functor.
\end{proof}

\begin{thm}
    \label{Thm:DaggerFactor}
    Assume $M$ is the comonoid from a commutative $\dagger$-Frobenius algebra.
    The functor $F$ is a dagger-functor with respect to the $M$-intrinsic dagger assignment if and only if the functor $\Xi_M \circ F$ is a dagger-functor.
\end{thm}

\begin{proof}
    Let $C = \Xi_M \circ F$.
    There are two cases to consider.
    \begin{enumerate}[align=left]
    \item[$\Rightarrow$]
          Assume that $F$ is a dagger-functor.
          By \cref{Lem:XiDag}, $\Xi_M$ is a dagger-functor as well.
          Since dagger-functors are closed under composition, then $C$ is a dagger-functor.
    \item[$\Leftarrow$]
          Assume that $C$ is a dagger-functor.
          Let $f \in \mathcal{U}( X, X )$.
          Then $C( f^\dagger ) = C( f )^\dagger$.
          Since $( \dagger )$ is monoidal, then the following equation holds by a vertical reflection of $C( f )$.
          \begin{equation*}
            \includegraphics[valign=c,scale=0.84]{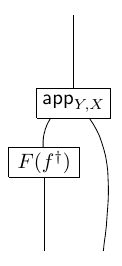}
            =
            \includegraphics[valign=c,scale=0.84]{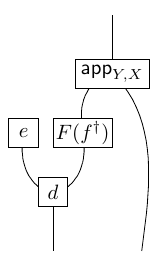}
            =
            \includegraphics[valign=c,scale=0.84]{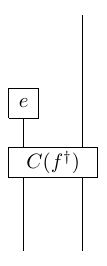}
            =
            \includegraphics[valign=c,scale=0.84]{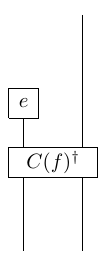}
            =
            \includegraphics[valign=c,scale=0.84]{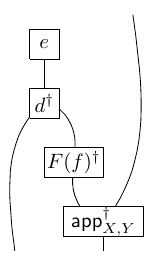}
            =
            \includegraphics[valign=c,scale=0.84]{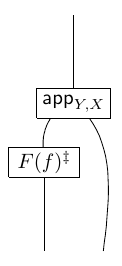}
          \end{equation*}
          Then $F( f^\dagger ) = F( f )^\ddagger$ by uniqueness.
          Since $f$ was arbitrary, $F$ is a dagger-functor.
          \endproof
    \end{enumerate}
\end{proof}

\begin{cor}
    \label{Cor:ConjControl}
    Assume that $\mathcal{U}$ is a symmetric monoidal dagger-subcategory and $M$ is the comonoid from a commutative $\dagger$-Frobenius algebra.
    If $F$ is an $M$-control structure, the following statements are equivalent.
    \begin{enumerate}
    \item $F: \mathcal{U} \to \Lambda_M( \mathcal{V} )$ is a dagger-functor with respect to the $M$-intrinsic dagger structure.
    \item $\Xi_M \circ F: \mathcal{U} \to \mathcal{U}$ is a $\dagger$-control functor.
    \end{enumerate}
\end{cor}

\begin{proof}
    Let $C = \Xi_M \circ F$.
    There are two cases to consider.
    \begin{enumerate}[align=left]
    \item[$\Rightarrow$]
          Assume that $F$ is a dagger-functor with respect to $( \ddagger )$.
          Since $F$ is an $M$-control structure, then $C$ is a control functor by \cref{Thm:ControlStructs}.
          Since $F$ is a dagger-functor with respect to $( \ddagger )$, then $C$ is a $\dagger$-control functor by \cref{Thm:DaggerFactor}.
    \item[$\Leftarrow$]
          Assume that $C$ is a $\dagger$-control functor.
          Then $F$ is a dagger-functor with respect to $( \ddagger )$ by~\cref{Thm:DaggerFactor}.
          This concludes the first part of the proof.
          \endproof
    \end{enumerate}
\end{proof}

\begin{exa}[Dagger Control Functors]
    The standard control functor from~\cref{Ex:StdCtrlStruct} and the trivial control functor from~\cref{Ex:TrivCtrl} are both $\dagger$-control functors.
\end{exa}

We conclude our discussion of intrinsic dagger assignments with a remark on conjugation.
As noted in~\cref{Thm:ConjStruct}, the conjugation axiom can be restated in terms of $\Lambda_M( \mathcal{V} )$ using conjugated $M$-control structures.
This can be further refined when $M$ is the comonoid from a commutative $\dagger$-Frobenius algebra.
First notice that when working with control functors and dagger-categories, the conjugation axiom is often invoked in the special case where $h = f^\dagger$, so that $f^\dagger \circ f = 1_X$.
In this case, the conjugation axiom enforces the condition that $C( f^\dagger \circ g \circ f ) = ( 1_P \otimes f^\dagger ) \circ C( g ) \circ ( 1_P \otimes f )$.
It follows from~\cref{Thm:ConjStruct} that this is equivalent to asking $F( f^\dagger \circ g \circ f ) = j( f^\dagger ) \star F( g ) \star j( f )$.
However, if $M$ is the comonoid of a commutative $\dagger$-Frobenius algebra, this equation can be rewritten in terms of the $M$-intrinsic dagger assignment.

\begin{cor}
    \label{Cor:DagConj}
    Assume that $\mathcal{U}$ is a symmetric monoidal subcategory, $M$ is the comonoid from a commutative $\dagger$-Frobenius algebra, $F$ is an $M$-control structure, and $C = \Xi_M \circ F$.
    For each $f \in \mathcal{V}( X, Y )$ and $g \in \mathcal{U}( Y, Y )$, $C( f^\dagger \circ g \circ f ) = ( 1_P \otimes f^\dagger ) \circ C( g ) \circ ( 1_P \otimes f )$ if and only if $F( g^\dagger \circ f \circ g ) = j( g )^\ddagger \star F( f ) \star j( g )$.
\end{cor}

\begin{proof}
    Let $f \in \mathcal{V}( X, Y )$ and $g \in \mathcal{U}( Y, Y )$.
    There are two cases to consider.
    \begin{enumerate}[align=left]
    \item[$\Rightarrow$]
          Assume that $C( f^\dagger \circ g \circ f ) = ( 1_P \otimes f^\dagger ) \circ C( g ) \circ ( 1_P \otimes f )$.
          Then $F( f^\dagger \circ g \circ f ) = j( f^\dagger ) \star F( g ) \star j( f )$ by \cref{Thm:ConjAction}.
          Then $F( f^\dagger \circ g \circ f ) = j( f )^\ddagger \star F( g ) \star j( f )$ by \cref{Lem:DagId}.
    \item[$\Leftarrow$]
          Assume that $F( f^\dagger \circ g \circ f ) = j( f )^\ddagger \star F( g ) \star j( f )$.
          Then $F( f^\dagger \circ g \circ f ) = j( f^\dagger ) \star F( g ) \star j( f )$ by \cref{Lem:DagId}.
          Then $C( f^\dagger \circ g \circ f ) = ( 1_P \otimes f^\dagger ) \circ C( g ) \circ ( 1_P \otimes f )$ by \cref{Thm:ConjAction}.
          \endproof
    \end{enumerate}
\end{proof}

\begin{exa}[Ancilla Initialization]
    In quantum computing, ancilla initialization corresponds to the linear map $f = \ket{0} \otimes 1_V$.
    The adjoint to $f$ is the projection $f^\dagger: \bra{0} \otimes 1_V$ and is referred to as ancilla termination.
    Ancilla initialization and termination are often used in quantum circuit synthesis.
    Quantum circuit synthesis asks to find a circuit $C$ such that $\sem{C} = g$ for a given unitary $g: V \to V$.
    It is sometimes more efficient to find a circuit $C$ such that $\sem{C} = h$ and $g = f^\dagger \circ h \circ f$, as in~\cite{Barenco_1995}.
    This is a $( f, h, f^\dagger )$ factorization of $g$.
    Since the standard control functor is conjugated~\cite{Delorme2026}, then this decomposition extends to a decomposition $C( g ) = f^\dagger \circ C( h ) \circ f$.
    If $F$ is the control structure associated with $C$, then we know from~\cref{Cor:DagConj} that $F( f^\dagger \circ h \circ f ) = j( f )^\ddagger \circ F( h ) \circ j( f )$ as well.
\end{exa}

%% file: circs/xij_lhs.tex
\begin{quantikz}[row sep=0.5em]
    & \gate[2]{\Xi_M( j( f ) )} & \\
    &                           &
\end{quantikz}

%% file: circs/xij_rhs.tex
\begin{quantikz}[row sep=0.5em]
    & \ghost{f} & \\
    & \gate{f}  &
\end{quantikz}

%% file: circs/xicomm_lhs.tex
\begin{quantikz}[row sep=0.5em]
    &              & \gate[2]{\Xi_M( g )} &              & \\
    & \gate{\,f\,} &                      & \gate{\,h\,} &
\end{quantikz}

%% file: circs/xicomm_rhs.tex
\begin{quantikz}[row sep=0.5em]
    & \gate[2]{\Xi_M( k )} & \\
    &                       &
\end{quantikz}

%% file: circs/xicopy_lhs.tex
\begin{quantikz}[row sep=0.5em,wire types={n,q}]
    & \gate{\,v\vphantom{f_v}\,} & \gate[2]{\Xi_M( f )}\setwiretype{q} & \\
    &                            &                                     &
\end{quantikz}

%% file: circs/xicopy_rhs.tex
\begin{quantikz}[row sep=0.5em,wire types={n,q}]
    & \gate{\,v\vphantom{f_v}\,} & \setwiretype{q} \\
    & \gate{f_v}                 &
\end{quantikz}

%% file: circs/xitensor_lhs.tex
\begin{quantikz}[row sep=0.5em]
    \ghost{H} & \gate[3]{\Xi_M( f \boxtimes g )} & \\
    \ghost{H} &                                  & \\
    \ghost{H} &                                  &
\end{quantikz}

%% file: circs/xitensor_rhs.tex
\begin{quantikz}[row sep=0.5em]
    \ghost{H} & \gate[2]{\Xi_M( f )} & & \gate[2]{\Xi_M( g )} & & \\
    \ghost{H} & & \permute{2,1} & & \unpermute{2,1} & \\
    \ghost{H} & & & & &
\end{quantikz}

%% file: circs/comptensor_1.tex
\scalebox{0.73}{\begin{quantikz}[row sep=0.5em]
    \ghost{H} & \gate[3]{F( f \otimes g )} & \\
    \ghost{H} &                              & \\
    \ghost{H} &                              &
\end{quantikz}}

%% file: circs/comptensor_2.tex
\scalebox{0.73}{\begin{quantikz}[row sep=0.5em]
    \ghost{H} & \gate[2]{F( f )} & & \gate[2]{\Xi_M ( j( g ) )} & & \\
    \ghost{H} & & \permute{2,1} & & \unpermute{2,1} & \\
    \ghost{H} & & & & &
\end{quantikz}}

%% file: circs/comptensor_3.tex
\scalebox{0.73}{\begin{quantikz}[row sep=0.5em]
    \ghost{H} & \gate[2]{F( f )} & & & & \\
    \ghost{H} & & \permute{2,1} & \gate{\,g\,} & \unpermute{2,1} & \\
    \ghost{H} & & & & &
\end{quantikz}}

%% file: circs/comptensor_4.tex
\scalebox{0.73}{\begin{quantikz}[row sep=0.5em]
    \ghost{H} & \gate[2]{F( f )} & & & & \\
    \ghost{H} & & \permute{2,1} & & \unpermute{2,1} & \\
    \ghost{H} & \gate{\,g\,} & & & &
\end{quantikz}}

%% file: circs/comptensor_5.tex
\scalebox{0.73}{\begin{quantikz}[row sep=0.5em]
    \ghost{H} & \gate[2]{F( f )} & \\
    \ghost{H} & & \\
    \ghost{H} & \gate{\,g\,} &
\end{quantikz}}

%% file: circs/ctrl_lhs.tex
\scalebox{0.73}{\begin{quantikz}[row sep=0.5em]
    \ghost{H} & \permute{2,1} & \gate[3]{C( C( f ) )} & \\
    \ghost{H} &               &                       & \\
    \ghost{H} &               &                       &
\end{quantikz}}

%% file: circs/ctrl_rhs.tex
\scalebox{0.73}{\begin{quantikz}[row sep=0.5em]
    \ghost{H} & \gate[3]{C( C( f ) )} & \permute{2,1} & \\
    \ghost{H} &                       &               & \\
    \ghost{H} &                       &               &
\end{quantikz}}

%% file: circs/compatible_u_1.tex
\scalebox{0.73}{\begin{quantikz}[row sep=0.5em,wire types={n,n,q}]
    \ghost{H} &          & \gate[2]{d}     & \setwiretype{q} & \\
    \ghost{H} & \gate{u} & \setwiretype{q} & \gate[2]{C(f)}  & \\
    \ghost{H} &          &                 &                 &
\end{quantikz}}

%% file: circs/compatible_u_2.tex
\scalebox{0.73}{\begin{quantikz}[row sep=0.5em,wire types={n,n,q}]
    \ghost{H} & \gate{u} & \setwiretype{q}                & \\
    \ghost{H} & \gate{u} & \gate[2]{C(f)} \setwiretype{q} & \\
    \ghost{H} &          &                                &
\end{quantikz}}

%% file: circs/compatible_u_3.tex
\scalebox{0.73}{\begin{quantikz}[row sep=0.5em,wire types={n,n,q}]
    \ghost{H} & \gate{u} & \setwiretype{q} \\
    \ghost{H} & \gate{u} & \setwiretype{q} \\
    \ghost{H} &          & \setwiretype{q}
\end{quantikz}}

%% file: circs/compatible_u_4.tex
\scalebox{0.73}{\begin{quantikz}[row sep=0.5em,wire types={n,n,q}]
    \ghost{H} &          & \gate[2]{d}     & \setwiretype{q} \\
    \ghost{H} & \gate{u} & \setwiretype{q} & \\
    \ghost{H} &          &                 &
\end{quantikz}}

%% file: circs/compatible_u_5.tex
\scalebox{0.73}{\begin{quantikz}[row sep=0.5em,wire types={n,n,q}]
    \ghost{H} &          &                                & \gate[2]{d} & \setwiretype{q} \\
    \ghost{H} & \gate{u} & \gate[2]{C(f)} \setwiretype{q} &             & \\
    \ghost{H} &          &                                &             &
\end{quantikz}}

%% file: circs/compatible_v_1.tex
\scalebox{0.73}{\begin{quantikz}[row sep=0.5em,wire types={n,n,q}]
    \ghost{H} &          & \gate[2]{d}     & \setwiretype{q} & \\
    \ghost{H} & \gate{v} & \setwiretype{q} & \gate[2]{C(f)}  & \\
    \ghost{H} &          &                 &                 &
\end{quantikz}}

%% file: circs/compatible_v_2.tex
\scalebox{0.73}{\begin{quantikz}[row sep=0.5em,wire types={n,n,q}]
    \ghost{H} & \gate{v} & \setwiretype{q}                & \\
    \ghost{H} & \gate{v} & \gate[2]{C(f)} \setwiretype{q} & \\
    \ghost{H} &          &                                &
\end{quantikz}}

%% file: circs/compatible_v_3.tex
\scalebox{0.73}{\begin{quantikz}[row sep=0.5em,wire types={n,n,q}]
    \ghost{H} & \gate{v} & \setwiretype{q} \\
    \ghost{H} & \gate{v} & \setwiretype{q} \\
    \ghost{H} & \gate{f} & \setwiretype{q}
\end{quantikz}}

%% file: circs/compatible_v_4.tex
\scalebox{0.73}{\begin{quantikz}[row sep=0.5em,wire types={n,n,q}]
    \ghost{H} &          & \gate[2]{d}     & \setwiretype{q} \\
    \ghost{H} & \gate{v} & \setwiretype{q} & \\
    \ghost{H} & \gate{f} &                 &
\end{quantikz}}

%% file: circs/compatible_v_5.tex
\scalebox{0.73}{\begin{quantikz}[row sep=0.5em,wire types={n,n,q}]
    \ghost{H} &          &                                & \gate[2]{d} & \setwiretype{q} \\
    \ghost{H} & \gate{v} & \gate[2]{C(f)} \setwiretype{q} &             & \\
    \ghost{H} &          &                                &             &
\end{quantikz}}

%% file: sections/conclusion.tex
\section{Conclusion and Future Work}

In this paper, we introduced a new approach to parameterized quantum circuits based on enriched category theory.
This approach allows for parameterized maps with additional structure, such as parameterized maps which are continuous in their parameters.
It was shown how many seemingly unrelated constructions from quantum computing, such as control functors and shared entanglement, can be seen as instances of parameterized matrices which are linear in their parameters.
To better understand parameterized quantum circuits, the case of parameterization over a Cartesian monoidal category was studied, and many standard circuit transformations were recovered from first principles.
To better understand quantum control, the case of parameterization over a self-enriched category was studied, yielding new theorems relating control functors and Frobenius algebras in this special case.
There are several directions for future work, such as characterizing control functors in the case of compact closed categories, accounting for special Frobenius algebras, and developing a linear dependent type theory for parameterized quantum circuits.
It may also prove fruitful to explore applications to quantum communication and temporal verification.
The authors would also like to explore connections to $\mathcal{V}$-graded categories as defined in~\cite{LucyshynWright2026} and to generalize $M$-control structures to braided monoidal categories using the more general definitions found in~\cite{KongZheng2018}.